\documentclass[12pt]{article} 

\usepackage{verbatim}
\newcommand{\URL}[1]{\url{\detokenize{#1}}}

\usepackage[numbers]{natbib}
\usepackage[utf8]{inputenc} 
\usepackage{amsmath,amssymb}
\usepackage{amsthm}
\theoremstyle{definition}
\usepackage[usenames, dvipsnames]{color}
\usepackage{url}
\usepackage{thmtools,thm-restate}
\usepackage[export]{adjustbox}
\usepackage{caption}
\usepackage{setspace}

 \newcommand{\eps}{\varepsilon}
\newtheorem{theorem}{Theorem}[section]
\newtheorem{lemma}[theorem]{Lemma}
\newtheorem{proposition}[theorem]{Proposition}

\theoremstyle{remark}
\newtheorem{rem}[theorem]{Remark}

\definecolor{Red}{rgb}{0,0,0}

\definecolor{DR}{rgb}{0,0,0}

\definecolor{Blue}{rgb}{0,0,0}

\definecolor{Green}{rgb}{0,0,0}

\definecolor{Grey}{rgb}{0,0,0}

\definecolor{Magenta}{rgb}{0,0,0}

\definecolor{Dgreen}{rgb}{0,0,0}

\usepackage{geometry} 
\usepackage{graphicx} 

\usepackage{booktabs} 
\usepackage{array} 
\usepackage{paralist} 
\usepackage{verbatim} 
\usepackage{subfig} 

\usepackage{fancyhdr} 
\usepackage{sectsty}
\allsectionsfont{\sffamily\mdseries\upshape} 

\usepackage[nottoc,notlof,notlot]{tocbibind} 
\usepackage[titles,subfigure]{tocloft} 

\title{{Optimal Placement Of A Small Order In A Diffusive Limit Order Book}}

\author{
Jos\'{e} E. Figueroa-L\'{o}pez
\thanks{Department of Mathematics, Washington University in St. Louis, St. Louis, MO, 63130, USA.\hspace{1 cm}
{\tt figueroa@math.wustl.edu}. Research supported in part by the NSF Grants: DMS-1561141 and DMS-1613016.}
\and
Hyoeun Lee
\thanks{Department of Statistics, Purdue University, West Lafayette, IN, 47907, USA.
{\tt lee1487@purdue.edu}.}
\and
Raghu Pasupathy
\thanks{Department of Statistics, Purdue University, West Lafayette, IN, 47907, USA.
{\tt pasupath@purdue.edu}.}
}

\begin{document}

\maketitle

\abstract{
We study the optimal placement problem of a stock trader who wishes to clear his/her inventory by a predetermined time horizon $t$, by using a limit order or a market order. For a diffusive market, we characterize the optimal limit order placement policy and analyze its behavior under different market conditions. In particular, we show that, in the presence of a negative drift, there exists a critical time $t_{0}>0$ such that, for any time horizon $t>t_{0}$, there exists an optimal placement, which, contrary to earlier work, is different from one that is placed ``infinitesimally" close to the best ask, such as the best bid and second best bid. We also propose a simple method to approximate the critical time $t_{0}$ and the optimal order placement.

%

\vspace{0.2 cm}

}

\section{Introduction}

\par In today's stock market, most stock exchanges have adopted electronic trading systems, {where} buyers and sellers can trade securities, foreign exchange or financial derivatives electronically. This has led to the development of algorithmic {trading, which} relies on trading strategies based on pre-programmed trading instructions. More generally, high frequency trading (HFT) is a recent trend with a focus on short time scales. {Estimates} of HFT exceeds 50\% of {the} U.S.-listed equities trading volume
\footnote{\texttt{https://www.sec.gov/marketstructure/research/hft\_lit\_review\_march\_2014.pdf}}.  

\par Typically, one of the first problems of stock traders is to split a large order into smaller orders to reduce market impact, that is, the {adverse} effect that an order may have {against the security's price, as} buying (respectively, selling) an asset tends to move the price upward (respectively, downward). Secondly, they need to place those small orders into multiple time intervals. When placing the orders, they also need to decide whether to use a market or a limit order, and, in the second case, which price level to put the order. Limit order is an order to trade an asset at a specified price. The buyer/seller can specify the price but the execution of the limit order is not guaranteed. By contrast, a market order is an order to buy or sell the asset at the best available price. The execution of the market order occurs immediately. The different methods used to solve these problems are broadly called optimal {execution/placement} strategies.

\par A limit order book (LOB) collects all the limit orders, including the quantities and the price. The LOB is updated as market orders are executed, or more limit orders are submitted, or existing limit orders are cancelled. In a traditional optimal execution problem, we are interested in deciding whether {(and when) to place a market order or a limit order, but only at the best bid or ask price (see \cite{jacquierLiu2017}, \cite{cartea2014modelling}, \cite{alfonsi2010}, \cite{GuilbaudPham2013}, \cite{Manglaras2015}, \cite{ContKukanov2013})}. 
However, a more recent stream of literature have also considered the problem of deciding whether placing the limit order deeper in the book could be preferable. This question of determining the optimal price level is often called the optimal placement problem of a limit order. 

In  \cite{guo2013optimal}, an optimal placement problem is studied under a discrete-time model for the level I prices of a LOB. Specifically, {Guo et.~al.}~\cite{guo2013optimal} investigates the optimal placement problem when the investor wants to buy one share of the asset by a certain specified time horizon $t$,  assuming that the best ask price follows a symmetric correlated random walk (CRW) (see \cite{renshaw1981correlated} for the definition of CRW).  {Guo et.~al.}~\cite{guo2013optimal} assumes a ``static" trading strategy where the investor's limit order cannot be cancelled before $t$, and a limit order not executed by time $t$ is automatically cancelled and changed to a market order at time $t$. {It is also assumed therein} that there is a positive constant probability that the investor's order would be executed each time the order's price becomes the best bid price of the LOB. A key conclusion in \cite{guo2013optimal} is that the optimal strategy that minimizes the investor's expected cost is one of three possibilities: (i) placement at the best bid; (ii) placement at the second best bid; or (iii) initial market buy order. Furthermore, the answer changes according to the respective values of the rebate, the market fee, and the transition probability. 

We refer to {strategy of placing} the order at the best or second best bid level (as in (i) or (ii) above) as the ``Level I-II" or ``trivial" solution throughout the current work.  
 Such terminology is on account of two reasons. 
 First, placing the order at the best or second best bid does not incorporate any information about the state of the book at time $0$, which is typically available and {should ideally be} taken into account when placing a limit order. Second, by construction, a symmetric random walk as what is assumed in~\cite{guo2013optimal} lacks ``drift," even though for {mid} range time horizons, actual price processes sometimes exhibit moderate drift {(see \cite{cartea2014modelling})}. Our investigation in this paper addresses both these issues directly. Specifically, we characterize the nature of the optimal placement when the price dynamics deviates from a symmetric correlated random walk, while also incorporating information about the initial state of the LOB.

In this paper, we discuss the optimal strategy when the price dynamics follow a diffusive model such as a  Brownian motion (BM) or a Geometric Brownian motion (GBM). A BM model, often called the Bachelier model, can be seen as a reasonable approximation of asset price dynamics at intermediate intraday time horizons  (see, e.g., \cite{cont2013price} and \cite{chavez2014one}). Also, bridging with the work of \cite{guo2013optimal}, {a BM with $0$ drift (respectively, nonzero drift) appears as the limit of a symmetric (respectively, asymmetric) correlated random walk} when the time step between price changes and the tick size goes  to $0$ in a certain way (cf. \cite[Section 3]{renshaw1981correlated}, \cite{gruber2006diffusion}). However, GBM (also known as the Black-Scholes model) is generally believed to better fit asset price dynamics for longer time periods, in line with more traditional macro asset price models.

It is expected that, under the presence of negative drift, there exists an optimal placement policy different from the Level I-II solution of \cite{guo2013optimal}. Intuitively, if the drift of the stock is $\mu<0$, so that on average the best ask price is at the level $S_{0}+\mu t$ at the time horizon $t$, we expect that placing the order around such a level would be better than placing it at a level close to the best ask $S_{0}$. Such intuition is rigorously justifiable and we demonstrate the existence of a critical time $t_{0}>0$ such that a nontrivial optimal solution exists for any horizon $t>t_{0}$. Furthermore, we find that such a time horizon $t_0$ admits the parsimonious closed-form approximation $\rho(0^{+})(r+f)/2|\mu|$ for the BM model, and $\rho(0^{+})(r+f)/2S_{0}|\mu|$ for the GBM model. Here, $r$ and $f$ are respectively the investor's rebate and fee per executed limit and market order, respectively, and $\rho(0^{+})$ measures the probability that an order placed at the initial best bid would be executed {before the best ask queue gets depleted}. In general, the optimal solution will depend on the time horizon $t$, the drift $\mu$, the volatility $\sigma$, and a function {$\rho:(0,\infty)\times(0,\infty)\to (0,1]$, such that $\rho(x,t)$ is the probability that an order placed at level $S_{0}-x$ is executed during the first {time period} that this level becomes the best bid price and before the investment's time horizon $t$}. We can incorporate information about the initial state of the LOB through {$\rho(x,t)$}: the larger the initial queue size at level $x$ is, the smaller {$\rho(x,t)$} would be. We also analyze the behavior of the non-trivial optimal solution in different market regimes. Thus, for instance, under a long horizon or small volatility regime, the optimal placement solution takes the form $-\mu t\theta$, where $\theta>1$ is explicitly characterized.

\par The paper is organized as follows. In section 2, we introduce the optimal placement problem, together with the investor's expected cost function that we aim to minimize.  In section 3, we study the problem under the BM model and show the existence of the critical horizon time $t_{0}$, together with the asymptotic behaviour of the optimal placement strategy when $t\searrow{}t_{0}$ and when $t\nearrow\infty$. In Section 4, we carry on the same plan for a GBM model and, in addition, we also consider the behavior of the optimal placement strategy in a small volatility regime $\sigma \searrow{}0$. {Section \ref{NwSect} investigates the behavior of the probability $\rho(x,t)$ defined above  and assess the plausibility of the assumptions used in the paper, both theoretically and empirically.} Section 6 gives some conclusions . The proofs of {our main results and some further details} are given in Appendices.

\paragraph{General Notation.}
{The partial derivatives of a function $f(x,t)$ are denoted by $\partial_{x}f$, $\partial_{t}f$, {$\partial_{t}\partial_{x}f$}, $\partial^{2}_{x}f$, etc. The pdf, cdf, and survival or tail distribution of a standard normal r.v. $Z$ are denoted by $\phi(z)=e^{-z^{2}/2}/\sqrt{2\pi}$, $N(z)=\int_{-\infty}^{z}\phi(x)dx$, and ${N}^{c}(z)={\bar{N}(z)}=1-N(z)$, respectively.}

\section{Toward The Optimal Placement Problem In Continuous Time}\label{Sec2}

An investor wishes to buy one share of the stock by some predetermined time horizon $t>0$. Obviously he wishes to buy at the lowest possible price, for which he places a limit buy order at the price level $\bar{S}_{0}-x$, where $x>0$. Hereafter, $\bar{S}_{u}:=\bar{S}^{(\delta,\eps)}_{u}$ denotes the best ask price per share at time $u\geq{}0$ when the average time span between price changes is governed by a parameter $\delta>0$ and the tick size, which is assumed to coincide with the price increment at each price change, is $\eps>0$. {In particular, we are also assuming} that the spread between the best bid and ask is always one-tick $\eps$ apart (see \cite{cont2013price} for some empirical evidence strongly supporting this assumption). Hence, the investor's order can only be executed when the best ask price is at level $\bar{S}_{0}-x+\eps$ {or, equivalently, when the best bid price is at level $\bar{S}_{0}-x$}. We also denote $\tau_{0}:=0$ and $0<\tau_{1}<\tau_{2}<\dots$ the consecutive {times of changes in the best ask price}. In particular, {we are assuming that} $\delta=E(\tau_{i+1}-\tau_{i})$, for all $i\geq{}0$. Let us remark here that,  in a more general setting, $\delta$ could be set to be a different model parameter such as one that increases the speed of the price changes as $\delta\searrow{}0$.

We adopt the following trading strategy and resulting investor's cost:
\begin{enumerate}
	\item If $\bar{S}_{u}>\bar{S}_{0}-x+\eps$, for all $0\leq{}u\leq{}t$ (in particular, $x>\varepsilon$), the investor's limit order won't be fulfilled and he will cancel the order at time $t$ and buy the share at the market price $\bar{S}_{t}$. In that case, the investor's cost/gain can be set to be $\bar{S}_{t}-\bar{S}_{0}+f$, where $f\geq{}0$ is the fee that the market imposes per executed market order.
	\item Otherwise, suppose that the first time that the ask price is at the level $\bar{S}_{0}-x+\eps$, hereafter denoted $\tau$, happens before time $t$. Set $j:=\min\{i\geq{}0:\tau_{i}=\tau\}$. Then, there are three possibilities:
	\begin{enumerate}
		\item[(i)] The investor's order is executed {at a time $s\leq{}t\wedge{}\tau_{j+1}$}. In that case, the investor's cost/gain is set to be $-x-r$, where $r\geq{}0$ is the rebate per executed limited order.
		\item[(ii)] On the contrary, if the order is not executed before $t$ and $\tau_{j+1}\leq{}t$ (so that necessarily $\bar{S}_{\tau_{j+1}}=\bar{S}_{0}-x+2\eps$), the investor will cancel the order and buy the share at the market price $\bar{S}_{0}-x+2\eps$. In that case, his cost/gain will be $-x+f+2\eps$.
		\item[(iii)] If, again, the order is not executed before $t$ and $\tau_{j+1}>t$,
 so that the next price change happens after the time horizon $t$, then the investor will cancel the order at time $t$ and buy the share at the market price $\bar{S}_{t}=\bar{S}_{0}-x+\varepsilon$. In that case, his cost/gain will be $\bar{S}_{t}-\bar{S}_{0}+f$, which is the same as the case 1 above.
	\end{enumerate}
\end{enumerate}
{We aim to minimize the investor's expected cost, as introduced in the points 1-2 above. Let us first derive an explicit formula for it, for which we need to define the event $E_{t}$ that the best bid price reaches the level $\bar{S}_{0}-x$ {by} time $t$ and that, during the first time period that this happens, the order is executed {before time} $t$. In terms of {this event  $E_{t}$,} the running minimum $\bar{Y}_{t}:=\bar{Y}^{(\delta,\eps)}_t := \inf\limits_{0 \leq u \leq t}\bar{S}^{(\delta,\eps)}_u$, and $j:=\min\{i\geq{}0:{S_{\tau_{i}}}=\bar{S}_{0}-x+\varepsilon\}$,  the cost function can then be written  as follows for $x>0$ (see Appendix \ref{SuppLmsC} for its derivation):
\begin{align}\label{eq:CostContinuousApprox}\nonumber
\bar{C}_{\delta,\eps}(x,t)&=E\left[\left.\bar{S}_t-\bar{S}_0\right|\bar{Y}_t >\bar{S}_0 -x+\eps\right]P\left(\bar{Y}_t>\bar{S}_0-x+\eps\right)\\
&\quad +P(\bar{Y}_t \leq \bar{S}_0-x+\eps)\left[-x-{\rho(x,t)}(r+f)+2\varepsilon(1-\rho(x,t))\right]+f\\
\nonumber
&\quad+ P(\bar{Y}_t \leq \bar{S}_0-x+\eps,\tau_{j+1}>{}t,E_{t}^{c})\eps,
\end{align}
where $\rho(x,t):=\rho^{(\delta,\varepsilon)}(x,t)=P(E_t|\bar{Y}_t \leq S_0 - x + \epsilon)$. 
The cost function (\ref{eq:CostContinuousApprox}) is inspired by that of \cite{guo2013optimal}, but there is a significant difference in the treatment of the situation 2-(ii) described above since early cancellation is not allowed in \cite{guo2013optimal}.}
\begin{rem}\label{CmptRho}
As stated above, $\rho(x,t)$ is the probability that {a bid limit} order placed at level $\bar{S}_{0}-x$ at time $0$ is executed {before time $t$} and during the first time period when {the best bid price is at level $\bar{S}_{0}-x$}, given that the latter event happened. One can {numerically} compute {$\rho(x,t)$} based on the initial bid queue size $Q^{b}_{x}(0)$ at level $\bar{S}_{0}-x$ and some specific assumptions about how the order flow takes place and the size of the best ask queue after a price change in a stationary state.  Intuitively, we expect {$\rho(x,t)$} to move in opposite direction with the initial queue size $Q^{b}_{x}(0)$: the larger $Q^{b}_{x}(0)$ is, the smaller the probability {$\rho(x,t)$} would be. For instance, slightly modifying the framework in \cite{cont2013price}, {suppose that, after a depletion of the best bid queue, which momentarily widens the spread, a flow of sell limit orders quickly fills the gap, and the spread reduces again to one tick\footnote{{This typically happens in just a few milliseconds according to \cite{cont2013price}. If, instead, a flow of buy limit orders fill the gap, we don't consider this time as a change of the ask price and move on to analyze the next time that the best bid price changes again.}}. We assume that the resulting size of the best ask queue after the gap is filled is drawn at random from a stationary distribution $f^{a}:\mathbb{Z}_{+}\to[0,1]$. Then,} $\rho(x,t)$ could be modeled as follows, for $x>\varepsilon$,
	\begin{equation}\label{ExprRho0}
		\rho(x,t):=\sum_{i=1}^{\infty}\sum_{j=0}^{Q^{b}_{x}(0)}{f^{a}(i)}\int_{0}^{t}f_{\tau}(s|0<\tau<t)P({{N^{b,x}_{s}}}=j)\alpha_{t-s}(i,Q^{b}_{x}(0)-j+1){ds,}
	\end{equation}
	where $\tau$ is the first time that the best bid price hits the level $\bar{S}_{0}-x$, {$f_{\tau}(s|0<\tau<t)$ is the density of $\tau$ conditioning on $0<\tau<t$}, ${N^{b,x}_{s}}$ represents the number orders cancelled by time $s$ out of the initial $Q^{b}_{x}(0)$ orders outstanding at level $\bar{S}_{0}-x$, and $\alpha_{u}(i,\ell)$ is the probability that the best bid of a LOB gets depleted before the best ask and before time $u$ when there are $i$ orders at the best ask and $\ell$ orders at the best bid at time $0$. We assume $\alpha$ only depends on the state of the LOB through $i$ and $\ell$. For $x=\varepsilon$, {${\rho(\eps,t)=\alpha_{t}(Q^{a}_{\eps}(0),Q_{\eps}^{b}(0)+1)}$, where now $Q^{a}_{\eps}(0)$} represents the outstanding orders at the best ask price at {time $0$}. The function $\rho(x,t)$ could be numerically computed, after imposing some reasonable assumptions on the LOB order flow and the dynamics of the best ask price (which determines $f_{\tau}$), and after estimating $f^{a}$ from real LOB data. Some details regarding the computation of $\rho(x,t)$ are provided in Section \ref{NwSect}.
\end{rem}

{Now, we are ready to move {to} the optimal placement problem in continuous time. We} assume that the {(average) time-step between price changes $\delta$} and the tick-size $\eps$ are {small-enough and related to each other in such a way that $\bar{S}^{(\delta,\eps)}$ can be approximated well by a suitable continuous-time process ${S:=\{S_{u}\}_{u\geq{}0}}$ {and $\rho^{(\delta,\varepsilon)}(x,t)\to \rho(x,t)$ for a smooth function $\rho(x,t)$}. Then, in the limit ($\delta,\eps\to{}0$), the analogous continuous time problem to (\ref{eq:CostContinuousApprox}) can be written as 
\begin{align}\label{eq:CostContinuous}\nonumber
C(x,t)&=E\left[\left.S_t-S_{0}\right|Y_t >S_0 -x\right]P\left(Y_t>S_0-x\right)\\
&\quad +P(Y_t \leq S_0-x)\left(-x-{\rho(x,t)}(r+f)\right)+f, \quad x>0.
\end{align}
In particular, the optimal placement $x^{*}(t)$, which minimizes ${C}(x,t)$ over all $x>0$, {only} depends on $r$ and $f$ through $r+f$, which can be considered as the ``penalty" for a non-executed limit order.} {In what follows we study the existence and behavior of the optimal placement $x^{*}(t)$ for arguably the two most important continuous models in finance: The Bachelier and the Black-Scholes Models.}

\section{{Optimal Order Placement Under The Bachelier Model}}
 In this section, we investigate the behavior of the optimal placement problem when 
the price process {$\{S_{t}\}_{t\geq{}0}$} follows a Brownian motion ({BM}) with drift, {often called the Bachelier model}.  All the proofs of this section are deferred to Appendix A. 

A Brownian motion with drift is a reasonable approximation for intermediate intraday time horizons (such as a few minutes) as shown by some recent works (see, e.g., \cite{cont2013price} and \cite{chavez2014one}). Also, we can see this model as the continuous-time counterpart of a correlated random walk (CRW). Concretely, by making the time step between price changes and tick size decay to $0$ in a certain way, a symmetric CRW converges to a drift-less Brownian Motion (cf. \cite[Section 3]{renshaw1981correlated}), while certain asymmetric CRW converges to a Brownian Motion with nonzero drift (cf. \cite{gruber2006diffusion}).

Let us start by giving a closed-form representation for the cost function.

\begin{restatable}{lemma}{primelemma}	
\label{BMCostFunc}

	Let  $dS_t = \mu dt + \sigma dW_t$, where  $W=\{W_t\}_{t\geq{}0}$ is a standard Brownian motion and $\mu \in \mathbb{R}$ and $\sigma>0$ are the drift and the volatility of the price process $S$, respectively. Then,  the cost function introduced in Eq.~(\ref{eq:CostContinuous}) admits the following representation:
	\begin{gather*}
	\begin{aligned}
		C(x,t)=& \left\{N\left(\frac{-x-\mu t}{\sigma\sqrt{t}}\right)+e^{\frac{-2x \mu }{\sigma^2} }N\left(\frac{-x+\mu t}{\sigma\sqrt{t}}\right) \right\}(-x-{\rho(x,t)(r+f)})\\
&+\mu t N\left(\frac{x+\mu t}{\sigma\sqrt{t}}\right)+e^{\frac{-2x \mu }{\sigma^2} }(2x-\mu t) N\left(\frac{-x+\mu t}{\sigma\sqrt{t}}\right)+f, \quad x>0.
	\end{aligned}
\end{gather*}
\end{restatable}

When {$\mu=0$,
	\begin{equation*}
		C(x,t) =-2N\left(\frac{-x}{\sigma\sqrt{t}}\right){\rho(x,t)}(r+f)+f, \quad x>0,
	\end{equation*}
which} is strictly increasing on  $x\in(0,\infty)$ {if, for instance, ${x\to\rho(x,t)}$ is nonincreasing}\footnote{{Per Remark \ref{CmptRho}, this is expected to happen if the initial state of the book $Q^{b}_{x}(0)$ is nondecreasing and, thus, the conclusion make sense.}}.
 Also, note that, for any $\mu \in (-\infty,\infty)$,
\begin{equation*}
	C(0^{+},t):=\lim \limits_{x\rightarrow 0^+} C(x,t)=-{\rho(0^{+},t)}(r+f)+f<  
	f,
\end{equation*}
and, {therefore, it is never optimal to immediately buy at the market order {at time $0$. Therefore,} for a zero drift {BM} {with nonincreasing ${x\to\rho(x,t)}$}, $C(0^{+},t)<{}C(x,t)$, for all $x>{}0$}. In that case, we {say that $x=0^+$ is the optimal placement solution and call it the ``trivial" optimal  placement solution}. Intuitively, the value $0^+$ represents an ``infinitesimal" number and, in practice, this can be interpreted as the strategy of putting the limit buy order at the best {or second best bid price}. The just mentioned optimal order placement for a driftless {BM} is consistent with the conclusion of \cite{guo2013optimal} for a symmetric correlated random walk, which is {expected}, since, as mentioned {above}, the Brownian motion with zero drift is the diffusion limit of a symmetric correlated random walk {used in \cite{guo2013optimal}}.

 It should be expected that when the drift is positive, the optimal placement policy is still $x=0^{+}$ {for nonincreasing functions {$x\to\rho(x,t)$}}. However, for negative drifts, there should exist a non trivial optimal placement solution. As shown in Figure~\ref{Graph:Expcost_two_t}, this is not necessarily the case if the time horizon is small. The following result explores conditions for the existence of a nontrivial optimal placement solution.

\begin{figure}
\centering
	\includegraphics[height=7cm,width=0.6\textwidth]{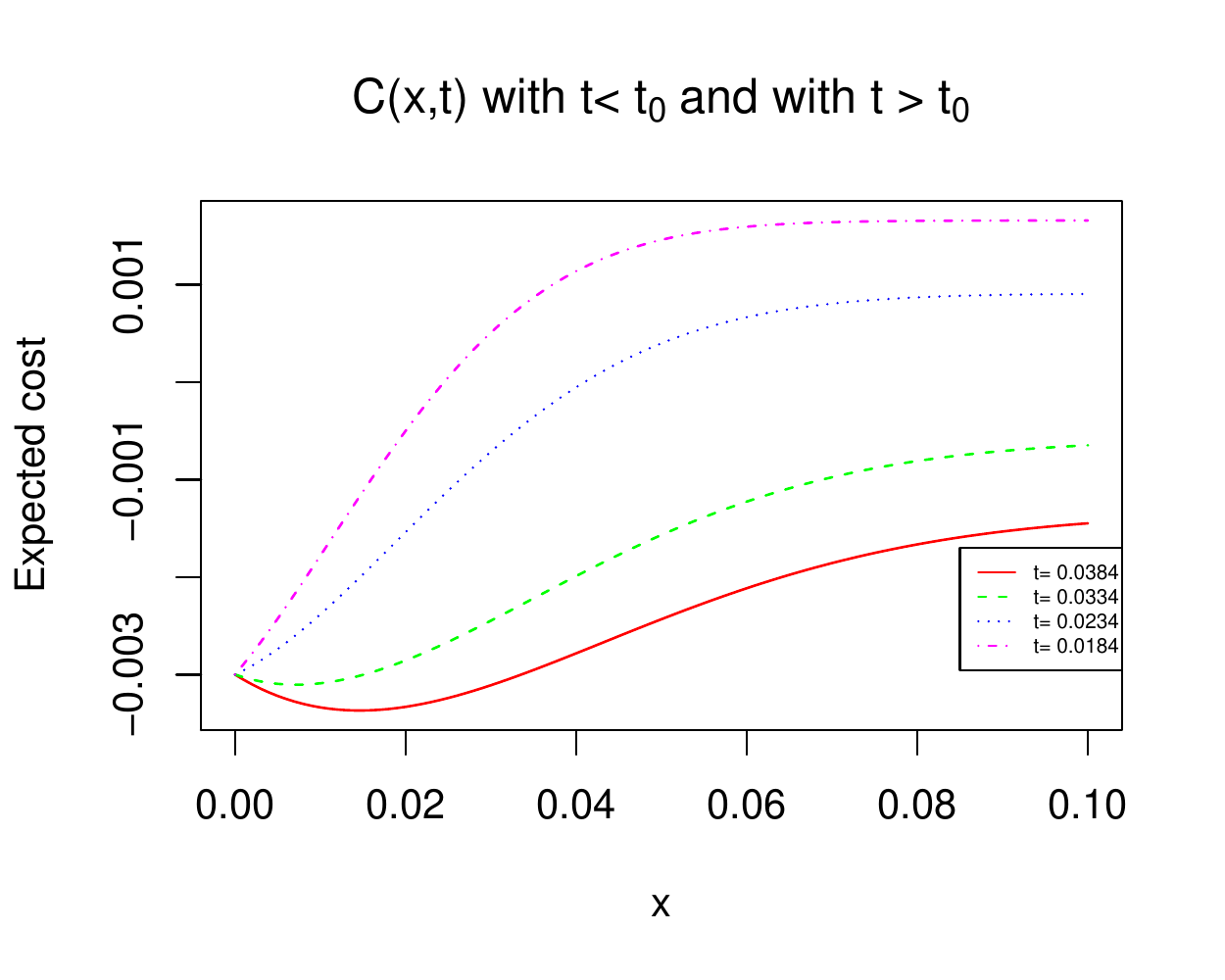}
	\caption{{$C(x,t)$ against $x$ with $t=0.0184$ (magenta), $t=0.0234$ (blue), $t=0.0334$ (green), and $t=0.0384$ (red) when $t_0=0.0284$, $(r+f)\rho(0^{+})=0.006, \sigma=0.2$.}}
	\label{Graph:Expcost_two_t}
\end{figure}

\begin{restatable}{theorem}{maintheorem}
\label{mainresult}
Let $C(x,t)$ and $\{S_t\}_{t\geq{}0}$ be as in Lemma \ref{BMCostFunc} {and suppose that $r+f>0$ and $\rho(x,t)$ is $C^{2}$ on $(0,\infty)\times (0,\infty)$. Then, the} following assertions hold:
\begin{enumerate}
	\item[(i)] If $\mu>0$ and $x\to{}\rho(x,t)$ is nonincreasing for each $t>0$, then $C(0^+,t) < C(x,t)$, for all $x>{}0$.
	\item[(ii)] Suppose that $\mu<0$ and that the following conditions hold for a given time horizon $t>0$:
\begin{align}\label{Sec5:Conditions0}
		\rho(0^{+},t)<\frac{2|\mu|t}{r+f}, \qquad {\frac{\partial\rho(0^{+},t)}{\partial x}\geq{}0}{.}
\end{align}
Then, there exists $x^*(t)\in(0,\infty]$ such that
	\[
		{C(x^*(t),t) \leq C(x,t), \quad \text{for all } x> 0.}
	\]
	Furthermore, $x^{*}(t)<\infty$, if the following additional conditions hold:
	\begin{align}
		\label{Sec5:Conditions0b}
		{\limsup_{x\to\infty}\frac{\partial \rho(x,t)}{\partial x}\leq{}0, \quad
		\liminf_{x\to\infty}\rho(x,t)x^{2}>\frac{2|\mu|\sigma^{2}t^{2}}{r+f},}
\end{align}
\end{enumerate}
\end{restatable}

\begin{rem}\label{Remark3.3}
	Per our discussion in Section \ref{Sec2}, one can interpret $\rho(0^{+},t)$ as $\rho(\varepsilon,t)$, the probability that a limit order put on the best bid queue at time $0$ is executed before (or at) the next price change and before time t. The condition $\partial_x \rho(0^+,t)>0$ essentially says that the $\rho(2\varepsilon,t)>\rho(\varepsilon,t)$ {and, as argued in Section \ref{NwSect}, is typically met in practice when there is some LOB imbalance (e.g., $0<Q^{b}_{\eps}(0)-Q^{a}_{\eps}(0)\leq{}Q^{b}_{2\eps}(0)-\sum_{i}if^{a}(i)$, where we used the same notation as in (\ref{ExprRho0}))}. {The conditions in (\ref{Sec5:Conditions0}) are} needed to rule out that $x=0^{+}$ may be optimal, while the conditions in (\ref{Sec5:Conditions0b}) are needed to rule out that $x=+\infty$ may be optimal. In Section \ref{NwSect}, we prove that, for a large class of models, the first condition in (\ref{Sec5:Conditions0b}) holds, while, the quantity $\liminf_{x\to\infty}\rho(x,t)x^{2}$ appearing in the second condition therein can be {lower bounded by an explicit quantity and, thus, precise constrains on the model's parameters can be imposed} for this condition to be satisfied {(see Proposition \ref{Rhoxtx2:xInfty:General} and Remark \ref{Rhoxtx2:xInfty:GeneralRem} below)}.
\end{rem}

{The two conditions in (\ref{Sec5:Conditions0}) are needed to guarantee that $\partial_{x}C(0^{+},t)<0$ (and, hence, to rule out that $x^{*}(t)={}0$). It is natural to consider the smallest time horizon $t_0$ such that a nontrivial optimal placement policy would exist for any $t>t_{0}$. Concretely, let
\begin{equation}\label{Dfnt0a}
	{t_0:=\inf\left\{u\geq{}0: \frac{\partial C}{\partial x}(0^{+},s)<0,\;\text{ for all }s>u\right\}}.
\end{equation}
This} threshold $t_0$ 
is important because if the  investor's preferred time horizon were bigger than $t_0$,  then there would be a nonzero optimal {placement} for the limit order available to him/her. 
As a corollary of the proof of Theorem \ref{mainresult}, we deduce the following upper bound for $t_{0}$, {which remarkably is the same for any value of $\sigma$.}
\begin{restatable}{corollary}{secondthm}
\label{Crl1}
{Suppose that $\mu<0$ and that the second condition in (\ref{Sec5:Conditions0}) is satisfied for all $t$. Let $t_{0}$ be defined as in (\ref{Dfnt0a}).} Then, we have that $t_{0}\leq\bar{t}_{0}$, where $\bar{t}_{0}$ is such that $\rho(0^{+},t)<2|\mu|t/(r+f)$ for all $t>\bar{t}_{0}$. {In particular, $t_{0}\leq{}(r+f)/2|\mu|$.}
\end{restatable}

{It is easy to see that $\rho(0^{+},t)$ is nondecreasing with $t$ and, as argued in Section \ref{NwSect}, it typically converges to its maximum value in just a few seconds under reasonable market conditions  {(see Figure \ref{rhoxt_timesx} therein)}. For these reasons, hereafter we assume that $\rho(0^+,t)$ is constant in $t$ and use  $\rho(0^+)$ for the rest of this Section 3.} 

Corollary \ref{Crl1} implies that $t_0$ is upper bounded by $(r+f)/2|\mu|$ and, thus, it gets smaller when the drift gets negatively larger or when the sum of the rebate and fee gets smaller. 
While having $\mu$ large may be too much to ask in practice, we do have that $r+f$ is quite small {in practice}. Hence, it is natural to ask about the asymptotic behavior of $t_{0}$ as $(r+f)\to{}0$. The following result provides further information about $t_{0}$.

\begin{restatable}{theorem}{secondthm}
\label{thm2b}
Let the {second condition in (\ref{Sec5:Conditions0}) and the two conditions in (\ref{Sec5:Conditions0b}) be satisfied for all $t>0$. Let also assume} that $\rho(0^{+},t)\equiv \rho(0^{+})\in(0,1]$, for all $t$, and $\limsup_{t\to{}0}\partial_{x}\rho(0,t)<\infty$. Then, {the critical time $t_{0}$ defined in (\ref{Dfnt0a}) is positive and, hence, is such that} $\partial_{x}C(0^{+},t_{0})=0$. Furthermore, we have 
\begin{equation*}
	{\lim_{(r+f)\to 0} \frac{t_0}{(r+f)} = \frac{{\rho(0^{+})}}{2 |\mu|}}.
\end{equation*}
{If, in addition, $\partial_{t}\partial_{x}\rho(0,t)\geq{}0$, for all $t$, then $t_{0}$ is the only solution of the equation $\partial_{x} C(0^{+},t)=0$.}
\end{restatable}

From a practical point of view, the approximation provided by the previous result is quite important since in most markets $r+f$ is negligible. For instance, for any US exchange, there is a fee and rebate cap of \$0.003 per share\footnote{See Code of Federal Regulations, Title 17, 242.610(c)(1) and Securities Exchange Act Release No. 51808 (Jun. 9, 2005), 70 FR 37496, 37545 (Jun. 29, 2005) (File No. S7-10-04)}, which makes $r+f\leq \$0.006$. Broadly, if the investor can wait longer than $(r+f){\rho(0^{+})}/2|\mu|$, he/she can use an optimal placement strategy that is better than the placement at the best bid or a market order. In the left panel of Figure~\ref{Graph:RelError_t0}, we show the relative error of $\bar{t}_{0}:=(r+f)\rho(0^{+})/2|\mu|$ against $|\mu|$ when $(r+f)\rho(0^{+})=0.006$ and $\sigma=0.1$. 

\begin{figure}
\centering
\includegraphics[height=8cm,width=0.5\textwidth]{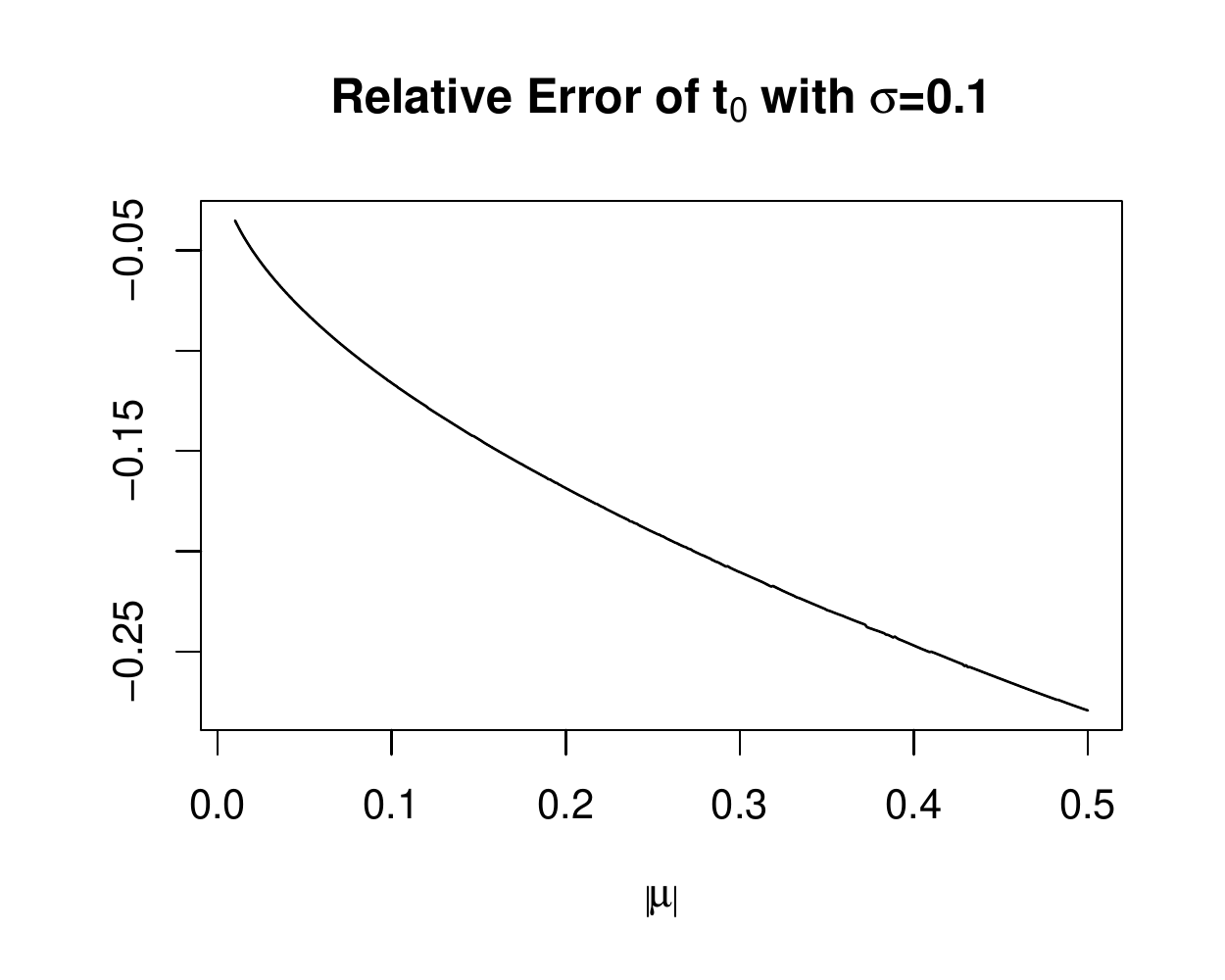}\includegraphics[width=0.5\textwidth,height=8 cm]{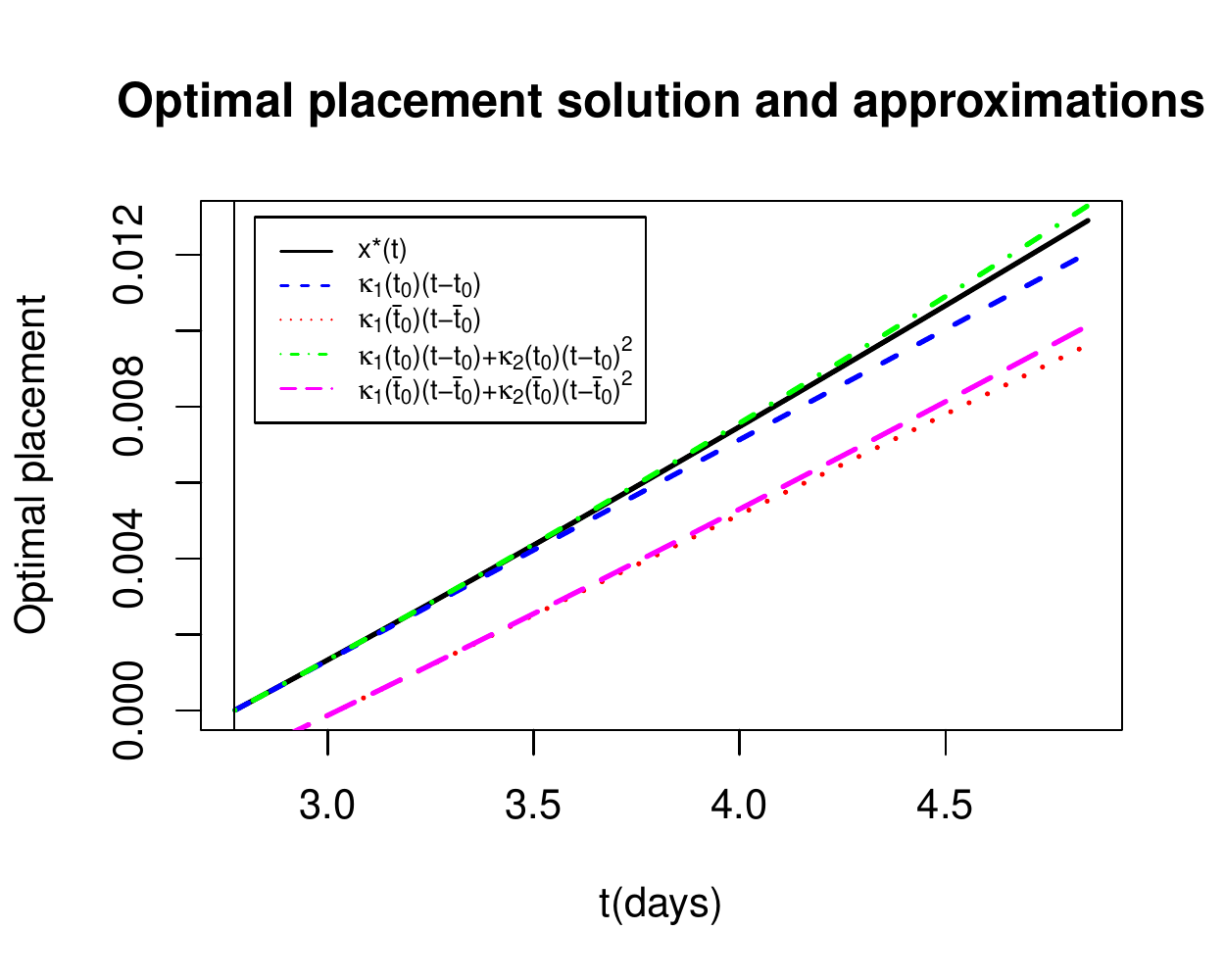}
\caption{Left Panel: Relative error, $(t_0-\bar{t}_0)/t_{0}$ for different values of $|\mu|$ when $(r+f)\rho(0^{+})=0.006$, $\sigma=0.1$, and $\bar{t}_0=\rho(0^{+})(r+f)/2/|\mu|$. Right Panel: Performance of the first- and second-order approximation derived from (\ref{FirstSecondApprx}) when $(r+f)\rho(0^{+})=0.006,\sigma=0.2,\mu=-0.25$: $x^*(t)$(black), $\kappa_1(t_0)(t-t_0)$(blue), $\kappa_1(\bar{t_0})(t-\bar{t}_0)$(red), $\kappa_1(t_0)(t-t_0)+\kappa_2(t-t_0)^2$(green), and $\kappa_1(\bar{t}_0)(t-\bar{t}_0)+\kappa_2(\bar{t}_0)(t-\bar{t}_0)^2$(magenta) against $t$(days).}\label{Graph:RelError_t0}
\end{figure}

 We now analyze the behavior of optimal {placement, where} the investor should put a limit order to minimize the expected cost.
\begin{theorem}
\label{BM:xt:ttot0}
	 {Suppose that the conditions of Theorem~\ref{thm2b} are satisfied and also that $\partial_{x}^{2}\rho(0^{+},t)<0$ and $\partial_{t}\partial_{x}\rho(0^{+},t)\geq{}0$, for all $t>0$. Also, let $t_0>0$ and $x^*(t)$, for $t>t_{0}$,  be as described in Theorems \ref{mainresult} and \ref{thm2b}. Then, as $t \searrow t_0 $, 
	 	\begin{equation}\label{FirstSecondApprx}
	  x^*(t)=\kappa_{1}(t-t_0)+\kappa_2(t-t_0)^{2}+o((t-t_0)^{2}),
	\end{equation}
	where} 
 \[
\kappa_1 :=-\frac{\frac{\partial^2 C}{\partial t \partial x} \left(0, t_{0}\right)}{\frac{\partial^2 C}{\partial x^2}\left({0, t_0}\right)}>0,\quad \kappa_2 := -\frac{\frac{1}{2}\frac{\partial^3 C}{\partial x^3}(0,t_0)\kappa_1^2 + \frac{\partial^3 C}{\partial t \partial x^2}(0,t_0)\kappa_1 + \frac{1}{2}\frac{\partial^3 C}{\partial x \partial t^2}(0,t_0) }{\frac{\partial^2 C}{\partial x^2}(0,t_0)} .
\]
\end{theorem}
{Let us remark that all the partial derivatives involved in $\kappa_{1}$ and $\kappa_{2}$ can be evaluated in closed form in terms of $N(\mu\sqrt{t_{0}}/\sigma)$ {(e.g., see (\ref{bm:2ndMixedDer:x0}) and (\ref{bm:2nd:x0}) below for the derivatives involved in $\kappa_{1}$)}.} 
 The above result gives us the first and second order approximation of the optimal placement solution when {$t \searrow t_0$}.
These approximations require the value of $t_0$, which, in light of Corollary ~\ref{thm2b}, can be approximated well when {$r+f \to 0$. Therefore,} the combination of Theorems~\ref{thm2b} and \ref{BM:xt:ttot0} gives us a simple approximation of the optimal placement solution $x^*(t)$, when $t \to t_0$ and {$(r+f)$ is} small, which is a reasonable assumption in most markets.
In the right panel of Figure~\ref{Graph:RelError_t0}, we graph the explicit optimal place solution, the first- and  second- order approximations using $t_{0}$ and also replacing $t_0$ with its approximation $\bar{t}_0=\rho(0^{+})(r+f)/2|\mu|$. As displayed in the figure, the performance of approximations using $\bar{t}_{0}$ is slightly less than the approximations using $t_0$.

 In the rest of this section, we analyze the behavior of the optimal placement solution {$x^*(t)$} for large time horizons $t$. For simplicity, we assume that $\rho=\rho(x,t)$ is constant in $x$ and $t$. Our first result in this direction gives upper and lower {estimates.}
\begin{restatable}{theorem}{thirdthm}
\label{xtupbd}
 {Let $\mu<0$,  $\theta_0:= \sqrt{1- 2\sigma^2/(\mu \rho (r+f))}$, and {let $x^*(t)$, for $t>{}t_{0}$,} be as in Theorem ~\ref{mainresult}. Then, 
 \[
 	-\sigma \sqrt{t} - \mu t \theta_0\leq{} x^*(t)\leq{}-\mu\theta_0 t,
\]
where the first and second inequalities above hold for any $t> \max\Big(\frac{\rho(r+f)}{-\mu}, \frac{\sigma^2}{\mu^2 (\theta_0-1)^2}  \Big)$ and $t>t_{0}$, respectively.} 
\end{restatable}
{As a} corollary of Theorem \ref{xtupbd}, we can deduce the first order approximation of $x^{*}(t)$ {as $|\mu|\theta_{0}t$, when the investor's time horizon $t$ is large}. The following theorem provides the second order approximation.
\begin{restatable}{theorem}{asymptotics}
\label{asymptotics}
Let $\mu<0$ and {$x^{*}(t)$} be the optimal position as defined in Theorem \ref{mainresult}  and $\theta_{0}$ be as in Theorem \ref{xtupbd}. Then,
\begin{equation}\label{MAH}
	\lim_{t\to\infty}t\left(\frac{{{x^{*}(t)}^{2}}}{t^{2}}-\mu^{2}\theta_{0}^{2}\right)=\theta_{1},
\end{equation}
where
	$\theta_{1}:={\frac{\sigma^{4}}{{2\rho}(r+f)|\mu|\theta_{0}}\left[-6\frac{(\theta_{0}-1)}{(\theta_{0}+1)^{2}}+\left(1+\frac{2 \mu {\rho}(r+f)}{\sigma^{2}}\right)\frac{(\theta_{0}-1)}{\theta_{0}+1} -\frac{(\theta_{0}+1)^{2}}{(\theta_{0}-1)^{2}}\right]}$.
\end{restatable}

\section{{Optimal Order Placement Under The Black-Scholes Model}}

While {a} Brownian motion {with drift is able to broadly} capture the price movement in short time (say, just a few minutes), geometric Brownian motion (GBM), also known as the Black-Scholes model, is generally believed to provide a better fit for longer time periods. Therefore, it is both natural and important to study the behavior of the optimal placement problem under this paradigm. All the proofs in this section are deferred to Appendix B.

The following lemma provides a closed form {representation} of the expected cost function {introduced in Eq.~(\ref{eq:CostContinuous}),} when the price process $S$ follows {a} geometric Brownian motion {with volatility $\sigma$ and drift $\mu$}.
\begin{lemma}
\label{GBMCostFunc}
	{Let $S_{t}:=S_{0}\exp\Big((\mu-\sigma^{2}/2)t+\sigma W_{t}\Big)$, $t\geq{}0$, %
where $\{W_t\}_{t\geq{}0}$ is the standard Brownian motion, and let $\widetilde{C}(y,t)$ be the expected cost if one puts a limit order at the price level $S_0 e^{-y}$ $(y>0)$; i.e., $\widetilde{C}(y,t):=C(S_{0}-S_0 e^{-y},t)$, with $C(x,y)$ defined as in Eq.~(\ref{eq:CostContinuous}). Similarly, fix $\tilde{\rho}(y,t):=\rho(S_{0}-S_{0}e^{-y},t)$}. Then, $\widetilde{C}(y,t)$ can be written as 
\begin{equation*}
\begin{aligned}
{\widetilde{C}(y,t)} &=  (S_0 e^{-y}-{(r+f)\tilde{\rho}(y,t)}) \bigg[ N \bigg(\frac{-y-\mu t+ \frac{\sigma^2 t}{2} }{\sigma\sqrt{t}} \bigg) + e^{-\frac{2y\mu}{\sigma^2}+y} N \bigg(\frac{-y+\mu t- \frac{\sigma^2 t}{2} }{\sigma\sqrt{t}} \bigg) \bigg]\\
	&+ S_0 \bigg[ e^{\mu t} N \bigg(\frac{y+\mu t+ \frac{\sigma^2 t}{2} }{\sigma\sqrt{t}} \bigg) 
	-e^{ -\frac{2y\mu}{\sigma^2}+\mu t -y} N \bigg(\frac{-y+\mu t+ \frac{\sigma^2 t}{2} }{\sigma\sqrt{t}} \bigg) \bigg] +f-S_0.
\end{aligned}
\end{equation*}
\end{lemma}
With certain abuse of notation, in what follows, we simply write $C(y,t)$ and $\rho(y,t)$ for $\widetilde{C}(y,t)$ and $\tilde{\rho}(y,t)$, respectively. As with the {BM} with drift, we have
	\[
	C(0^+,t):=\lim \limits_{y \to 0} C(y,t) = f-(r+f)\rho(0^{+},t)<f,
	\]
and, thus, it is again never optimal to immediately place a market order. Similarly, when $\mu=0$, 
\begin{equation}\label{zeromu}
	{C}(y,t)=-(r+f)\rho(y,t)\left[ N\left(\frac{-y+\frac{\sigma^2}{2}t}{\sigma \sqrt{t}}\right)+e^{y}  N\left(\frac{-y-\frac{\sigma^2}{2}t}{\sigma \sqrt{t}}\right)\right]+f-S_{0},
\end{equation}
which is an increasing function on $y\in (0,\infty)$ when $y\to{}\rho(y,t)$ is nonincreasing\footnote{This is because the function inside the square brackets in (\ref{zeromu}) is the probability that a limit order placed at $S_{0}-S_{0}e^{-y}$ is executed, which decreases as $y$ increases.}. We can again say that $y=0^+$ is the optimal placement solution, which can be interpreted as the strategy where the investor place his limit order at the best or second best bid price. {As before, we call this the ``trivial" limit order  placement.}	
Now, the following theorem shows that there is {a} non-trivial optimal placement solution {when} the drift is negative {and the time horizon is long enough}.
\begin{theorem}
\label{gbmtknotexistence}
	Let $C(y,t):=\tilde{C}(y,t)$ be as in Lemma \ref{GBMCostFunc} and suppose that $\mu<0$ and that the following conditions hold:
\begin{equation}\label{ConditionsGBM}
	\rho(0^{+},t)<\frac{a{}S_{0}|\mu|t}{r+f},\qquad {\frac{\partial\rho(0^{+},t)}{\partial y}\geq{}0},
\end{equation}
where $a=2$ if $\mu>-\sigma^{2}/2$ and $a=1$ if $\mu\leq{}-\sigma^{2}/2$. Then, there exists a $y^*(t) \in (0,\infty]$ such that
	\[
		C(y^{*}(t),t)\leq{}C(y,t),\quad \text{for all}\quad y>{}0.
	\]
Furthermore, $y^{*}(t)<\infty$ if the following additional conditions hold:
	\begin{equation}\label{ConditionsGBM2}
		\limsup_{y\to\infty}\frac{\partial \rho(y,t)}{\partial y}\leq{}0, \quad
		\liminf_{y\to\infty}e^{y}y^{2}\rho(y,t)>\frac{2S_{0}\sigma^{2}t^{2}|\mu|}{r+f}.
\end{equation}
\end{theorem}
\begin{rem}\label{Remark3.3GBM}
	{In Section 5, we shall verify the plausibility of the conditions in Theorem \ref{gbmtknotexistence}. In particular, we prove that the first condition in (\ref{ConditionsGBM2}) is satisfied for a large class of models, while the second condition in (\ref{ConditionsGBM2}) is always satisfied under mild conditions. Based on empirically reasonable parameters, the second condition in (\ref{ConditionsGBM}) is also typically satisfied.}
\end{rem}

{As in the Bachelier model, the two conditions in (\ref{ConditionsGBM}) guarantee that $\partial_{y}C(0^{+},t)<0$ and, hence, that the optimal placement problem admits a nontrivial solution. It is then natural to define
\begin{equation}\label{Dfntstar0}
	t_0^{*}:=\inf\left\{u\geq{}0: \frac{\partial C}{\partial y}(0^{+},s)<0,\;\text{ for all }s>u\right\}.
\end{equation}
{The critical} value $t_{0}^{*}$ is such that a nontrivial {optimal} placement policy would exist for any investment horizon $t>t_{0}^{*}$.
As in the Bachelier model, in order to specify further the critical time $t_{0}^{*}$ and provide an estimate when $r+f\to{}0$, we need further assumptions. This is provided in the following result.}
\begin{theorem}\label{prop:t0limit}
Let the {second condition in (\ref{ConditionsGBM}) and the two conditions in (\ref{ConditionsGBM2}) be satisfied for all $t>0$. Let $t_0^*$ be defined as in (\ref{Dfntstar0}). In addition, we assume that $\rho(0^{+},t)\equiv \rho(0^{+})\in(0,1]$, for all $t$, and $\limsup_{t\to{}0}|\partial_{y}\rho(0,t)|<\infty$. Then, $t_{0}^{*}$ is positive and, thus, is such that $\partial_{y}C(0^{+},t_{0}^{*})=0$. Furthermore, when $(r+f)/S_0 \to 0$, we have
			\begin{equation}\label{LRNH}
				t_0^{*} \sim \frac{\rho({0^{+}})(r+f)}{2|\mu|S_0}.
		\end{equation}
If, in addition, $\partial_{t}\partial_{y}\rho(0,t)\geq{}0$, for all $t$, then $t_{0}^{*}$ is the only solution of the equation $\partial_{y} C(0^{+},t)=0$.}
\end{theorem}

In the left panel of Figure~\ref{Graph:GBM:t0}, we graph $t^*_0$ and $\bar{t}_{0}:=(r+f)\rho(0^{+})/(2|\mu|S_0)$ from Theorem~\ref{prop:t0limit} against $S_0$. As shown therein, $\hat{t}_{0}$ converges to $t^*_0$ at a significantly fast rate, so that $\hat{t}_{0}$ is a reasonably accurate approximation of $t^*_0$. {Per our discussion after Theorem \ref{gbmtknotexistence} and (\ref{LRNH}), we know that $\bar{t}_{0}:=(r+f)\rho(0^{+})/(2|\mu|S_0)$ is a tight upper bound for $t^*_0$ when $\mu>-\sigma^{2}/2$. However, for $\mu\leq{}-\sigma^{2}/2$, it is not known whether or not $\bar{t}_{0}$ is still an upper bound for $t^*_0$. However, $\tilde{t}:=(r+f)\rho(0^{+})/(|\mu|S_0)$ is an upper bound (in both cases).}

\begin{figure}
\centering
\includegraphics[height=8cm,width=0.5\textwidth]{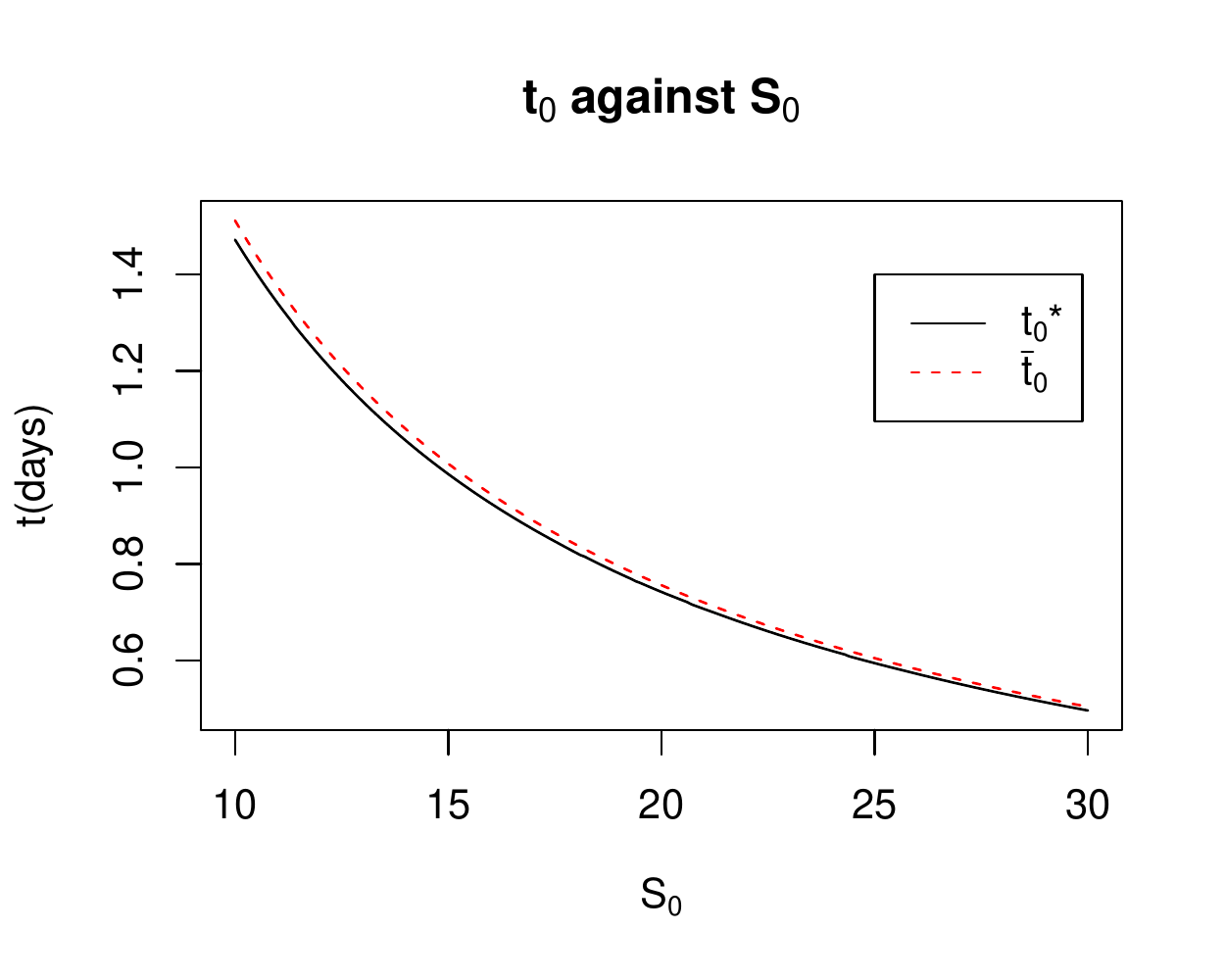}\includegraphics[height=8cm, width=0.5\textwidth]{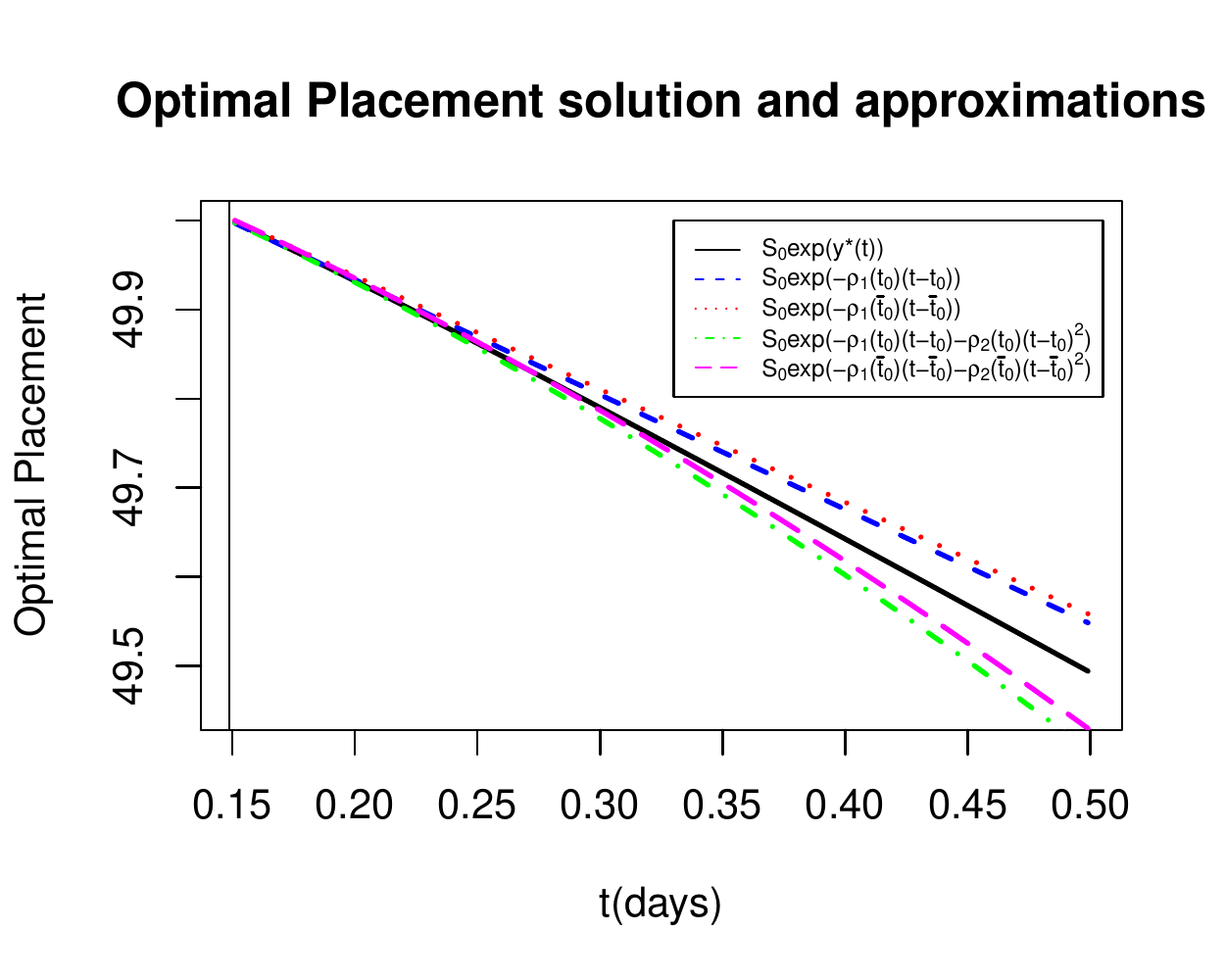}
\caption{Left Panel: $t^*_0$(black) and $\bar{t}_{0}=(r+f)\rho(0^{+})/(2|\mu|S_0)$(red dotted line) against $S_0$ when $(r+f)\rho(0^{+})=0.006, \sigma=0.2, \mu=-0.05$. Right Panel: $S_0 \exp (-y^*(t))$(black), $S_0 \exp (-\rho_1(t_0)(t-t_0))$(blue), $S_0 \exp (-\rho_1(\bar{t_0})(t-\bar{t_0}))$(red), $S_0 \exp (-\rho_1(t_0)(t-t_0)-\rho_2(t_0)(t-t_0))$(green), $S_0 \exp (-\rho_1(\bar{t_0})(t-\bar{t_0})-\rho_2(\bar{t_0})(t-\bar{t_0})^2)$(magenta) against $t$(days) when $(r+f)\rho(0^{+})=0.006, \sigma=0.2, \mu=-0.1, S_0=50$.}
\label{Graph:GBM:t0}
\end{figure}

We now proceed to discuss the behavior of the optimal placement solution when the investor has a time horizon $t$  close to $t_0^*$.  {The following result is the analogous of Theorem \ref{BM:xt:ttot0}.} 
\begin{restatable}{theorem}{ystar_ttot0}
\label{ttot0}
Suppose that the conditions of Theorem \ref{prop:t0limit} are satisfied and also that $\partial_{y}^{2}\rho(0,t)<0$ and $\partial_{t}\partial_{y}\rho(0,t)>0$, for all $t$. Then, as $t \searrow t_0^{*} $, 
	 	\begin{equation*}
	  y^*(t)=\rho_{1}(t-t_0^{*})+\rho_2(t-t_0^{*})^{2}+o((t-t_0^{*})^{2}),
	\end{equation*}
where 
 \[
\rho_1 :=-\frac{\frac{\partial^2 C}{\partial t \partial y} \left({0, t_{0}^*}\right)}{\frac{\partial^2 C}{\partial y^2}\left({0, t_0^*}\right)},\quad \rho_2 :=-\frac{ \frac{1}{2}\frac{\partial^3 C}{\partial y^3}({0,t_0^*})\rho_1^2 + \frac{\partial^3 C}{\partial t \partial y^2}({0,t_0^*})\rho_1 + \frac{1}{2}\frac{\partial^3 C}{\partial y\partial t^2}({0,t_0^*})}{\frac{\partial^2 C}{\partial y^2}({0, t_0^*})} .
\]
\end{restatable}

{Let us remark that all the partial derivatives involved in $\rho_{1}$ and $\rho_{2}$ can be evaluated in closed form in terms of $N(\sqrt{t_{0}^{*}}(\mu\pm \sigma^{2}/2)/\sigma)$ (see (\ref{MixedPartial:y0:L1}) and (\ref{SecondPartial:y0:L1}) below for the derivatives involved in $\rho_{1}$). 
To use the approximation,} {an investor would need an approximation of $t_0^*$}, which {can be obtained from Theorem~\ref{prop:t0limit} in a small fee/rebate regime. Therefore,} the combination of Theorems ~\ref{prop:t0limit} and \ref{ttot0} gives us {a simple, yet accurate, approximation of the} optimal placement solution. 
In right panel of Figure~\ref{Graph:GBM:t0}, we show the optimal placement $S_{0}e^{-y^{*}(t)}$ and its first- and second-order approximations with $t_0$ and $t_0^*$ as given in Theorem~\ref{prop:t0limit}. It is evident that the second-order approximation shows comparably better performance than the first-order approximation when $t$ is close to $t_0^*$, while the result is opposite when $t$ is large.

{In addition to the upper bounds we had mentioned, our next result also provides lower bounds for $t_0^{*}$.}
\begin{proposition}
\label{tknot_gbm_bounds}
{Suppose that the conditions of Theorem \ref{prop:t0limit} are satisfied and let $\bar{t}:=\rho({0^{+}})(r+f)/(2|\mu|S_0)$ and $\tilde{t}:=\rho({0^{+}})(r+f)/(|\mu|S_0)$. Also, let $\underline{t}(z)$ be defined as
\begin{align*}
\underline{t}(z)=\rho^{2}({0^{+}})(r+f)^{2}\left(\frac{-(\mu-\frac{\sigma^{2}}{2})- \sqrt{(\mu-\frac{\sigma^{2}}{2})^{2}-\frac{32\sigma^{2}\phi(0)S_{0}z}{\rho({0^{+}})(r+f)}}}{8\sigma S_{0}z}\right)^{2}.
\end{align*}
Then, $\underline{t}(\mu\phi(0))<t_{0}^*<\tilde{t}$, if $\mu<-\sigma^{2}/2$, while $\underline{t}(\mu\phi(0)-\beta\sqrt{-\mu /2}e^{-0.5})<t_{0}^*<\bar{t}$, if $\mu>-\sigma^{2}/2$.}
\end{proposition}
The upper and lowers bound obtained in the previous result depend on the sign of $\beta:=\mu-\sigma^{2}/2$. Broadly,  the case $\beta>0$ represents a large $\sigma$ or small $\mu$ regime. In Figure~\ref{Graph:GBM:t0Bounds}, we show the {lower bound $\underline{t}$ as well as $\bar{t}$ when $\beta>0$ (left panel) and $\beta<0$ (right panel). As seeing therein, the lower bound $\underline{t}$ is rather rough in the case $\beta>0$, but performs extremely well when $\beta<0$. For the chosen parameter setting, $\bar{t}$ also provides a good approximation and turns out to upper bound $t_{0}^{*}$ in both cases, though we only have proof of this when $\beta>0$.} 
\begin{figure}
\centering
\includegraphics[height=7cm,width=0.48\textwidth]{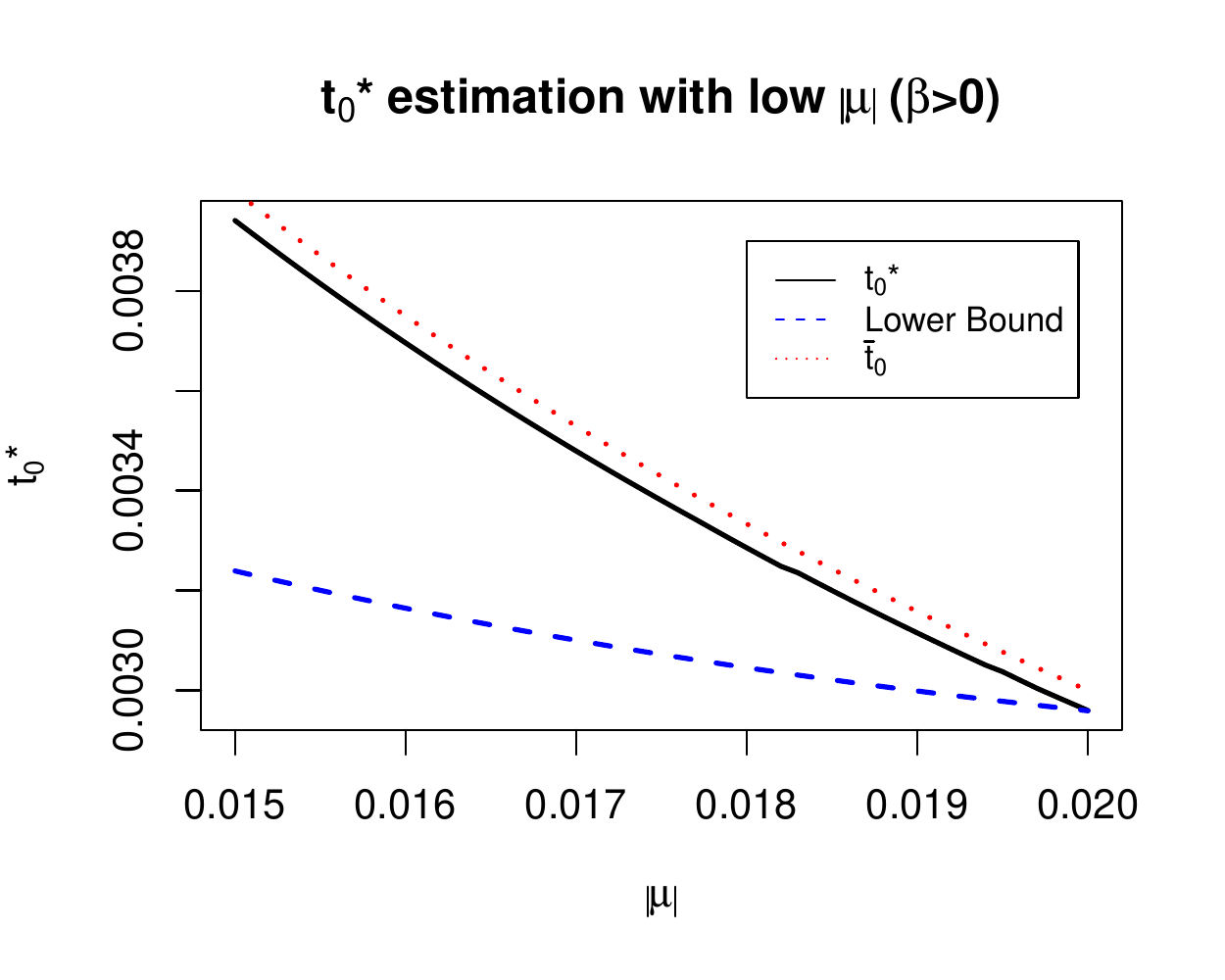}
\includegraphics[height=7cm,width=0.48\textwidth]{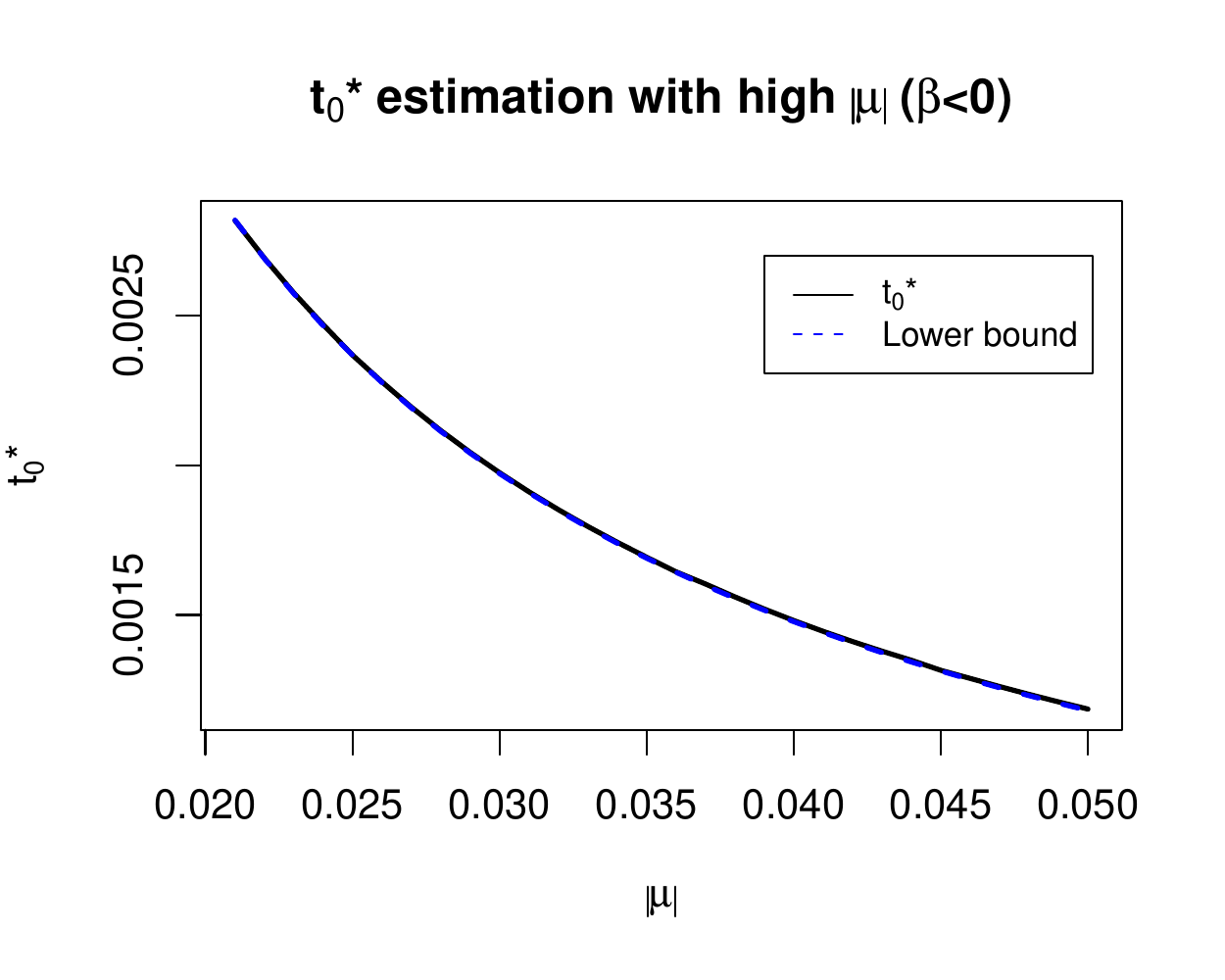}
\caption{Left: {$t^{*}_0$ (solid black), $\bar{t}$ (dotted red), and the lower bound $\underline{t}$ (dashed blue) when $\beta>0$. Right: $t^{*}_0$ (solid black), $\bar{t}$ (dotted red), and the lower bound $\underline{t}$ (dashed blue) when $\beta<0$}. In both graphs, $r+f=0.006$, $\sigma=0.2$, and $S_0=50$. }
\label{Graph:GBM:t0Bounds}
\end{figure}

Now, we proceed to analyze the behavior of the optimal solution when investor's time horizon, $t$, is large. {As in the Bachelier model, we assume that $\rho=\rho(y,t)$ is constant in $y$ and $t$, which for large $t$ is a reasonable assumption as argued in Section \ref{NwSect}.}
\begin{restatable}{theorem}{ystar_tinfty}
\label{tinfty:regime1}
Suppose that $\rho=\rho(y,t)$ is constant in $y$ and $t$. Let $\mu<0$ and {$y^{*}(t)$} be the optimal position as defined in Theorem \ref{gbmtknotexistence}.
Then, we have
\begin{equation*}
	 \lim_{t \to \infty}\frac{ \frac{y^*(t)}{t} - (-\mu + \frac{3}{2}\sigma^2 ) + \sigma \sqrt{-2\mu +2\sigma^2}}{\ln t /t} 
  =  \frac{\sigma}{2\sqrt{-2\mu+2\sigma^2}},
\end{equation*}
and, in particular, $y^*(t)/t \to -\mu + \frac{3}{2}\sigma^2  - \sigma\sqrt{-2\mu+2\sigma^2}$, as $t\to\infty$.
\end{restatable}

\par  When the investor can wait for a long period of time, the previous result provides a suitable approximation for the optimal placement of a limit order. However, there is a shortcoming in this analysis because we are not considering time value of money.

\par {For the final asymptotic regime, we consider the behavior of the optimal placement problem in a low volatility case, i.e., $\sigma\to{}0$. This theorem gives us the first and second order approximation of optimal placement solution when the volatility is small. If the the investor choose to participate in a market, which has the relatively low price volatility}, this approximation can be useful. For simplicity, we again assume that {$\rho=\rho(x,t)$ is constant in $x$ and $t$}.
\begin{restatable}{theorem}{deltaApprox}
\label{sig0:delta_approx}
Suppose that $\rho=\rho(y,t)$ is constant in $y$ and $t$ and let $\mu<0$.
	{Let us denote the expected cost function as $C(y,\sigma)$, i.e., a function of $y$ and $\sigma$. Then, there exists a $\sigma_{0}>0$ such that for each $0<\sigma<\sigma_{0}$, there exists $y^*(\sigma)$ in $(0,-\mu t)$ such that $C(y^{*}(\sigma),\sigma)\leq{}C(y,\sigma)$, for all $y>0$. Furthermore, we have that:
		\begin{equation*}
		y^*(\sigma)= -\mu t - \sqrt{2\sigma^2 t \ln (1/ \sigma)} + \frac{a}{2\ln (1/\sigma)}\sqrt{2\sigma^2 t \ln (1/ \sigma)}+o\left(\frac{\sigma}{\sqrt{\ln(1/\sigma)}}\right),
	\end{equation*}
	where $a:=\ln S_0 + \mu t + \frac{1}{2} \ln t - \ln ({\rho(r+f)})+\frac{1}{2} \ln 2\pi$.}
\end{restatable}

\section{Computation and Behavior of $\rho(x,t)$}\label{NwSect}

\par In Sections 3 and 4, we made several assumptions regarding $\rho(x,t)$, which {we recall was defined as} the probability that a bid limit order placed at level ${S}_{0}-x$ at time $0$ is executed {before time $t$} during the first time period {when} the best bid price is at the level ${S}_{0}-x$, given that the latter event happens. The purpose of this section is to investigate the behaviour of $\rho(x,t)$, and verify the plausibility of our assumptions, {both theoretically and empirically.}

As stated in Remark \ref{CmptRho}, a reasonable model for $\rho(x,t)$ is given in Eq.~(\ref{ExprRho0}). This depends on the presumed order flow of cancellations at each level (which determines $P({N^{b,x}_{s}}=j)$), the presumed order flow of orders at the best bid and ask (which determines $\alpha_{t}(i,j)$), the distribution $f^{a}$ (which can be estimated from real LOB data; {see Figure \ref{Fai_distibution} below}), and the distribution $f_{\tau}(s|0<\tau<t)$, where $\tau$ is the first time that the best ask price hits the level $S_{0}-x$. In the spirit of the present work, in what follows, we assume that the best ask price  follows either the Brownian motion (BM) or the geometric Brownian motion (GBM). In that case, $f_{\tau}(s|0<\tau<t)$ can be explicitly computed (see the details in Appendix \ref{DtlsRho}). 

The following result shows that the first condition in {(\ref{Sec5:Conditions0b}) (respectively, (\ref{ConditionsGBM2})) is satisfied in our BM (respectively, GBM) setting.}
\begin{proposition}
	\label{Rhoxt:Diff:xINfty}
	Let $\rho(x,t)$ be defined as in Eq.~(\ref{ExprRho0}) (respectively, $\tilde{\rho}(y,t):=\rho(S_{0}-S_{0}e^{-y},t)$), and let the best ask price follows a {BM (respectively, a GBM)}.  Suppose that {$Q_{x}^{b}(0)=0$} for large enough $x$. Then, 
\begin{equation}\label{ADE0}
	\limsup_{x\to\infty}\frac{\partial \rho(x,t)}{\partial x}\leq 0, \quad \limsup_{y\to\infty}\frac{\partial \tilde{\rho}(y,t)}{\partial y}\leq{}0.
\end{equation}
\end{proposition}
The following result shows {that, under relatively mild assumptions,} the $\liminf$ in the second conditions of Eqs.~(\ref{Sec5:Conditions0b}) and (\ref{ConditionsGBM2}) {remains positive and gives explicit lower bounds.}
\begin{proposition}
	\label{Rhoxtx2:xInfty:General}
Let the best ask price follows a {BM or a GBM}. Also, let $\sigma_a^{i}$ and $\sigma_b^{\ell}$ be the time until the best ask and bid {queues} are  depleted when the initial sizes of the best bid and ask queues are $i$ and  $\ell$, respectively.   Suppose that {$Q_{x}^{b}(0)=0$} for large enough $x$. Then, when $\sigma_a^{i}$ and $\sigma_b^{\ell}$ are independent of each other and the density $g_b^{1}(s)$ of $\sigma_b^{1}$ is bounded and $\mathcal{C}^1$ in $[0,t]$, we have that
\begin{equation}\label{ADE0b}
 \liminf_{x\to\infty}\rho(x,t)x^2 \geq 2 g_b^{1}(0)\sigma^2 t^2,\quad
 \liminf_{y\to\infty}y^{2}\tilde{\rho}(y,t) \geq 2 g_b^{1}(0)\sigma^2 t^2 .
 \end{equation}
\end{proposition}

\begin{rem}\label{Rhoxtx2:xInfty:GeneralRem}
	From the result of Proposition~\ref{Rhoxtx2:xInfty:General},
	 the second condition in (\ref{Sec5:Conditions0b}) is met when the density $g_b^{1}$ is bounded and $g_b^{1}(0)> (-\mu)/(r+f)$.  In the Appendix \ref{DtlsRho}, we show that, under a Poissonian order flow, $g_b^1(0)=\mu_b+\theta_b$, the net depletion rate of the best bid queue (see the details in Appendix \ref{DtlsRho}). 
	 The second condition in (\ref{ConditionsGBM2}) is met whenever $g_b^{1}(0)>0$.
\end{rem}

\par We now turn to the plausibility of {the conditions} in (\ref{Sec5:Conditions0}) and (\ref{ConditionsGBM}), which is done numerically. For simplicity, we only present the analyses for the GBM, though the same conclusions hold for the BM. For the order flow, we assume one of the simplest (and yet relevant) settings, in which the arrival,  cancellation, and execution  of limit orders follow independent Poisson processes with respective intensity rates $\lambda_{\ell}$, $\theta_{\ell,k}$, and $\mu_{\ell}$, where $\ell$ is either $a$ or $b$, depending on whether the order is in {the} ask or bid side, and $k$ is the number of ticks away from the best bid or ask price (we refer to  Appendix \ref{DtlsRho} for the details).  The arrival rates $\lambda_{\ell}$, $\theta_{\ell,1}$, and $\mu_{\ell}$ are given in Table \ref{Table:Rates} below, while the initial LOB profile $k\to{}Q^b_{k\eps}(0)$ is taken as shown in the left panel of Figure \ref{rhoxt_DIFFQY}, which is consistent with the average depth profile of \cite{abergel2013mathematical} and \cite{cont2010stochastic}. The chosen values for $\theta_{\ell,k}$ ($k=2,3,\dots$) are borrowed from \cite{cont2010stochastic}. Based on the just stated assumptions, we compute $\rho(x,t)$ with parameter values estimated from real LOB data. The details of the computation of $\rho(x,t)$ are given in the Appendix \ref{DtlsRho}. 

The graphs of $i\to{\rho}(\eps i,t)$ ($i$ represents number of ticks and $\eps=0.01$) for $t=30$, $60$, and $90$ sec. are shown in the right panel of Figure~\ref{rhoxt_DIFFQY}. As shown therein, $\rho(\eps,t) < \rho(2\eps,t)$ and, even,  $\rho(2\eps,t) < \rho(3\eps,t)$, which justifies our assumption $\partial_y \tilde{\rho}(0^{+},t) >0$.
The conditions in (\ref{ADE0}), which are already proved in Proposition~\ref{Rhoxt:Diff:xINfty}, are also evident from Figure~\ref{rhoxt_DIFFQY}. Regarding the conditions in (\ref{ADE0b}), the left panel of Figure~\ref{rhoxt_timesx} suggests that $y\to{}e^{y}y^{2}\tilde{\rho}(y,t)$ drifts toward $\infty$ and, thus, the second condition in (\ref{ADE0b}) is reasonable. Let us remark that, as $t$ gets larger, $\rho(x,t)$ get flatter, which justifies to take $\rho(x,t)$ approximately constant in $x$ for large values of $t$ as we did in our large horizon asymptotics of Sections 3 and 4.
\begin{figure}
	\includegraphics[width=0.5\textwidth]{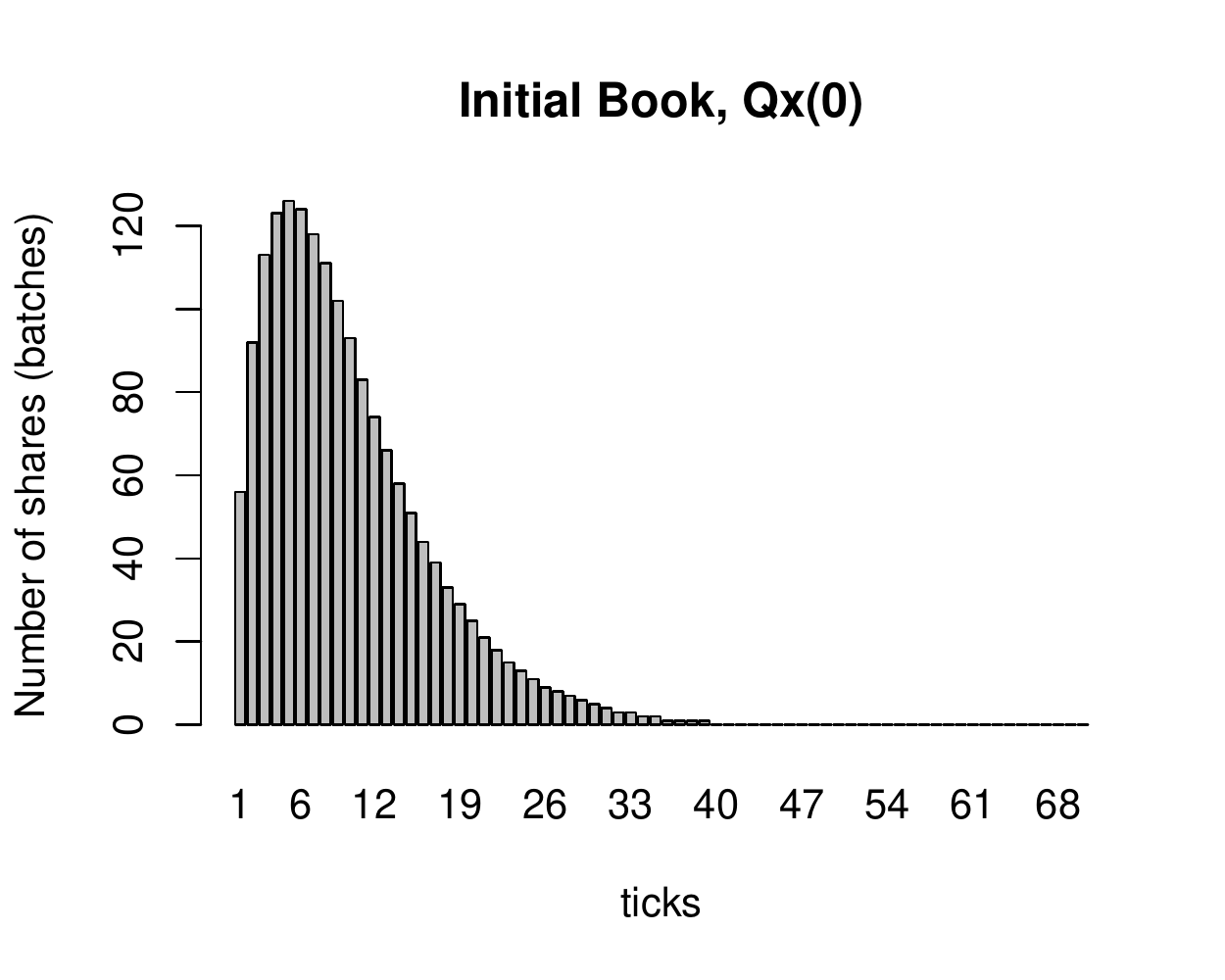}
	\includegraphics[width=0.5\textwidth]{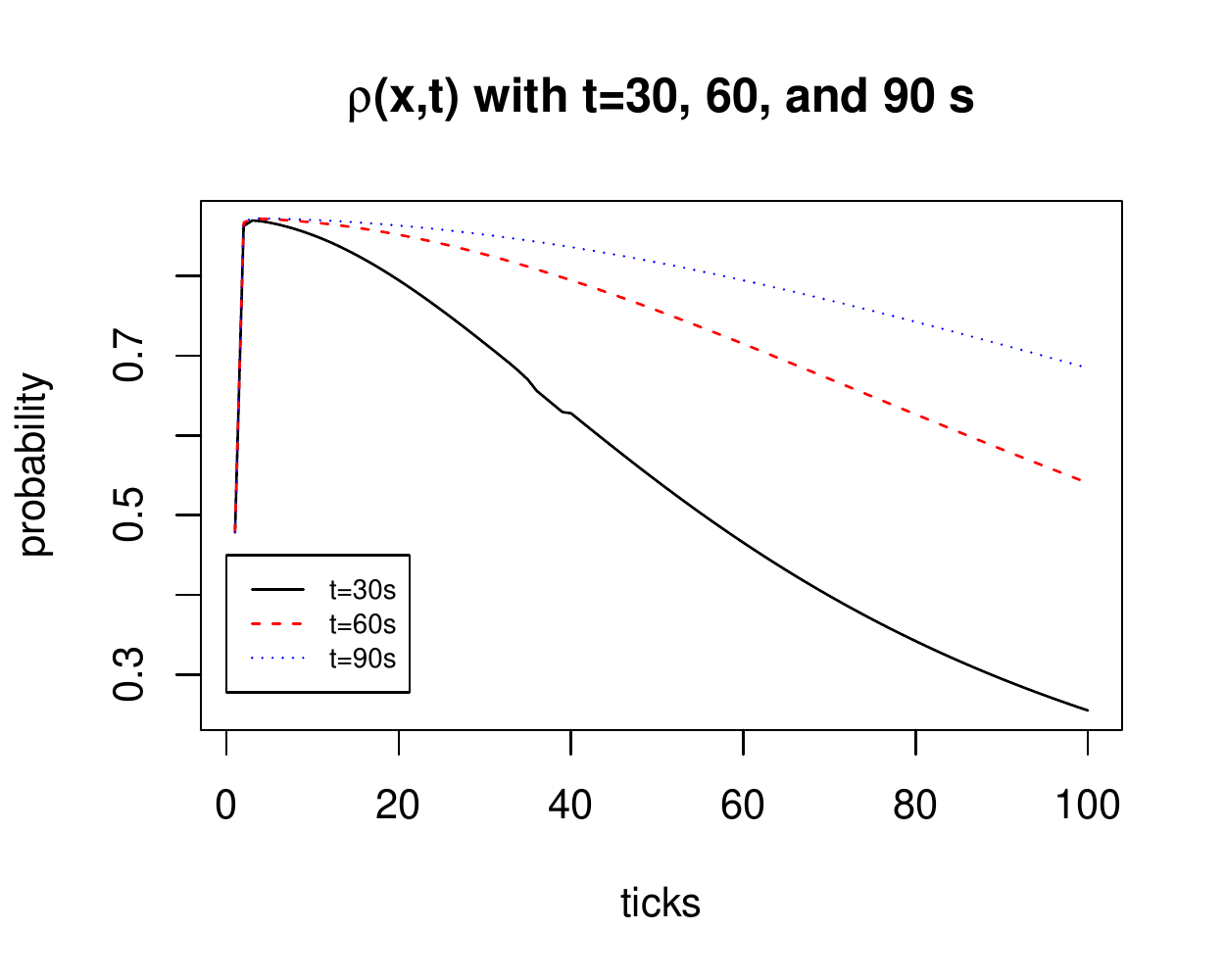}
	\caption{Left Panel: Initial LOB profile $i\to{}Q_{i\eps}^b(0)$; Right Panel: Graphs of $i\to{}{\rho}(\eps i,t)$ with $\eps=0.01$ and $t=30$ sec (black line), $t=60$ sec (red dashed line), and $t=90$ sec (blue dotted line). }
\label{rhoxt_DIFFQY}
\end{figure}

\begin{figure}
	\includegraphics[width=0.5\textwidth]{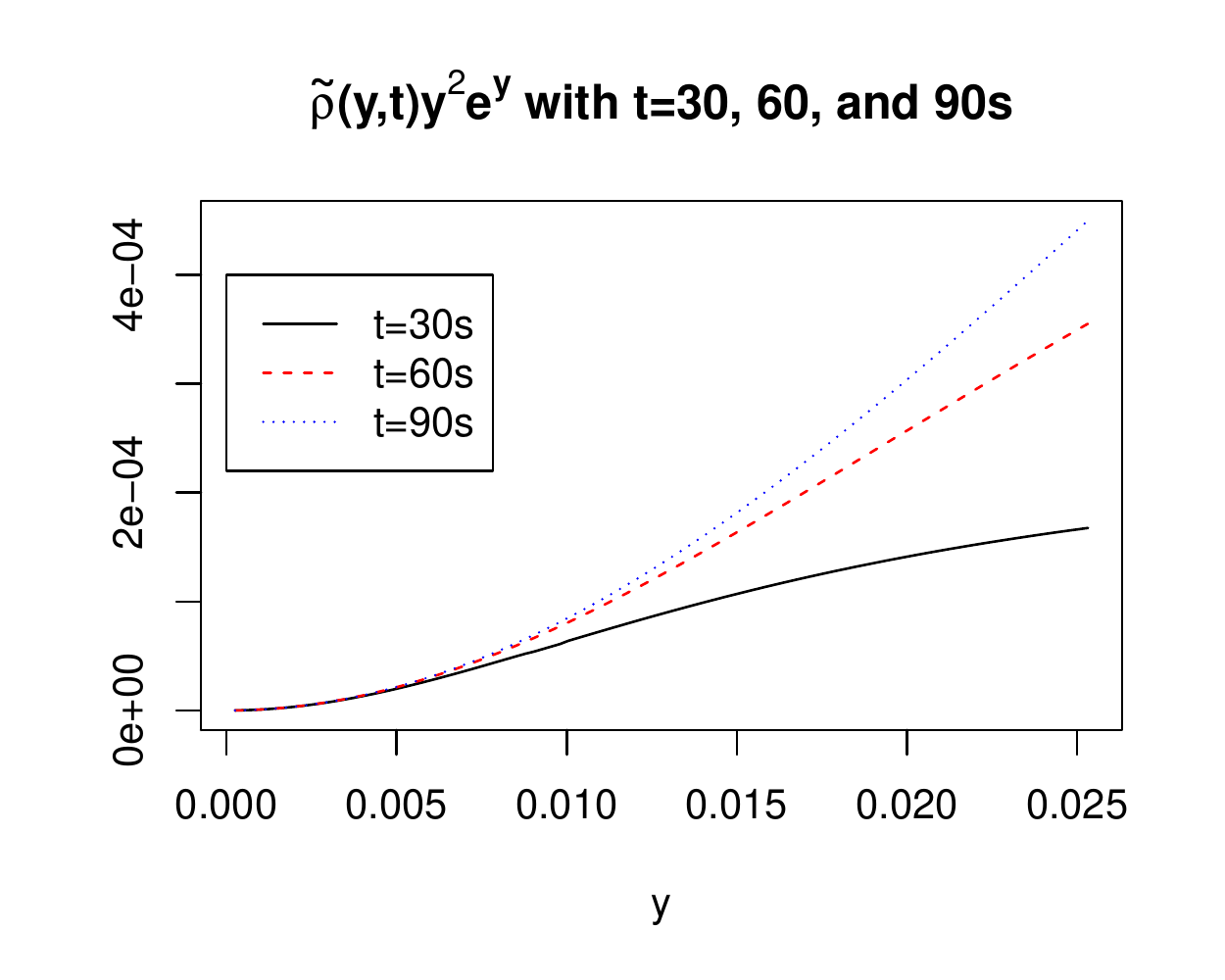}
	\includegraphics[width=0.5\textwidth]{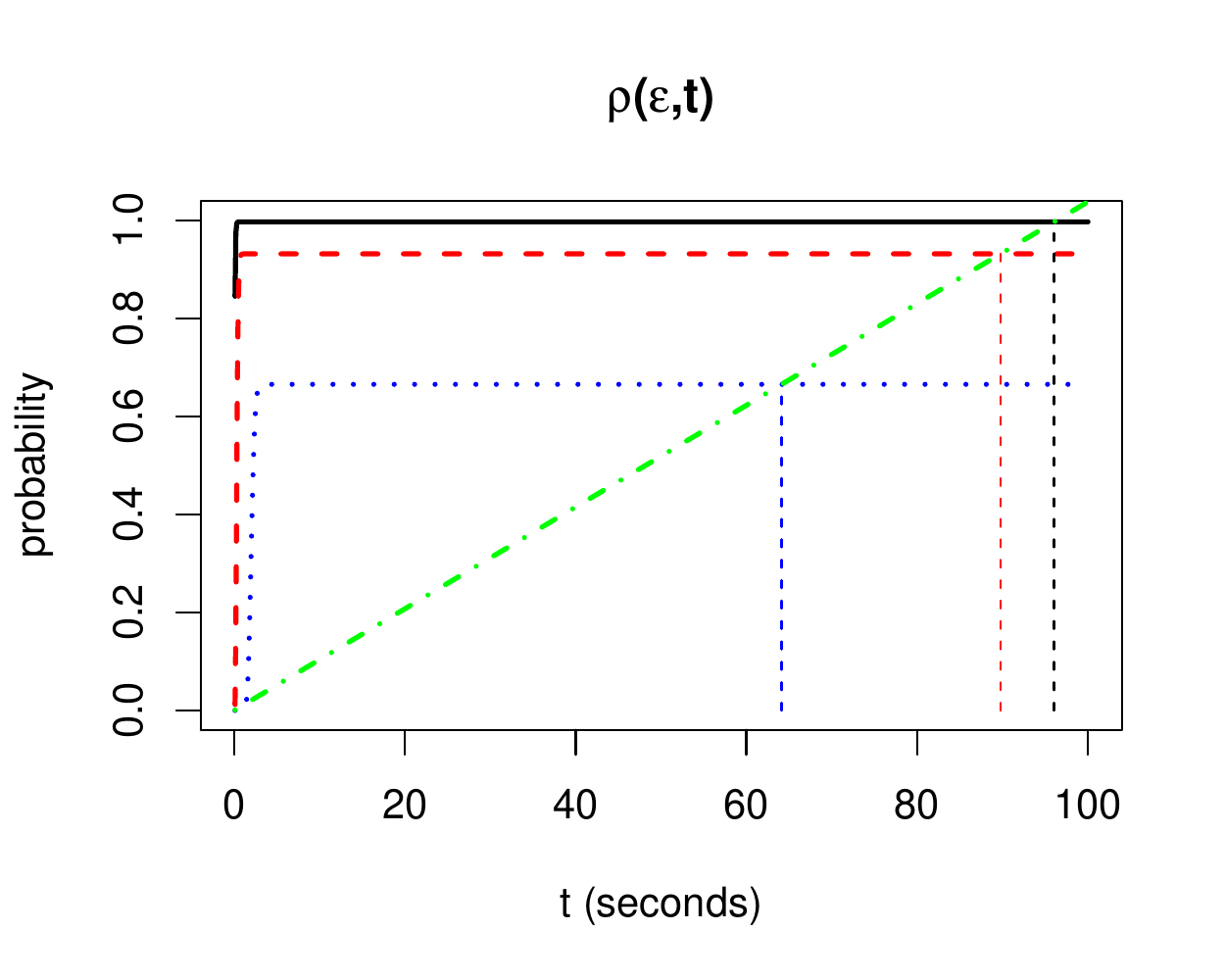}
	\caption{Left Panel: Graph of $y\to{}e^{y}\tilde{\rho}(y,t)y^2$ with $t=30$ sec (black line), $t=60$ sec (red dashed line), $t=90$ sec (blue dotted line). Right Panel: $\rho(0^+,t)=\tilde{\rho}(0^+,t)$ with $Q_{\eps}^{a}(0)=6$ and $Q_{\eps}^{b}(0)=1$ (solid black), $Q_{\eps}^{b}(0)=6$ (dashed red), and $Q_{\eps}^{b}(0)=38$ (dotted blue). The (dot-dashed) green line has a slope of $2S_0 |\mu|/(r+f)$ so that $t^*_0\approx{}64,\; 90,\text{ and }96$ seconds when $Q_{\eps}^{b}(0)=38$,  $6$, and  $1$, respectively.}
	\label{rhoxt_timesx}
\end{figure}

\par The next major assumption we made in Sections 3 and 4 is that $t \mapsto \rho(0^+,t)$ converges to its limiting value $\rho(0^+):=\rho(0^+,\infty)$, as $t\to \infty$, fast enough so that we can use $\rho(0^+)$ instead of $\rho(0^+,t)$ for the estimation of $t_0$ and $t^*_0$. To show that this assumption is indeed plausible when $\mu<0$, in the right panel of Figure~\ref{rhoxt_timesx}, we show the graphs of $t\to{}\rho(\eps,t)$ (a proxy for $\rho(0^{+},t)$) for different initial values of the best bid and ask queues. Here, $Q_{1}^{a}(0)=6=\sum_{i}f^{a}(i)i$ (i.e., the average size of the best ask queue after a price drop) and $Q_{1}^{b}(0)=38$ is the average size of the best bid after a price drop obtained from real LOB data (see Appendix  \ref{DtlsRho} for details). As can be seen in the right panel of Figure~\ref{rhoxt_timesx}, the convergence of $\rho(0^+,t)$ to its limit $\rho(0^+)$ happens near instantaneously, in a matter of just a few seconds, thus validating our assumption $\rho(0^{+},t)\equiv \rho(0^{+})$, for reasonable investment horizons $t$. As it turns out, under the Poisson order flow setting, the limiting value $\lim_{t\to{}\infty}\rho(0^+,t)$ can be computed explicitly  (see Eq.~(\ref{LmtRho0}) below). 
In the right panel of Figure~\ref{rhoxt_timesx}, we also show the approximation $\bar{t}_{0}$ of $t^*_0$ such that $2|\mu|S_0 t/(r+f)>\rho(0^{+},t)$, for any $t>\bar{t}_0$. Since $\rho(0^{+},t)$ is constant after few seconds, $t^*_0 \approx {\rho({ 0^{+}})(r+f)}/{2|\mu|S_0}$, which take the values of $64$, $90$, and $96$ seconds when $Q_{\eps}^{b}(0)=38$,  $6$, and  $1$, respectively.

Now, we discuss the validity of the assumption $\partial_{t}\partial_{x}\rho(0^{+},t)>0$, which is used in Theorems \ref{prop:t0limit} and \ref{ttot0} to guarantee that $t_{0}^{*}$ is the only critical value of $\partial_{y} C(0^{+},t)$ and for the approximation of $y^{*}(t)$ as $t\searrow{}t_{0}^{*}$ to hold. We can take $D(t):=\rho(2\eps,t)-\rho(\eps,t)$ as a proxy for $\eps\partial_{x}\rho(0^{+},t)$. In Figure~\ref{Rxt_Diffeps}, we plot $D(t)$ as a function of $t$ for different values of $Q^{a}_{\eps}(0)$, $Q^{b}_{\eps}(0)$, and $Q^{b}_{2\eps}(0)$. As shown there, $D(t)$ is increasing in $t$ and positive after just a few seconds.

\begin{figure}
	\includegraphics[width=0.5\textwidth]{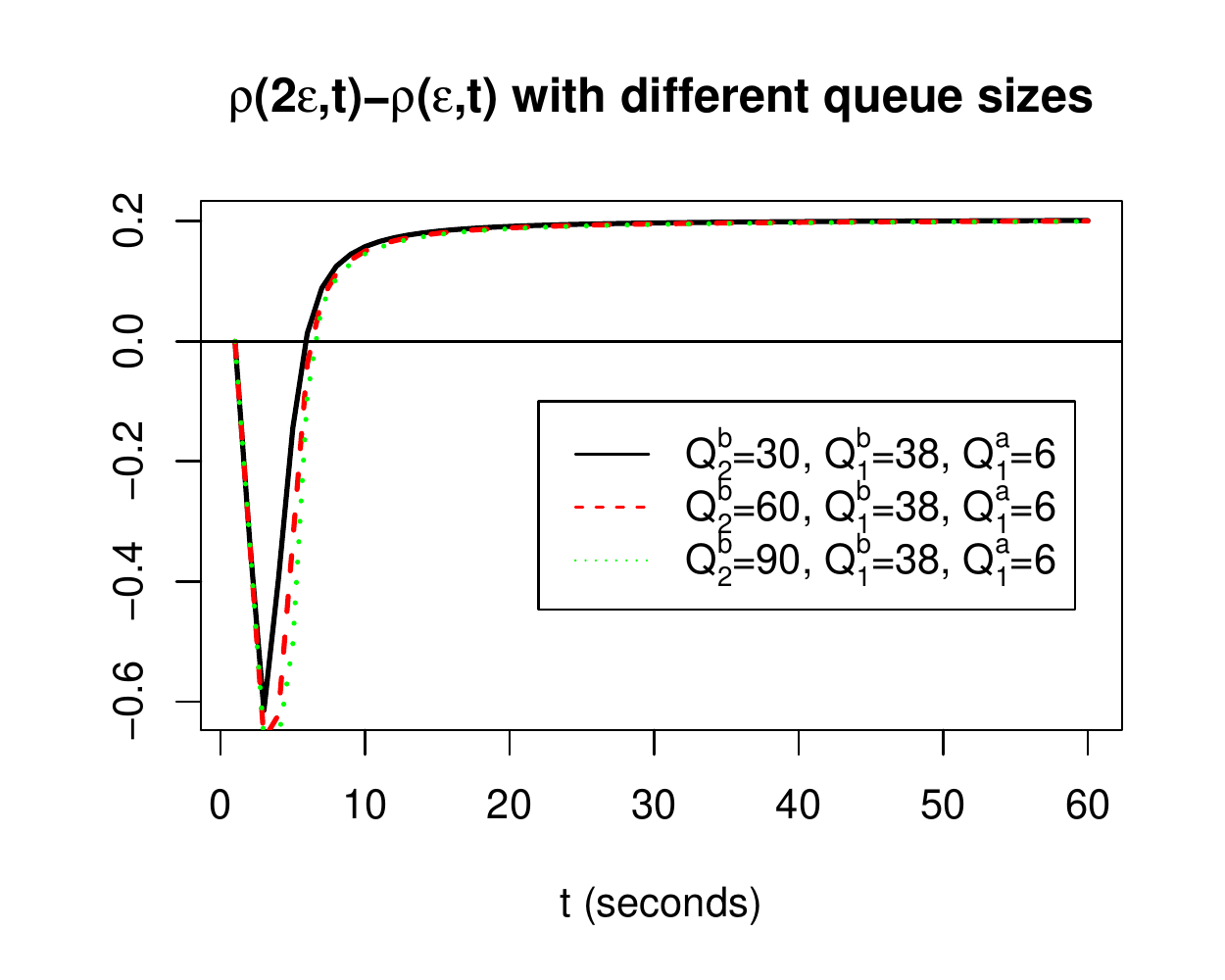}
	\includegraphics[width=0.5\textwidth]{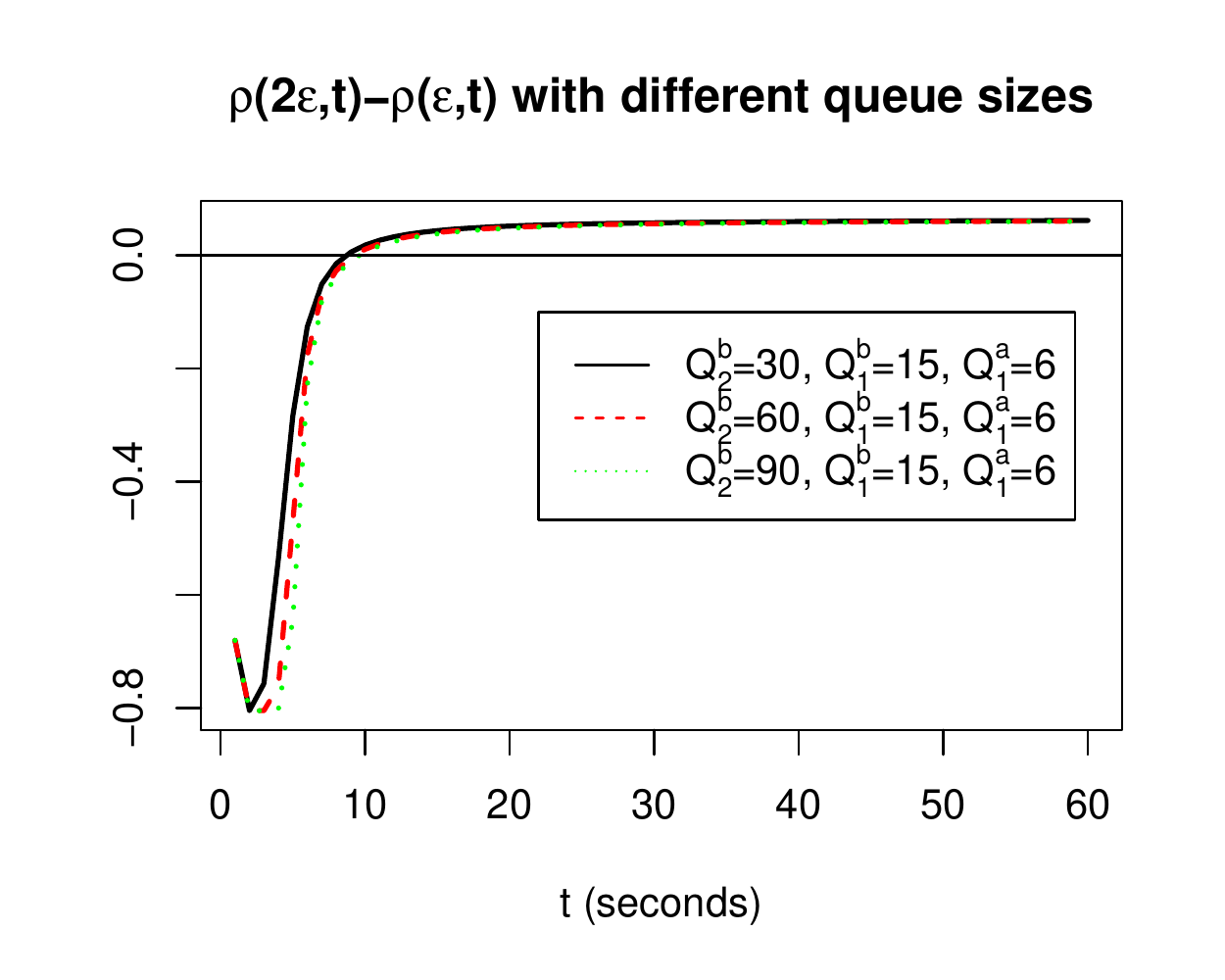}
	\caption{Graphs of $t\to\rho(2\eps,t)-\rho(\eps,t)$ under the Black-Scholes Model for different values of $q^{a}_{1}=Q^{a}_{\eps}(0)$, $q^{b}_{1}=Q^{b}_{\eps}(0)$, and $q^{b}_{2}=Q^{b}_{2\eps}(0)$.}\label{Fig: rhoxt}
	\label{Rxt_Diffeps}
\end{figure}

\par Finally, let us discuss the approximation of $\rho(x,t)$. While the computation of $\rho(x,t)$ is somewhat complicated, Figure~\ref{Rxt_DiffQ} demonstrates that $\rho(x,t)$ can be well approximated with $\rho(x,\infty)$, which we conjecture is given by $\rho(x,\infty)=\sum_{i}f^{a}(i)\alpha_{\infty}(i,1)$. In Figure~\ref{Rxt_DiffQ}, we show the numerical computation of $\rho(x,t)$ in a Bachelier model (left) and a Black-Scholes model (right) with $x=0.1$ and $Q_x^b(0)=0,1,10,38,50$ or $100$ batches. As shown in Figure~\ref{Rxt_DiffQ}, $\rho(x,t)$ is approximately constant with respect to $t$ after a {short time period}.

\begin{figure}
	\includegraphics[width=0.5\textwidth]{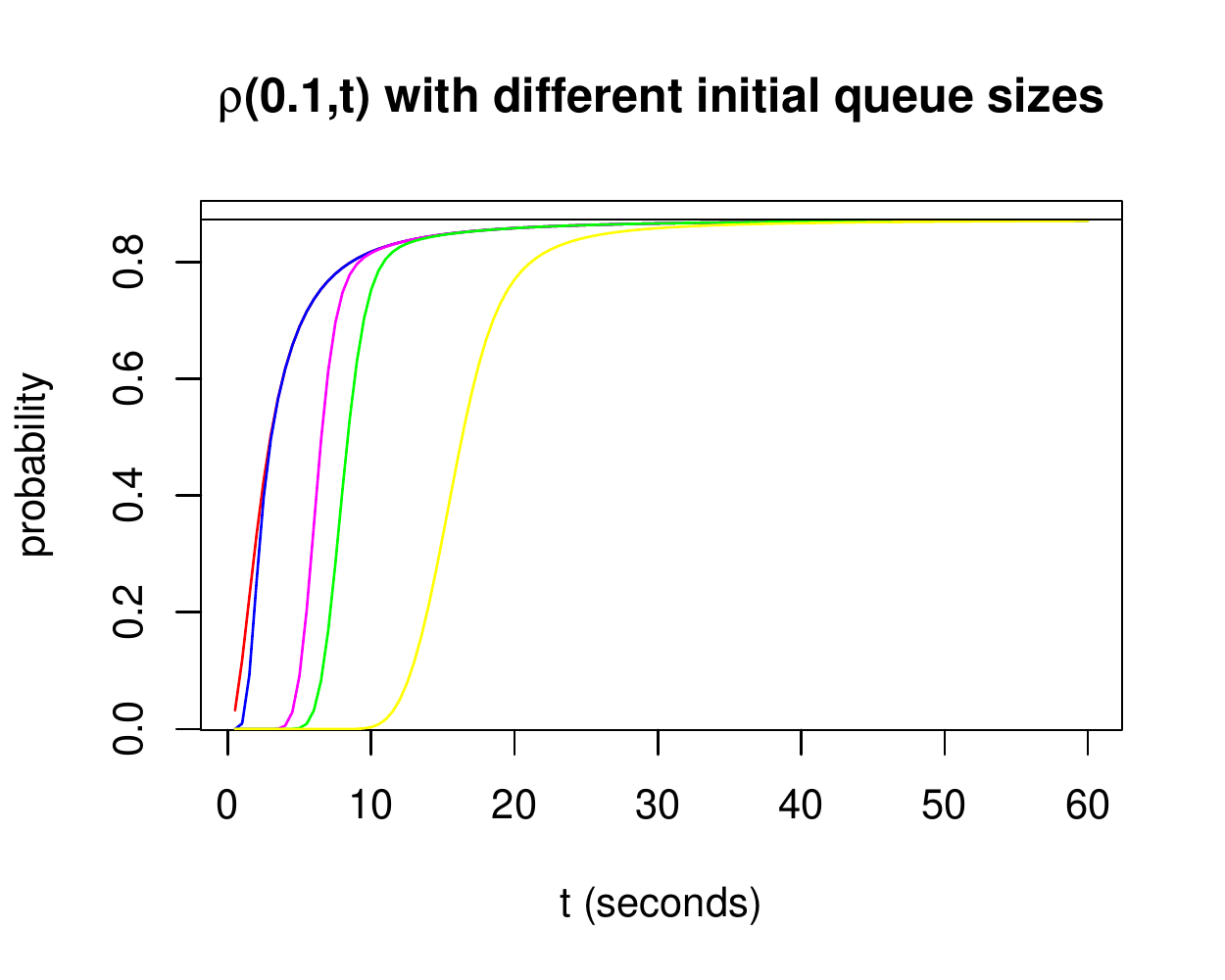}
	\includegraphics[width=0.5\textwidth]{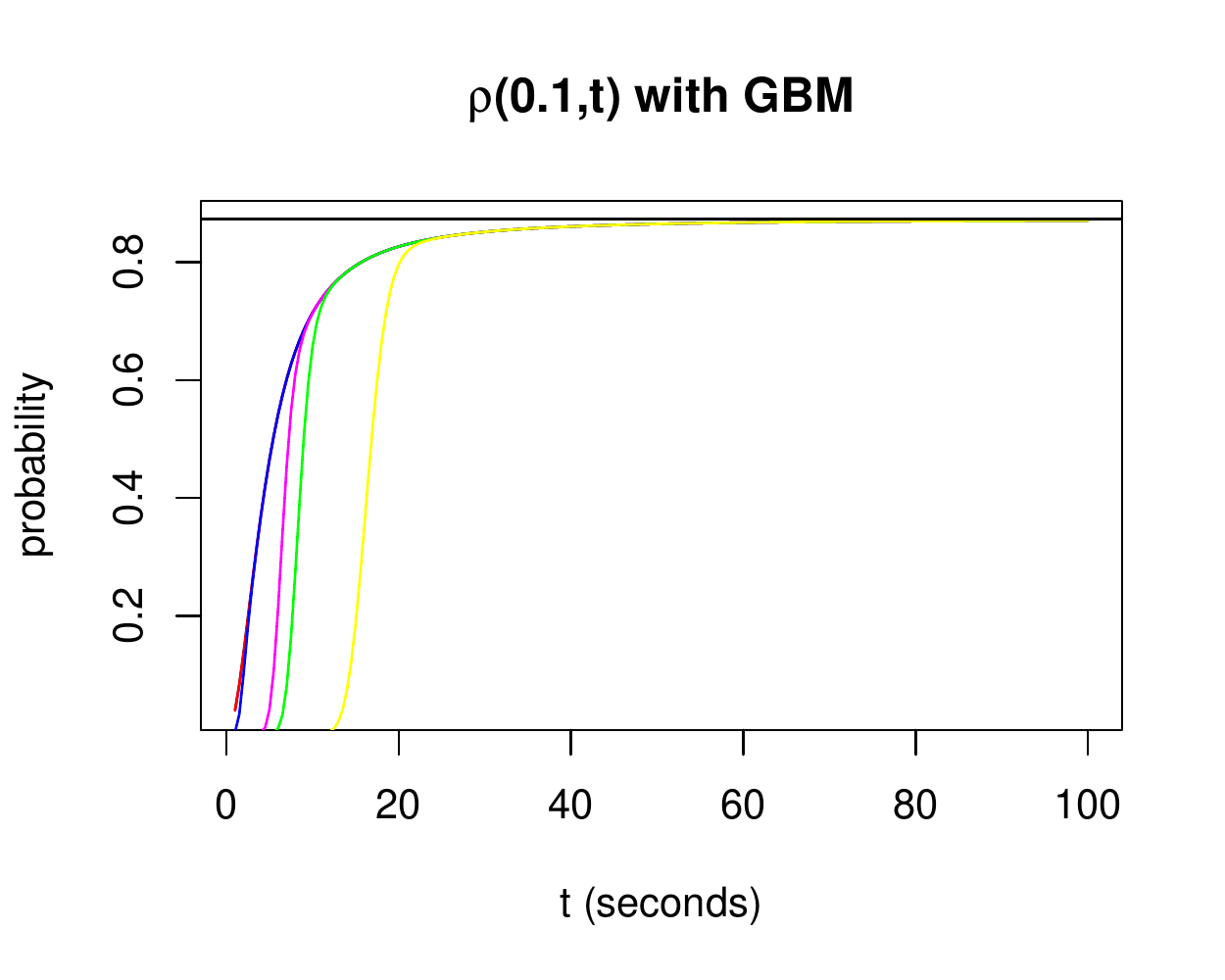}
	\caption{$t\to\rho(0.1,t)$ under the Bachelier model (left) and the Black-Scholes Model (right) for $Q_x^b(0)=0$ (black), $Q_x^b(0)=1$ (red), $Q_x^b(0)=10$ (blue), $Q_x^b(0)=38$ (magenta), $Q_x^b(0)=50$ (green), and $Q_x^b(0)=100$ (yellow). Queue size is in batches (each of size 100 shares). Black line is almost overlapping with red line. Black flat line is $\sum_{i}f_{a}(i)\alpha_{\infty}(i,1)= 0.87$.}\label{Fig: rhoxt}
	\label{Rxt_DiffQ}
\end{figure}

\section{{Conclusions and Future Work}}
{The question of whether to place a market order or a limit order {and,} in the latter case, the position in the LOB at which to place the limit order has gained much recent attention. In this paper, we pose this question as an optimal order placement problem with asset prices following certain diffusive dynamics. The effects of the LOB queue are captured through a certain time-dependent execution probability. Our model and the ensuing analysis lead to a number of important insights: (i) the existence of a threshold-horizon length beyond which the presence of a non-trivial optimal placement is guaranteed; (ii) the characterization of the threshold-horizon length under a certain asymptotic regime involving the rebate and the trading fee; (iii) the characterization and approximation of the optimal placement of the limit order as the horizon length approaches the threshold-horizon length, and (iv) the behaviour of the optimal placement under different asymptotic regimes involving the horizon length and the volatility. Importantly, these insights on optimal placement rely on assumptions that seem to hold widely in real markets, as seen through data and numerical justification.

Numerous other important contexts in the LOB context, and closely related to what we consider in this paper, seem largely under-studied. For instance, the nature of time-dependent optimal order placement, the effect of multiple correlated assets, the consideration of large limit orders, and the presence of diffusive models whose parameters are not known (but can be estimated), are all interesting LOB contexts needing further development. The models considered in this paper and the nature of our analyses could inform such development.}

\appendix
\section{Proofs: Brownian motion}
{For future reference, we introduce the following notation:
\begin{equation}\label{TNFBM}
	\alpha_t(x):= \frac{x+\mu t}{\sigma\sqrt{t}}, \quad \beta_t(x):=\frac{-x+\mu t}{\sigma\sqrt{t}},\quad 
	a_{t}:=\frac{\mu \sqrt{t}}{\sigma},\quad c:=r+f.
\end{equation}
When there is no confusion, we will often omit the dependence on $t$ and/or $x$ in $\alpha$, $\beta$, {$\rho$}, and $a$. Let us also remark that
\begin{align}
\label{bm:1stder:L1}
		\frac{\partial C}{\partial x}&=
	{2\phi \left( {\alpha_t} \right) \frac{c \rho+\mu t}{\sigma\sqrt{t}} +\frac{2N\left( {\beta_t} \right)}{\sigma^{2}} e^{\frac{-2x \mu }{\sigma^2}} \left[ -\mu (x-\mu t)+\mu{c}{\rho}+\frac{\sigma^{2}}{2} \right] 
	+N({\alpha_t} ) -1}\\
	&\quad -{c{\partial_{x}\rho(x,t)}\left(N(-\alpha_{t})+e^{-\frac{2x\mu}{\sigma^{2}}}N(\beta_{t})\right)}
	\label{bm:1stder:L1b}
	\\
\nonumber
		\frac{\partial^2 C}{\partial x^2}&=
	{- \frac{\phi \left( {\alpha_t} \right)}{\sigma^{3}t\sqrt{t}}
	\left[ 2{c}{\rho}x + 4 \mu t ({c}{\rho}+ \mu t) \right]+\frac{4\mu N\left( {\beta_t} \right)}{\sigma^{4}}e^{\frac{-2x \mu }{\sigma^2}} 
	\left[ \mu (x-\mu t) - \mu{c}{\rho} -\sigma^{2}  \right]}\\
\nonumber
	&\quad-{{c\,{\partial_{x}^{2}\rho(x,t)}\left(N(-\alpha_{t})+e^{-\frac{2x\mu}{\sigma^{2}}}N(\beta_{t})\right)+
	\frac{4c\,{\partial_{x}\rho(x,t)}}{\sigma^{2}\sqrt{t}}\left(\sigma\phi(\alpha_{t})+\mu\sqrt{t}e^{-\frac{2x\mu}{\sigma^{2}}}N(\beta_{t})\right).}}
\end{align}
{For ease of notation, we will plug in $x=0$ instead of $x=0^+$ when evaluating the limit of the functions $C(x,t)$ and $\rho(x,t)$ and their derivatives as $x\searrow{}0$.} 
 Finally, the following well-known {inequalities are} often used in the proofs:}
\begin{equation}
\label{imp.inequ}
	{\phi(z)\left(\frac{1}{z}-\frac{1}{z^{3}}\right)\leq{}\phi (z) \frac{z}{z^2 +1} \leq N(-z) \leq \phi (z) \frac{1}{z},\quad z \geq 0.}
\end{equation}

\begin{proof}[\textbf{Proof of Lemma~\ref{BMCostFunc}}]{Without loss of generality we assume that $S_{0}=0$. The result then follows from the formula for the joint distribution of $Y_t$ and $S_t$  (cf. \cite*[Section 3.2]{jeanblanc2009mathematical}):
\begin{equation}\label{JntCndSY}
	P(S_t>z, Y_t > -x)= N\Big(\frac{-z+\mu t}{{\sigma}\sqrt{t}}\Big)- e^{\frac{-2x \mu }{{\sigma}^2}}  N\Big(\frac{-z-2x+\mu t}{{\sigma} \sqrt{t}} \Big),\quad x>0, z\geq{}-x.
\end{equation}
Indeed, from the previous formula, we directly have that
\begin{align*}
	P(Y_t \leq -x)&= N\Big(\frac{-x-\mu t}{{\sigma}\sqrt{t}}\Big)+e^{\frac{-2x \mu }{{\sigma}^2} }N\Big(\frac{-x+\mu t}{{\sigma}\sqrt{t}}\Big)\\
	P(S_t \in dz, Y_t>-x) & = \frac{1}{{\sigma}\sqrt{t}} \phi\Big(\frac{-z+\mu t}{{\sigma}\sqrt{t}}\Big)
-\frac{1}{{\sigma}\sqrt{t}}e^{\frac{-2x \mu }{{\sigma}^2}} \phi\Big(\frac{-z-2x+\mu t}{{\sigma}\sqrt{t}}\Big),
\end{align*}
which can be used to find $E[S_t|Y_t > -x]P(Y_t>-x)=\int_{-x}^{\infty} z P(S_t \in dz, Y_t>-x) dz$.}
\end{proof}

\begin{proof}[\textbf{Proof of Theorem~\ref{mainresult}}]
It is clear that the term in (\ref{bm:1stder:L1b}) is nonnegative and, thus, for $x\to C(x,t)$ to be increasing, we only need to show that the expression in (\ref{bm:1stder:L1}), that we denote $D(x,t)$,  is always positive. The derivative of this expression is given by 
\begin{align}
\label{bm:2stder:L1c}
	\frac{\partial D}{\partial x}&=- \frac{\phi \left( {\alpha_t} \right)}{\sigma^{3}t\sqrt{t}}
	\left[ 2{c}{\rho}x + 4 \mu t ({c}{\rho}+ \mu t) \right]+\frac{4\mu N\left( {\beta_t} \right)}{\sigma^{4}}e^{\frac{-2x \mu }{\sigma^2}} 
	\left[ \mu (x-\mu t) - \mu{c}{\rho} -\sigma^{2}  \right]\\
\label{bm:2stder:L1d}	
	&\quad+
	\frac{2c\,{\partial_{x}\rho(x,t)}}{\sigma^{2}\sqrt{t}}\left(\sigma\phi(\alpha)+\mu\sqrt{t}e^{-\frac{2x\mu}{\sigma^{2}}}N(\beta)\right).
\end{align}
When {$x\to\rho(x,t)$} is nonincreasing, the term in (\ref{bm:2stder:L1d}) is nonpositive because of the third inequality of (\ref{imp.inequ}). 
Now, when $x \in [0, \mu t]$, all terms in (\ref{bm:2stder:L1c}) are negative, while, for $x>\mu t$, we can apply the {last inequality of (\ref{imp.inequ}) 
to the term associated with $\mu(x-\mu t)$} in (\ref{bm:2stder:L1c}) to get
\begin{equation*}
{\frac{\partial D}{\partial x}} \leq  
	\frac{\phi \left( \alpha_t \right)}{\sigma\sqrt{t}}
	\left[ -\frac{2c\rho x}{\sigma^2 t} - \frac{4 \mu c\rho}{\sigma^2} \right]
	 +N\left( \beta_t \right) e^{\frac{-2x \mu }{\sigma^2}} 
	\left[- \frac{4 \mu^2 c\rho }{\sigma^4}-\frac{4\mu}{\sigma^2}  \right]<0.
\end{equation*}
We then deduce that $x\to D(x,t)$ is decreasing. But, {also, $D(0^{+},t)>0$ and 
$D(x,t) \to 0$, as $x \to \infty$. These facts imply that, for each $t>0$, 
$D(x,t)$ is strictly positive for $x>0$ and, thus, finally $C(x,t)$ is increasing in $x$ for each $t$.} 

\noindent
{To prove the second assertion}, 
{assume that $x>-\mu t$ and 
apply the {appropriate inequality of (\ref{imp.inequ}) to each} $N(\cdot)$ term in  (\ref{bm:1stder:L1}) to get}
\begin{align}\label{DL1}
		\frac{\partial C}{\partial x}&\geq{} \phi \left(  \alpha_t \right) 
		\left( \frac{2({c\rho}+\mu t)}{\sigma \sqrt{t}} - \frac{2 \mu {c\rho} }{\beta_t \sigma^2} - \frac{1}{\alpha_t} - \frac{ (-2 \mu x+ 2\mu^2 t +\sigma^2)}{\sigma^2}\left(\frac{1}{\beta}-\frac{1}{\beta^{3}}\right) \right)\\
		\label{DL2}
		&\quad -{{\partial_{x}\rho(x,t)}c\left(N(-\alpha)+e^{-\frac{2x\mu}{\sigma^{2}}}N(\beta)\right).}
\end{align}
Because of the second condition in  (\ref{Sec5:Conditions0}), the expression in (\ref{DL2}) is positive for large enough $x$. After some simplifications, it is not hard to see that the expression inside the parentheses in (\ref{DL1}) is of the form $2c\rho(x,t)/\sigma\sqrt{t}+ 4\mu \sigma t \sqrt{t}/x^2+ O(1/x^{3})$ and, thus, it becomes positive for large enough $x$ in light of the last condition in (\ref{Sec5:Conditions0}). Therefore, $\partial_{x}C(x,t)$ is positive for large enough $x$ and, hence, it is not possible that $C(\infty,t)\leq{}C(x,t)$, for all $x$.
To conclude, 
note that
\begin{align}
\label{eq:bm:1st:x0}
	\frac{\partial C}{\partial x} (0,t) &= \left( \phi \left(a_t \right)+ N\left(a_t \right) a_t \right) \frac{2}{\sigma}\left( \frac{c{\rho}(0,t)}{\sqrt{t}}+ \mu\sqrt{t}\right) + 2 N \left(a_t \right) -1-{{\partial_{x}\rho(0,t)}c}\\
\label{eq:bm:1st:x0b}	
&<{}2	\phi \left(a_t \right)\frac{c\rho({0,t})+2\mu t}{\sigma\sqrt{t}}+2\mu N(a_{t})\frac{c\rho(0,t)+\mu t}{\sigma^{2}} -{{\partial_{x}\rho(0,t)}c},
\end{align}
where the inequality in (\ref{eq:bm:1st:x0b}) follows from the inequality $2N(a_{t})-1<2\phi(a_{t})a_{t}$. 
Now, we claim that $\partial_{x} C(0,t)<0$ if $c\rho(0,t)+2\mu t<0$. Indeed, if $c\rho(0,t)+\mu t\leq{}0$, this follows directly from the expression in (\ref{eq:bm:1st:x0}), while if $c\rho(0,t)+\mu t>0$, this is a consequence of the inequality in (\ref{eq:bm:1st:x0b}). Finally, since clearly $c\rho(0,t)+2\mu t<0$ holds for $t$ large enough, we conclude that there exists $t_{0}\geq{}0$ such that for all $t>t_{0}$, $\partial_{x} C(0,t)<0$. The latter condition implies that it is impossible to have $C(0,t)<C(x,t)$ for all $x>0$ and, thus, the existence of $x^{*}(t)$ as stated in the statement of theorem holds. 
\end{proof}

\begin{proof}[\textbf{Proof of Theorem~\ref{thm2b}}]
Recall from the proof of Theorem~\ref{mainresult} that $\partial_{x} C(0,t)<0$, for $t$ large enough. So, for the first assertion of the theorem, we only need to show that $\partial_{x} C(0,t)>0$, for $t$ small enough. But this is clear from (\ref{eq:bm:1st:x0}) since $\rho(0,t)/\sqrt{t}\to+\infty$, as $t\to{}0$, under our assumptions that $\rho(0,t)\equiv{}\rho(0)\in (0,1]$, for all $t$.
We then conclude that $t_0$ can be chosen so that $\partial_{x} C(0,t_0)=0$. For the second assertion, let us start by noting that $t_0 \to 0$, as $(r+f) \to 0$ because of the upper bound of Corollary \ref{Crl1}. 
Next, expanding $\phi \left(a_t\right)$ and $N\left(a_t\right)$ as powers of $\sqrt{t}$ and some simplifications, we have
\begin{align}\nonumber
\frac{\partial C}{\partial x}(0,t)=
c \left( \frac{2\rho(0^{+}) \phi(0)}{\sigma}\frac{1}{\sqrt{t}}+\frac{\rho(0^+)\mu}{\sigma^2} - \partial_{x}\rho(0^+,t)+\frac{\phi(0)\rho(0^+)\mu^2}{\sigma^2}\sqrt{t}\right)
+\frac{4\mu\phi(0)}{\sigma}\sqrt{t} +\mathcal{E}(t)
\end{align}
where
\[
\mathcal{E}(t)=\frac{\mu^2}{\sigma^2} t-\frac{\rho(0^+)\mu^4\phi(0)c}{12\sigma^5}   t^{1.5} +\mathcal{O}(t^2)
\]
Now, plugging $t=t_0$, recalling that 
$\partial_{x} C (0,t_0)=0$, we have 
\begin{equation}\nonumber
0=1+\frac{2\mu}{\rho(0^+)}\frac{t_0}{c} +\left(\frac{\rho(0^+)\mu/\sigma - \partial_{x}\rho(0^+,t_{0})\sigma}{2\rho(0^+)\phi(0)} \right)\sqrt{t_0}  +\frac{\mu^2}{2\sigma}t_{0}+\frac{\mathcal{E}(t_0)\sigma\sqrt{t_0}}{2\rho(0^+)\phi(0)c}
\end{equation}
Since, by Corollary  \ref{Crl1},  $t_{0}\to{}0$,  as $c\to{} 0$, and $t_{0}/c\leq{}2/|\mu|$, the we have that the last three terms on the right hand side of the previous equation converge to $0$ and, hence, we must have that 
\[
	1+\frac{2\mu}{\rho(0^+)}\frac{t_0}{c} \to{}0,
\]
which implies the second assertion. For the last assertion, note that 
\begin{align}
\label{bm:2ndMixedDer:x0}
	\frac{\partial^{2} C}{\partial t\partial x}(0,t) = -\phi \left(a_{t} \right) \frac{{c{\rho(0,t)}}}{\sigma t\sqrt{t}} +
	\frac{2\mu}{\sigma\sqrt{t}}\left(\phi\left(a_{t}\right)+a_{t}N\left(a_{t}\right) \right) -c\partial_{t}\partial_{x}\rho(0,t),
\end{align}
which is negative because of our assumption $\partial_{t}\partial_{x}\rho(0,t)\geq{}0$ and the last inequality in (\ref{imp.inequ}).
\end{proof}

\begin{proof}[\textbf{Proof of Theorem~\ref{BM:xt:ttot0}}]
\par We will use the {mean value theorem} to show the behavior of the optimal placement solution, $x^*(t)$, when $t$ is close to $t_0$. {To this end, 
the} following conditions are necessary: $\partial^2_{x} C$ {needs to be positive at $(0,t_{0})$}, $x^*(t) \to 0$ as $t \to t_0$, and $C$ should be $C^{2}$ in a neighborhood of $(0,t_{0})$. First, let us recall from the proof of Theorem~\ref{thm2b} that $t_0$ satisfies $\partial_{x} C(0^{+},t_0)=0$. Then, using (\ref{eq:bm:1st:x0}), we can find an expression for $\phi \left(a_{t_{0}} \right)+ N\left(a_{t_{0}} \right) a_{t_{0}}$, which can be substituted into 
\begin{align*}
	{\frac{\partial^2 C}{\partial x^2}(0,t)=-\left( \phi \left(a_t \right)+ N\left(a_t\right)a_t \right)\left(\frac{4\mu(c\rho(0,t)+\mu t)}{\sigma^3\sqrt{t}}-\frac{4c\partial_{x}\rho(0,t)}{\sigma\sqrt{t}}\right)-N\left(a_t\right)\frac{4\mu}{\sigma^2}-c\partial_{x}^{2}\rho(0,t)}
\end{align*}
to get
\begin{equation}
\label{bm:2nd:x0}
{\frac{\partial^2 C}{\partial x^2}(0,t_0)= \frac{2|\mu| }{\sigma^2}\left(1+\partial_{x}\rho(0,t_{0})c\right)+2c\partial_{x}\rho(0,t_{0})\frac{1-2N(a_{t})+\partial_{x}\rho(0,t_{0})}{c\rho(0,t_{0})+\mu t_{0}}-c\partial^{2}_{x}\rho(0,t_{0}).}
\end{equation}
The last expression is positive since $\partial_{x}\rho(0,t_{0})>0$, $\partial_{x}^{2}\rho(0,t_{0})<0$, and  $t_{0}<c\rho(0^{+})/2|\mu|$ as a consequence of Corollary \ref{Crl1}.
{Next, let us recall that $x^*(t)$ satisfies $\partial_{x} C(x^*(t),t)=0$ and, thus, by} the Implicit Function Theorem, there exist an open set $U$ containing $x=0$, an open set $V$ containing $t=t_0$, and a unique continuously differentiable function $x^*(t)$ such that
\[
\left \{ (x^*(t), t)| t \in V \right \}=\left\{ (x,t) \in U\times V |\quad \frac{\partial C}{\partial x}(x,t)=0\right\}.
\]
{In particular, 
$x^*(t) \to 0$ as $t \to t_0$. {Furthermore, 
since} $t_0>0$ by Theorem \ref{thm2b}, it is clear that $\partial_{x} C(x,t)$ is differentiable in a neighborhood of $(x,t)=(0,t_0)$,
 and, thus,} we can apply the mean value theorem to show that {there exists $\beta \in (0,1)$} such that
\begin{equation*}
	0=\frac{\partial C}{\partial x}(x^*(t),t)= \frac{\partial^2 C}{\partial x^2}(\beta x^*(t),t_0+ \beta (t-t_0))x^*(t) + \frac{\partial^2 C}{\partial t \partial x}(\beta x^*(t),t_0+ \beta (t-t_0)) (t-t_0).
\end{equation*}
Since $\partial^2 C/ \partial x^2 $, $\partial^2 C/ \partial x \partial t$ are both continuous when $t>0$, and there is an open set containing $(0,t_0)$ such that $\partial^2 C/ \partial x^2$ is strictly positive and, furthermore,
\begin{equation*}
\label{bm:taylor:firstorder}
		\frac{x^*(t)}{t-t_0}=-\frac{\frac{\partial^2 C}{\partial t \partial x}(\beta x^*(t),t_0+ \beta (t-t_0))}{\frac{\partial^2 C}{\partial x^2}(\beta x^*(t),t_0+ \beta (t-t_0))} \; \xrightarrow[t\to t_0] \;-\frac{\frac{\partial^2 C}{\partial t \partial x}(0,t_0) }{\frac{\partial^2 C}{\partial x^2}(0,t_0)}:=\kappa_{1}.
\end{equation*}
Note that $\kappa_{1}>0$ because, as seeing from (\ref{bm:2ndMixedDer:x0}), $\partial_{t}\partial_{x} C(0,t)>0$.
For the second order approximation, we apply a second-order Taylor's expansion of $0=\partial_{x} C(x^{*}(t),t)$ around $(0,t_{0})$ to get:
\begin{align*}
	0&= \frac{\partial^2 C}{\partial x^2}(0,t_0)x^*(t) + \frac{\partial^2 C}{\partial t \partial x}(0,t_0) (t-t_0)+ \frac{1}{2}\frac{\partial^3 C}{\partial x^3}(\beta x^*(t),t_0+ \beta (t-t_0)){x^*(t)}^3\\
	& + \frac{\partial^3 C}{\partial t \partial x^2}(\beta x^*(t),t_0+ \beta (t-t_0))x^{*}(t)(t-t_0)+
	 \frac{1}{2}\frac{\partial^3 C}{\partial t^3}(\beta x^*(t),t_0+ \beta (t-t_0)) (t-t_0)^2.
\end{align*}
and follow similar steps as above. 
\end{proof}

\begin{proof}[\textbf{Proof of Theorem~\ref{xtupbd}}]
{To find the upper bound, we need the following inequality
\begin{equation}\label{imp.ineq}
	e^{\frac{-2\mu x}{\sigma^2}} N\left( \beta_t \right) - N\left(-\alpha_t\right) + \phi\left( \alpha_t \right) \left( 1/\alpha_t + 1/\beta_t \right) > 0,
\end{equation}
which is valid for $\mu<0$ and $x>-\mu t$.
Next, applying (\ref{imp.inequ}) and (\ref{imp.ineq}) to the RHS of (\ref{bm:1stder:L1})}, we can get the following expression:
\begin{align*}\nonumber
	\frac{\frac{\partial C}{\partial x}}{\phi \left(  \alpha_t \right) }& > 
	 \frac{2({c\rho}+\mu t)}{\sigma \sqrt{t}} + \left( \frac{1}{-\beta_t} + \frac{1}{\beta_t^3}\right) \frac{2\mu^2 t - 2 \mu x }{\sigma^2}+  \frac{2 \mu {c\rho}}{\sigma^2(-\beta_t)} - \frac{1}{\alpha_t} - \frac{1}{\beta_t}\\
	&=  \frac{2 x {c\rho} }{\sigma \sqrt{t} (x-\mu t)^2 (x+\mu t)} \left(x^2 - \mu^2 t^2 \left(1-\frac{2 \sigma^2}{\mu {c\rho}}\right)\right):=f(x,t).
	\end{align*}
{It is then clear that $\bar{x}^*(t):=-\mu t \theta_0$ is such that $f(\bar {x}^*(t),t)=0$ and $f(x,t)>0$, for all $x>\bar{x}^*(t)$. Therefore, $x^{*}(t)\leq{}\bar{x}^*(t)$.} 
{For the lower bound, let us again apply the appropriate inequality in (\ref{imp.inequ}) to the different terms of (\ref{bm:1stder:L1}) to get the upper bound:}
\begin{align}\label{ineq:xtlb:1}
\frac{\frac{\partial C}{\partial x}}{\phi \left(  \alpha_t \right)}
	  <    \frac{2 x {c\rho} (x^2 - \mu^2 t^2 \theta_0^2)}{\sigma \sqrt{t} (x-\mu t)^2 (x+\mu t)}  + \sigma t \sqrt{t} \left( \frac{-2\mu}{(x-\mu t)^2} - \frac{2\mu {c\rho}}{(x-\mu t)^3} + \frac{\sigma^2}{(x+\mu t)^3}\right).
\end{align}
 Note that when $t>{c\rho}/(-\mu), x>-\mu t $, the following inequality is true:
\begin{equation}\label{ineq:needtoprove}
	\frac{2 x {c\rho}( - \sigma^2 t - 2 \mu t \sigma \sqrt{t} \theta_0) }{\sigma \sqrt{t} (x-\mu t)^2 (x+\mu t)}  + \sigma t \sqrt{t} \left( \frac{-2\mu }{(x-\mu t)^2} - \frac{2\mu {c\rho}}{(x-\mu t)^3} + \frac{\sigma^2}{(x+\mu t)^3}\right) > 0.
\end{equation}
Then, by adding up the RHS of (\ref{ineq:xtlb:1}) and the LHS of (\ref{ineq:needtoprove}), 
\begin{equation*}
	\frac{\partial C}{\partial x}   <   \frac{2 x {c\rho}  \phi \left(  \alpha_t \right)}{\sigma \sqrt{t} (x-\mu t)^2 (x+\mu t)} \left(x^2 - \mu^2 t^2 \theta_0^2 - \sigma^2 t - 2 \mu t \sigma \sqrt{t} \theta_0\right)=:g(x,t)
\end{equation*}
{Since $\underline{x}^*(t):=-\sigma \sqrt{t}- \mu t \theta_0$ is such that $g(t,\underline{x}^*(t))=0$ and $g(t,x)<0$, for all $x\in (-\mu t,\underline{x}^*(t))$, we conclude that $\underline{x}^*(t)\leq{}x^*(t)$, provided that $x^*(t)>-\mu t$. For the latter, 
we need the additional condition $t > \sigma^2/(\mu^2 ( \theta_0 -1)^2)$.}
\end{proof}

\begin{proof}[\textbf{Proof of Theorem \ref{asymptotics}}]
Let 
\begin{align}\label{ErrorFuncDef}
		H(z)=\frac{N^{c}{\rho}(z)}{\phi(z)},\quad
	\mathcal{E}_{1}(z)=H(z)-\frac{1}{z},\quad 
	\mathcal{E}_{2}(z)=H(z)-\frac{1}{z}+\frac{1}{z^{3}}.
\end{align}
Then, we can write  (\ref{bm:1stder:L1}) as
\begin{align*}
	\frac{\frac{\partial C}{\partial x}}{\phi \left( {\alpha_{t}} \right)}&=\frac{2({c}{\rho}+\mu t)}{\sigma \sqrt{t}}
	+\frac{2\mu(\mu t-x)}{\sigma^{2}}\left\{\frac{1}{|{\beta_{t}}|}-\frac{1}{|{\beta_{t}}|^3}\right\}
+\frac{\sigma^{2}+2 \mu {c}{\rho}}{{\sigma}^2|\beta|}-\frac{1}{{\alpha_{t}}}\\
&\qquad +\left(\frac{2\mu(\mu t-x)}{\sigma^{2}}\right)\mathcal{E}_{2}\left( |{\beta_{t}}| \right)
 +\left(\frac{\sigma^{2}+2 \mu {c}{\rho}}{{\sigma}^2}\right)\mathcal{E}_{1}\left( |{\beta_{t}}| \right)-\mathcal{E}_{1}\left( {\alpha_{t}} \right)
\end{align*}
After some simplification, 
the optimal $x=x^{*}(t) $ is such that 
\begin{align*}
	\frac{x^{2}}{t^{2}}&=\mu^{2}\theta_{0}^{2}+\frac{2\mu\sigma}{2\sigma^{2}{c}{\rho}}\frac{\sqrt{t}(x-t\mu)^{3}(x+t\mu)}{t^{2}x}\mathcal{E}_{2}\left( |{\beta_{t}}| \right)\\
	&\quad-\frac{\sqrt{t}(x-t\mu)^{2}(x+t\mu)}{2{c}{\rho} \sigma xt^{2}}\left(\sigma^{2}+2 \mu {c}{\rho}\right)\mathcal{E}_{1}\left( |{\beta_{t}}| \right)+\frac{\sigma\sqrt{t}(x-t\mu)^{2}(x+t\mu)}{2{c}{\rho}xt^{2}}\mathcal{E}_{1}\left( {\alpha_{t}} \right).
\end{align*}
{Since the error terms $\mathcal{E}_{i}$ converge to $0$ at the order of $\mathcal{O}(t^{-1})$, we conclude that $x^{*}(t)\sim |\mu|\theta_{0}t$, as $t\to{}\infty$.}
{In particular, as $t\to{}\infty$, we also have that 
 $(x^{*}(t)-\mu t)\sim |\mu| t (\theta_{0}+1)$ and $(x^{*}(t)+\mu t)\sim |\mu| t (\theta_{0}-1)$, which in turn implies that $\mathcal{E}_{i}(|\beta_{t}|)=O(t^{-(2i+1)/2})$ and $\mathcal{E}_{1}(\alpha_{t})=O(t^{-3/2})$. We then obtain that:
\begin{align*}
		\frac{x^{*}(t)^{2}}{t^{2}}-\mu^{2}\theta_{0}^{2}
		&=\frac{2\mu\sigma}{2\sigma^{2}{c}{\rho}}\frac{|\mu|^{3}(\theta_{0}+1)^{3}|\mu|(\theta_{0}-1)}{|\mu|\theta_{0}}\frac{3\sigma^{5}}{|\mu|^{5}(\theta_{0}+1)^{5}}t^{-1}\\
		&\quad-\frac{\sigma|\mu|^{2}(\theta_{0}+1)^{2}|\mu|(\theta_{0}-1)}{2({c}{\rho}|\mu|\theta_{0}}\left(\frac{\sigma^{2}+2 \mu {c}{\rho}}{{\sigma}^2}\right)\frac{-\sigma^{3}}{|\mu|^{3}(\theta_{0}+1)^{3}}t^{-1}\\
		&\quad +\frac{\sigma|\mu|^{2}(\theta_{0}+1)^{2}|\mu|(\theta_{0}-1)}{2{c\rho}|\mu|\theta_{0}}\frac{-\sigma^{3}}{|\mu|^{3}(\theta_{0}-1)^{3}}t^{-1}+o(t^{-1}),
\end{align*}
and, after some simplifications, we conclude (\ref{MAH}).}
\end{proof}

\section{Proofs: Geometric Brownian motion}

For the easiness of notation, we write $C$ instead of $\tilde C$ for the expected cost, and write $y=0$ instead of $y=0^+$ when evaluating the functions $C$ and $\rho$ and their derivatives. We also frequently use the the notation 
\begin{equation}\label{TNFGBM}
	\alpha_{\pm}:=\mu\pm \sigma^{2}/2, 
	\quad c:=r+f, \quad 
	\alpha:=\frac{\mu- \sigma^{2}/2}{\sigma},\quad 
	\beta:=\frac{\mu+\sigma^{2}/2}{\sigma}. 
\end{equation}
For future reference let us also note that:
\begin{align}\label{gbm:1st:Dy:L1}
\frac{\partial C}{\partial y} &=-S_0 e^{-y}  \bigg( \frac{2\mu}{\sigma^2} \bigg) e^{-\frac{2y\mu}{\sigma^2}} \bigg[ e^y N \bigg(\frac{-y+\alpha_{-}t }{\sigma\sqrt{t}} \bigg) - e^{\mu t} N \bigg(\frac{-y+\alpha_{+}t }{\sigma\sqrt{t}} \bigg)  \bigg]\\ \label{gbm:1st:Dy:L2}
&\quad +    S_0 e^{-y} \bigg[ e^{\mu t}e^{-\frac{2\mu}{\sigma^2}y}N \bigg(\frac{-y+\alpha_{+}t }{\sigma\sqrt{t}} \bigg)- N \bigg(\frac{-y-\alpha_{-}t }{\sigma\sqrt{t}} \bigg) \bigg]\\
&\quad + c\rho(y,t) e^{-\frac{2y\mu}{\sigma^2}+y}  \bigg[ \frac{2}{\sigma \sqrt{t}} \phi \bigg(\frac{-y+\alpha_{-}t}{\sigma\sqrt{t}} \bigg) - ( \frac{-2\mu}{\sigma^2}+1) N \bigg(\frac{-y+\alpha_{-}t}{\sigma\sqrt{t}} \bigg) \bigg]\label{gbm:1st:Dy:L3}\\
&\quad -c\partial_{y}\rho(y,t) \bigg[ N \bigg(\frac{-y-\alpha_{-}t}{\sigma\sqrt{t}} \bigg) + e^{-\frac{2y\mu}{\sigma^2}+y} N \bigg(\frac{-y+\alpha_{-}t}{\sigma\sqrt{t}} \bigg) \bigg],\label{gbm:1st:Dy:L4}
\end{align}
and, thus, 
\begin{align}
\label{1stD:y0:L1}
\frac{\partial C}{\partial y} (0,t) &= \frac{2c\rho(0,t)}{\sigma\sqrt{t}} \bigg[ \phi \left(\alpha\sqrt{t}\right) +\alpha\sqrt{t}N\left(\alpha\sqrt{t} \right) \bigg]\\
\label{1stD:y0:L2}
&\quad - S_0 \bigg[ 1 + \frac{2\alpha}{\sigma} N\left(\alpha\sqrt{t}\right) - \frac{2\beta}{\sigma}e^{\mu t}  N \left(\beta\sqrt{t}\right) \bigg]-c\partial_{y}\rho(0,t).
\end{align}

\begin{proof}[\textbf{Proof of Lemma \ref{GBMCostFunc}}]
{Let $X_{t}:=\ln (S_t / S_0)= (\mu - \frac{\sigma^2}{2}) t + \sigma W_t$, $x=S_{0}-S_{0}e^{-y}$, and $\tilde{Y}_t:=\inf \limits_{u\leq t}X_u$. Then, it is easy to see that 
\begin{equation*}
\tilde C(y,t)=(S_0 e^{-y}-(r+f)\tilde{\rho}(y)) P(\tilde{Y}_t \leq  -y) + S_0 E[e^{X_t}| \tilde{Y}_t > -y] P (\tilde{Y}_t > -y)+ f-S_0.
\end{equation*} 
We can then 
use (\ref{JntCndSY}) to calculate the cost function using similar steps as those used in Lemma \ref{BMCostFunc}.}
\end{proof}

\begin{proof}[\textbf{Proof of Theorem~\ref{gbmtknotexistence}}]
{We first show that the two conditions in (\ref{ConditionsGBM}) rule out that $y^{*}(t)=0$, for which we will show that $\partial_{y}C(0^{+},t)<0$. Let us start by noting that, in light of (\ref{1stD:y0:L1}) and (\ref{SmplIneqBVU}),
\begin{align*}
\frac{\partial C}{\partial y} (0,t) &<\frac{2c\rho(0,t)}{\sigma\sqrt{t}} \phi \left(\alpha\sqrt{t}\right)-S_{0}\phi(\alpha\sqrt{t})\frac{2 a |\mu|}{\sigma}\sqrt{t}-c\partial_{y}\rho(0,t)\\
&\quad =\frac{2 \phi \left(\alpha\sqrt{t}\right)}{\sigma\sqrt{t}}\left(c\rho(0,t)-a S_{0}|\mu|t\right)-c\partial_{y}\rho(0,t).
\end{align*}
The last expression is then negative due to our two assumptions in (\ref{ConditionsGBM}). This implies that  the optimum is such that $y^*(t)\in (0,\infty]$. To rule out that $y^{*}(t)=\infty$, we now show that the first derivative of the expected cost is positive when $y$ is big enough under the two condition in (\ref{ConditionsGBM2}). Let $D(y,t)$ be the terms of $\partial_{y}C$ in lines  (\ref{gbm:1st:Dy:L1})-(\ref{gbm:1st:Dy:L3}). Next, applying the first or last inequality in (\ref{imp.inequ}) (depending on the sign) to every $N(\cdot)$ term of $D(x,t)$, we get that, for $y>\max\{-\alpha_{-}t,\alpha_{+}t\}$:
\begin{align*}
\frac{D(y,t)}{ \phi \bigg(\frac{-y-\alpha_{-}t}{\sigma\sqrt{t}} \bigg)} &> S_0 e^{-y} 
		 \bigg[ \bigg( \frac{-2\mu}{\sigma^2}\bigg) \bigg(  \frac{\sigma\sqrt{t}}{y-\alpha_{-}t} - \frac{(\sigma\sqrt{t})^3}{(y-\alpha_{-}t)^3 } - \frac{\sigma\sqrt{t}}{y-\alpha_{+}t} \bigg)\\
		 &\qquad \qquad\quad+ \frac{\sigma\sqrt{t}}{y-\alpha_{+} t} 
		 -\frac{(\sigma\sqrt{t})^3}{(y-\alpha_{+}t)^3}
		 -\frac{\sigma\sqrt{t}}{y+\alpha_{-}t } \bigg] +c\rho(y,t) \frac{2y}{\sigma\sqrt{t} (y -\alpha_{-}t)}.
\end{align*}
It is easy to see that, as $y\to\infty$, the first term is asymptotically equivalent to $4S_{0}e^{-y}t\sqrt{t}\sigma|\mu|/y^{2}$, while the second terms is asymptotically equivalent to $2c\rho(y,t)/\sigma\sqrt{t}$. Thus, for large enough $y$, the function $D(y,t)$ is positive in light of the last condition in (\ref{ConditionsGBM2}). Since the expression in (\ref{gbm:1st:Dy:L4}) is eventually nonnegative in light of the second condition in (\ref{ConditionsGBM2}), we conclude that the first derivative of $C$ with respect to $y$ is positive for large enough $y$ and, thus, $y^{*}(t)<\infty$.}
\end{proof}

\begin{proof}[\textbf{Proof of Theorem~\ref{prop:t0limit}}]
Indeed, from (\ref{1stD:y0:L1})-(\ref{1stD:y0:L2}) and the additional assumptions that $\rho(0^{+},t)\equiv \rho(0^{+})\in(0,1]$ and $\limsup_{t\to{}0}|\partial_{y}\rho(0,t)|<\infty$, we obviously have that
\[
		\lim \limits_{t \to 0} \frac{\partial C}{\partial y} (0,t) = +\infty,
\]
which implies that $t^*_{0}>0$ and is such that $\frac{\partial C}{\partial y} (0,t^*_{0})=0$. Furthermore, $t^{*}_{0}$ is only root of $\frac{\partial C}{\partial y} (0,t)$ because, under the additional assumption, $\frac{\partial C}{\partial y} (0,t)$ is strictly decreasing in $t$. Indeed, first note that
\begin{align}
	\label{MixedPartial:y0:L1}
	\frac{\partial^2 C}{\partial t \partial y}(0,t)&=
	S_0 \frac{2\mu}{\sigma\sqrt{t}} 
	\phi \left({\alpha}\sqrt{t} \right)
	+\frac{2S_0\beta}{\sigma}\mu e^{\mu t}
		N\left({\beta}\sqrt{t}\right) - \frac{c\rho(0,t)}{\sigma t \sqrt{t}}
	\phi \left({\alpha}\sqrt{t} \right)-c\frac{\partial^{2}\rho}{\partial t \partial y}(0,t).
\end{align}
The last term is nonpositive in light of our additional assumption $\partial_{t}\partial_{y}\rho(0,t)\geq{}0$. To check that the first two terms are negative we consider two cases. If $\beta>0$, it is clear that (\ref{MixedPartial:y0:L1}) is negative, while if $\beta<0$, we can apply the last inequality in (\ref{imp.inequ}) to the $N(\cdot)$ term to get
\begin{align*}\label{mixedpartial:negative}
\frac{\partial^2 C}{\partial t \partial y}(0,t)< -\frac{c\rho(0)}{\sigma t \sqrt{t}}
\phi \left({\alpha}\sqrt{t} \right)- \frac{c\rho(0,t)}{\sigma t \sqrt{t}}
	\phi \left({\alpha}\sqrt{t} \right)-c\frac{\partial^{2}\rho}{\partial t \partial y}(0,t) <0.
\end{align*}
Therefore, regardless the sign of $\beta$, $\partial_{y}C(0,t)$ is strictly decreasing and we conclude the last assertion of the theorem. 

Finally, we check the second assertion. For the convenience of notation, let us use $t_0$ instead of $t_0^{*}$. We first show that $t_0 \to 0$ as $(r+f)/S_0 \to 0$. Otherwise, suppose that there exist sequences $c_n=r_{n}+f_{n}$ and $S_{n0}$ such that $c_n/S_{n0} \to 0$ and $\lim_{n\to\infty} t_{n0}=:\underline{t}_0\in (0,\infty)$, where $t_{n0}$ is a solution of $\partial_{y} C(0,t_{n0})=0$ corresponding to $c_n$ and $S_{n0}$. 
In that case, 
\begin{equation*}\label{NSETC}
	\lim_{ n \to \infty} \frac{1}{S_{n0}}\frac{\partial C(0,t_n)}{\partial y} = -1-\frac{2\alpha}{\sigma} N({\alpha} \sqrt{\underline{t}_0})+ \frac{2\beta}{\sigma} e^{\mu \underline{t}_0}N({\beta} \sqrt{\underline{t}_0})- c\partial_{y}\rho(0,\underline{t}_{0}),
\end{equation*}
which is actually strictly negative  {because of (\ref{SmplIneqBVU}) and the second condition in (\ref{ConditionsGBM}).} 
Now, let us prove the asymptotic behavior of $t_0$.
Expanding $\phi \left(\alpha\sqrt{t}\right)$, $N\left(\alpha\sqrt{t}\right)$, and $N\left(\beta\sqrt{t}\right)$ as powers of $\sqrt{t}$ and making some simplifications, we have
\begin{align*}
\frac{\partial C}{\partial y}(0,t)&=\frac{2}{\sigma} \left( \frac{c\rho(0) \phi(0)}{\sqrt{t}}  + c\rho(0)\frac{\alpha}{2} + \sqrt{t} \left( c\rho(0) \frac{\alpha^2}{2}\phi(0) + S_0 \phi(0) 2\mu\right)+ \mathcal{E} (t)\right)-c\partial_{y}\rho(0),
\end{align*}
where
\begin{align*}\label{Errorterm}
	\mathcal{E} (t)=S_0\beta\mu N(0)t-\rho(0)c\phi(0)\frac{\alpha^4}{24}t\sqrt{t}  +S_0\frac{\alpha^4}{6}\phi(0)t\sqrt{t}+(S_{0}+c)\mathcal{O}(t^2),
\end{align*}
and the term $\mathcal{O}(t^2)$ does not depend on neither $S_{0}$ nor $c$. 
Now, plugging $t=t_0$, recalling that 
%
$\partial_{y} C (0,t_0)=0$, and setting $\gamma:=c\rho(0)/(2\mu S_0)$, we have: 
\begin{equation*}
	0=1+ \frac{\alpha}{2\phi(0)}\sqrt{t_0} + \frac{t_0}{\gamma}+ \frac{\alpha^2}{2}t_0 + \frac{\mathcal{E}(t_0)\sqrt{t_0}}{2\gamma S_0\mu\phi(0)}-\frac{\sigma \partial_{y}\rho(0)\sqrt{t_0}}{2\rho(0)\phi(0)}.
\end{equation*}
Therefore, since $t_{0}\to{}0$  as $c/S_0 \to 0$, we have
\begin{equation*}
\frac{t_0}{\gamma}\left(1+\frac{\mathcal{E}(t_0)}{2 S_0\sqrt{t_0}\mu\phi(0)}\right) = -1 -\frac{\alpha}{2\phi(0)}\sqrt{t_0}-\frac{\alpha^2}{2}t_0-\frac{\sigma \partial_{y}\rho(0)\sqrt{t_0}}{2\rho(0)\phi(0)}\to -1,
\end{equation*}
which completes the proof of the second assertion since clearly 
\[
	\frac{\mathcal{E}(t_0)}{2 S_0\sqrt{t_0}\mu\phi(0)}=\frac{\beta}{2\phi(0)} N(0)\sqrt{t_{0}}-\frac{\rho(0)c\alpha^4}{48 \mu S_{0}}t_{0}+\frac{\alpha^4}{12\mu}t_{0}+(1+\frac{c}{S_{0}})\mathcal{O}(t^2)\to{}0,
\]
as $c/S_0 \to 0$.
\end{proof}
\begin{proof}[\textbf{Proof of Theorem~\ref{ttot0}}]
For the easiness of notation, let us use $t_0$ instead of $t_0^*$. We basically need to show that $\partial^2_{y} C (0,t_0)>0$ and the rest of the proof follows along the same line as that of Theorem \ref{BM:xt:ttot0}. To check that $\partial^2_{y} C (0,t_0)>0$, let us start by noting that 
\begin{align}
\label{SecondPartial:y0:L1}
	\frac{\partial^2 C}{\partial y^2}(0,t)&= S_0 \left( \frac{2\mu}{\sigma^2}\right)^2  
	N\left(\alpha\sqrt{t}\right)
	-\frac{4S_0\beta^{2}}{\sigma^{2}} e^{\mu t}
	N\left(\beta\sqrt{t}\right)+S_0
	N\left(-\alpha\sqrt{t}\right)\\
	\label{SecondPartial:y0:L2}
	&\quad -\frac{4\alpha}{\sigma^{2}\sqrt{t}} \rho(0)c\left[ \phi \left(\alpha\sqrt{t} \right) +\alpha\sqrt{t} N\left(\alpha\sqrt{t} \right) \right]\\
		\label{SecondPartial:y0:L3}
	&\quad +\frac{8c\rho'(0)}{\sigma\sqrt{t}} \left[ \phi \left(\alpha\sqrt{t} \right) +\alpha\sqrt{t} N\left(\alpha\sqrt{t} \right) \right]
	-c\rho''(0).
\end{align}
The terms in (\ref{SecondPartial:y0:L3}) are positive since, by assumption,  $\rho''(0)<0$, $\rho'(0)>0$, and $\alpha<0$ so that $\phi \left(\alpha\sqrt{t} \right) +\alpha\sqrt{t} N\left(\alpha\sqrt{t} \right)>0$. Let us then analyze the terms in (\ref{SecondPartial:y0:L1})-(\ref{SecondPartial:y0:L2}) that we denote $D(t)$. To this end, we recall that $\partial_y C (0,t_0)=0$ and, from the expression in (\ref{1stD:y0:L1})-(\ref{1stD:y0:L2}), we get that
\begin{align}\nonumber
&\frac{\rho(0)c}{S_{0}}\left[ \phi \left(\alpha\sqrt{t_{0}} \right) +\alpha\sqrt{t_{0}} N\left(\alpha\sqrt{t_{0}} \right) \right]=1 + \frac{2\alpha}{\sigma}N\left(\alpha\sqrt{t_{0}}\right) - \frac{2\beta}{\sigma} e^{\mu t_0}  N \left(\beta\sqrt{t_{0}} \right).\nonumber
\end{align}
We can then substitute into (\ref{SecondPartial:y0:L2}) to get:
\begin{align*}
	D(t_{0}) &= S_0 
	N\left( -\alpha \sqrt{t_{0}}\right)
	-e^{\mu t_{0}} N\left( \beta \sqrt{t_{0}}\right)
	-\frac{4\mu}{\sigma^2}\left(\frac{1}{2}
	-N\left( \alpha \sqrt{t_{0}}\right)
	+ e^{\mu t} N\left(\beta \sqrt{t_{0}}\right) \right)\\
	&+N\left(\beta \sqrt{t_{0}}\right)\left(1-e^{\mu t_0}\right)
\end{align*}
Finally, when $\mu < 0$,  all the terms above are positive and, therefore, $\partial^2_{y} C(0,t_0) >0$. 
\end{proof}

\begin{proof}[\textbf{Proof of Proposition~\ref{tknot_gbm_bounds}}]
{Recall from Theorem \ref{prop:t0limit} that $\partial_{y} C(0,t_0^*)=0$ and $\partial_{y} C(0,t)$ is a strictly decreasing function of $t$. Therefore, to find a lower bound (respectively, upper bound) of $t_0:=t_{0}^{*}$, we will find the root of a function that is a lower (respectively, upper) bound of $\partial_{y} C(0,t)$. The upper bounds and subsequent roots were obtained in both cases in the proof of Theorem \ref{gbmtknotexistence}. Let us define $A(t)$ be the first two terms of (\ref{1stD:y0:L1})-(\ref{1stD:y0:L2}). Then we can rewrite $A(t)= 2\sigma^{-1}\rho(0)c t^{-1/2} f(\sqrt{t})
+S_0 g(\sqrt{t})$, where 
\begin{align*}
	f(x)=\phi(\alpha x) + \alpha x N(\alpha x),\quad 
	g(x)=-1-\frac{2}{\sigma}\alpha N(\alpha x)+ \frac{2}{\sigma}\beta e^{\mu x^{2}}N(\beta x).
\end{align*}
Note that $f'(x)=\alpha N(\alpha x)$, $f''(x)=\alpha^{2}\phi(\alpha x)$, and $g'(x)=\frac{4}{\sigma}\left( \mu \phi (\alpha x) +\beta \mu  x e^{\mu x^2}N (\beta x) \right)$. Then,} for some 
$\theta\in(0,\sqrt{t})$, we have
\[
	f(\sqrt{t})=\phi(0)+\frac{\alpha}{2}\sqrt{t}+\frac{1}{2}f''(\theta) t=\phi(0)+\frac{\alpha}{2}\sqrt{t}+\frac{\alpha^{2}}{2}\phi(\alpha\theta) t\geq{}\phi(0)+\frac{\alpha}{2}\sqrt{t}.
\]
Similarly, for $\beta<0$, $g'(x)\geq{}\frac{4}{\sigma}\mu \phi (\alpha x)\geq{}\frac{4}{\sigma}\mu \phi (0)$ and, thus, for some $\theta'\in(0,\sqrt{t})$, $g(\sqrt{t})=g'(\theta')\sqrt{t}\geq{}\frac{4}{\sigma}\mu \phi (0)\sqrt{t}$. 
Putting together the two previous inequalities, 
\begin{align*}
	A(t)&\geq{}\frac{1}{\sigma\sqrt{t}}\left(2\rho(0)c\phi(0)+\alpha \rho(0)c\sqrt{t}+4S_{0}\mu\phi(0) t\right).\label{L2In2}
\end{align*}
The root of the last expression is precisely $\underline{t}(\mu\phi(0))$.
For $\beta \geq 0$, we instead use that, for some $\theta'\in(0,\sqrt{t})$, 
\[
	g(\sqrt{t})=g'(\theta')\sqrt{t}\geq{}\left(	\frac{4\mu}{\sigma}\phi(0)-\frac{4\beta}{\sigma}\sqrt{-\mu /2}e^{-0.5}\right)\sqrt{t},
\]
which follows from the inequality $\mu  x \exp(\mu x^2)>-\sqrt{-\mu /2}e^{-0.5}$.
Putting together the inequality above with the lower bound of $f(\sqrt{t})$, 
we get that
\begin{align*}
A(t)&\geq{}\frac{1}{\sigma\sqrt{t}}\left(2\rho(0)c\phi(0)+\alpha \rho(0)c\sqrt{t}+4S_0\left(	\mu\phi(0)-\beta\sqrt{-\mu /2}e^{-0.5}\right) t\right).\label{L2In3}
\end{align*}
It is easy to check that $\underline{t}(	\mu\phi(0)-\beta\sqrt{-\mu /2}e^{-0.5})$ is the positive root of the last expression. 
\end{proof}

\begin{proof}[\textbf{Proof of Theorem \ref{tinfty:regime1}}]
Recall that we are assuming that $\rho=\rho(y,t)$ is constant in $y$ and $t$. Under this assumption, $\partial C/\partial y$ reduces to the expression in (\ref{gbm:1st:Dy:L1})-(\ref{gbm:1st:Dy:L3}). Let us also use the following notation through the proof:
\begin{equation}\label{TNUP}
	d_{t} := \frac{-y^{*}(t)-\mu t+ \frac{\sigma^2 t}{2} }{\sigma\sqrt{t}}, \quad 
	e_{t}:=\frac{-y^{*}(t)+\mu t + \frac{\sigma^2}{2}t}{\sigma\sqrt{t}},\quad 
	f_{t}:=\frac{-y^{*}(t)+\mu t - \frac{\sigma^2}{2}t}{\sigma\sqrt{t}}.
\end{equation}
We will first show that $d_{t}\to\infty$ as $t\to\infty$.
For easiness of notation, in what follows we simply write $y(t)=y^{*}(t)$ and often omit $t$ in $d_{t}$ and $y(t)$. To begin with, note that $(-y+\alpha_{-}t)/\sqrt{t}\to-\infty$, as $t\to\infty$, regardless of the value of $y$, which suggests to approximate the first and second terms of (\ref{gbm:1st:Dy:L1}) and (\ref{gbm:1st:Dy:L3}), respectively, with $\phi(z)/z$.  Now, let us first assume that $d_{t}\to{}\bar{d}\in(-\infty,\infty)$, as $t$ goes to $\infty$ along a subsequence $t_{n}\to\infty$. This implies that $e_{t}=d_t+2\mu t/(\sigma\sqrt{t})\to-\infty$ as  $t$ goes to $\infty$ along $t_{n}$, 
and, thus, we can approximate the second and first terms of (\ref{gbm:1st:Dy:L1}) and (\ref{gbm:1st:Dy:L2}), respectively, by $\phi(z)/z$. Concretely, recalling the error functions $\mathcal{E}_1$ and $\mathcal{E}_2$ from (\ref{ErrorFuncDef}), we write (\ref{gbm:1st:Dy:L1})-(\ref{gbm:1st:Dy:L3}) as
\begin{align}\label{Eq4L1}
	\frac{\partial C}{\partial y} &=  -S_0 e^{-y}  N (d_{t}) +\phi \left(d_{t} \right) \frac{2S_0 e^{-y}}{\sigma^{2}} \left\{ \frac{\mu}{d_t} + \frac{\alpha_{+}}{e_{t}} \right\}+\phi \left(d_{t} \right)  \frac{2\rho c}{\sigma\sqrt{t}} \left\{ 1- \frac{\alpha\sqrt{t}}{f_{t}}\right\} \\
	 \label{Eq4L4}
	 &\quad +\phi \left(d_{t} \right) \frac{2S_0 e^{-y}}{\sigma^{2}}\left\{ |\mu| \mathcal{E}_1(-f_{t}) +\alpha_{+} \mathcal{E}_1(-e_{t}) \right\} +\phi \left(d_{t} \right)  \frac{2\rho c\alpha_{-}}{\sigma^{2}}\mathcal{E}_1( -f_{t}),
\end{align}
where we have used the notation in (\ref{TNFGBM}) and (\ref{TNUP}). 
Note that the first two terms on the right-hand side of  (\ref{Eq4L1}) times $\sqrt{t}$ 
goes to $0$ since $d_{t}\to{}\bar{d}$. In a similar way, the terms of (\ref{Eq4L4}) times $\sqrt{t}$ all converge to $0$, while the third term on the right-hand side of (\ref{Eq4L1}) times $\sqrt{t}$  can be rewritten as
\[
\phi (d_{t}) 2\rho c \frac{y}{y-\mu t + \frac{\sigma^2 t}{2}} =  \phi (d_{t}) 2\rho c\frac{-d_{t}\sigma\sqrt{t}-\mu t + \frac{\sigma^2 t}{2}}{-d_{t}\sigma\sqrt{t}+2(-\mu t + \frac{\sigma^2 t}{2})}\to \phi (\bar{d}) \rho(0)c
\]
as $t\to\infty$ along $t_{n}$. We conclude that  $\sqrt{t_{n}} \partial_{y} C(y(t_{n}),t_{n}) \to \phi (\bar{d}) \rho d$, 
which is contradiction, since $\partial_{y} C(y(t_{n}),t_{n}) =0$.
Secondly, let us assume that  $d_{t} \to -\infty$, as $t\to\infty$, and  note that, in that case,  we can apply the approximation $N(-z)\approx \phi(z)/z$ to all terms and, after some simplifications, we obtain:
\begin{align*}
	0 &=  
  \frac{S_0 e^{-y} \sigma \sqrt{t} 2\mu t 2y}{(y - \mu t +\frac{\sigma^2 t}{2})(y-\mu t - \frac{\sigma^2 t}{2}) ( y+\mu t - \frac{\sigma^2 t}{2})} + \frac{ 2\rho c y}{\sigma\sqrt{t} (y-\mu t + \frac{\sigma^2 t}{2})}\\
	&\quad- \frac{2\mu S_0 e^{-y}}{\sigma^2} \mathcal{E}_1(-f_{t})+ \frac{\alpha_{+}S_0 e^{-y}}{\sigma^{2}}\mathcal{E}_1 (-e_{t}) - S_0 e^{-y}\mathcal{E}_1 (-d_t)+\frac{2\rho c\alpha_{-}}{\sigma^{2}}\mathcal{E}_1(-f_t).
\end{align*}
From here, we conclude that
\[
\frac{ \sqrt{t} (y-\mu t + \frac{\sigma^2}{2}t)  }{ y\phi(d_{t}) } \frac{\partial C}{\partial y}\to \frac{2\rho c}{\sigma}, 
\]
which again leads to a contradiction. 
 We can then conclude that $d_{t} \to \infty$, as $t\to\infty$.  We are now ready to show the first and second approximation for $y^*(t)$. We consider two cases:

\medskip
\noindent \textbf{Case 1: $\mu+\frac{\sigma^2}{2}<0$}. Let us apply the approximation $\phi(z)/z$ to the appropriate terms in (\ref{gbm:1st:Dy:L1})-(\ref{gbm:1st:Dy:L3}) and divide by 
$2\phi \left({d_{t}} \right) {\rho{c}}y/\left(\sigma\sqrt{t}(y-\mu t +\sigma^2 t/2)\right)$. Then, we have the following equation:
\begin{align*}
	0& = \frac{ -S_0 e^{-y}}{  \phi \left({d_{t}} \right){\rho{c}}\frac{2}{\sigma\sqrt{t}} \frac{y}{y-\mu t + \frac{\sigma^2 t}{2}}} +1\\
	&\quad + \frac{S_0 e^{-y}\sigma \sqrt{t}}{2{\rho{c}}} \left( \frac{\sigma \sqrt{t} (y+\mu t + \frac{\sigma^2 t}{2})}{y(y-\mu t - \frac{\sigma^2 t}{2})}+ \frac{y-\mu t + \frac{\sigma^2 t}{2}}{y {d_{t}}}\right)\\
	 &\quad -S_0 e^{-y}\frac{y }{2{\rho{c}}}  \big(\frac{2|\mu|}{\sigma^2}\big) \frac{1}{f_{t}}\mathcal{E}_1 (-f_t)-S_0 e^{-y}\frac{y}{{\rho{c}}} \frac{\alpha_{+}}{\sigma^{2}}\frac{1}{f_{t}} \mathcal{E}_1(-e_{t})\\
	 	 & \quad -S_0e^{-y}\frac{\sigma^{2}t}{2{\rho{c}}}\mathcal{E}_1 \left({d_{t}} \right)\frac{f_{t}}{y}  - \frac{y\alpha_{-} }{\sigma^{2}} \frac{1}{f_{t}}\mathcal{E}_1(-f_{t}).
\end{align*}
Therefore, as $t \to \infty$, we have $S_0 e^{-y}\sim  \phi(d_{t})  {\rho{c}}\frac{2}{\sigma^{2}{t}} \frac{y}{f_{t}} $. Let us take the logarithm to the both sides  to get
\begin{equation}\label{KLFSAp}
\frac{(y - ty_{0}^{+})(y-ty_{0}^{-})}{2\sigma^2 t} + \frac{1}{2} \ln t +\ln\left( \frac{y-\mu t + \sigma^2 t/2}{y}\right)+ \ln \bigg( \frac{S_0\sigma \sqrt{2\pi}}{2{\rho{c}}} \bigg) \to 0,
\end{equation}
where  $y_{0}^{\pm}:=-\mu + \frac{3}{2}\sigma^2  \pm \sigma\sqrt{-2\mu+2\sigma^2}$. 
Let us divide both sides by $t$ and, since $y/t \nrightarrow 0$ (as proved in Lemma \ref{gbm:ttoinfty:Lem2}), $\ln((y-\mu t + \sigma^2 t/2)/y)$ is bounded, which implies that $(y/t -y_{0}^{+})(y/t -y_{0}^{-})\to{}0$.
To decide whether $y/t\to y_{0}^{+}$ or $y/t\to y_{0}^{-}$, note that, since $d_{t} \to \infty$, $y/t < -\mu + \sigma^2 /2$, for large enough $t$ and, thus,  $y/t - y_{0}^{+}< -\sigma^2 - \sigma\sqrt{-2\mu +2\sigma^2}<0$. Therefore, we must have that $y/t-y_{0}^{-}\to{}0$, 
and, we conclude the first order approximation. For the second approximation, let us divide both sides of (\ref{KLFSAp}) by $\ln(t)$ and rewrite the resulting first term to conclude that 
\[
	\frac{\left(\frac{y}{t} - y_{0}^{+}\right)\left(\frac{y}{t} - y_{0}^{-}\right)}{2\sigma^2 \frac{1}{t}\ln t} \to -\frac{1}{2},
\]
and, since $y/t -y_{0}^{+}\to y_{0}^{-}-y_{0}^{+}=-2\sigma\sqrt{-2\mu+2\sigma^2}$, we finally conclude 
the second-order approximation for $y^{*}(t)$.

\medskip
\noindent \textbf{Case 2: $\mu+\frac{\sigma^2}{2}\geq{}0$}.
We consider the case of $e_{t}\to \bar{e}\in(-\infty,\infty]$, when $t$ goes to $\infty$ along a sequence $t_{n}\to\infty$ (the case $e_{t}\to -\infty$ is similar). Then, by approximating the first term in (\ref{gbm:1st:Dy:L1}) and the second terms in (\ref{gbm:1st:Dy:L3}) by $\phi(z)/z$ and making some simplifications, we can write (\ref{gbm:1st:Dy:L1})-(\ref{gbm:1st:Dy:L3}) as:
\begin{align*}
	0 &= \frac{ -S_0 e^{-y} + S_0e^{-y-\frac{2y\mu}{\sigma^2}+\mu t}\bigg( \frac{2\mu}{\sigma^2}+1\bigg) N(d_{t})}{  \phi({d_{t}}) {\rho{c}}\frac{2}{\sigma\sqrt{t}} \frac{y}{y-\mu t + \frac{\sigma^2 t}{2}}} +1+ S_0 e^{-y}\left(\frac{(-\mu)t}{{\rho{c}}y}+  \frac{y-\mu t + \frac{\sigma^2 t}{2}}{y {d_{t}}}\right)\\
	 &\quad +S_0 e^{-y}\frac{1}{2{\rho{c}}}\left\{ \frac{2\mu}{\sigma^2}\frac{y^{2}}{f_{t}} \mathcal{E}_1 (-f_{t})- \sigma^{2}t f_{t}\mathcal{E}_1 \left({d_{t}} \right)\right\} - \frac{y \alpha_{-}}{\sigma^{2}} \frac{1}{f_{t}}\mathcal{E}_1(-f_{t}).
\end{align*}
This in turn implies that  
\begin{align*}
	\frac{S_0 \sqrt{2 \pi}\exp \left( -y + \frac{(-y-\mu t + \sigma^2 t/2)^2}{2\sigma^2 t} \right) \left[ 1-\left( \frac{2\mu}{\sigma^2}+1\right)N(d_t) \exp \left( \frac{2\mu}{\sigma}d_t\sqrt{t} -\frac{2\mu^2}{\sigma^2}t \right)\right]}{\left(r+f\right)\frac{2}{\sigma\sqrt{t}} \frac{y}{y-\mu t + \frac{\sigma^2 t}{2}}} \to 1.
\end{align*}
Now, let us take logarithm to the both sides of the previous equation to get
\begin{align}
\nonumber
	&-y +\frac{(-y-\mu t+ \frac{\sigma^2 t}{2})^2 }{2\sigma^2 t}+\frac{\ln t}{2} +\ln \frac{(y-\mu t + \sigma^2 t/2)}{y} + \ln \left( \frac{S_0\sigma \sqrt{2\pi}}{2{\rho{c}}} \right)\\
	\label{Eq17}
	&+\ln \left[ 1-\left( \frac{2\mu}{\sigma^2}+1\right) N(d_t) \exp \left( \frac{2\mu}{\sigma}d_t\sqrt{t} -\frac{2\mu^2}{\sigma^2}t \right)\right] \to 0.
\end{align}
Note that we are under the assumption that $d_{t}\to \bar{d}\in(-\infty,\infty]$, 
which implies the last term in (\ref{Eq17}) converges to 0. 
It is then clear that, 
\begin{equation}
\nonumber
	-y +\frac{(-y-\mu t+ \frac{\sigma^2 t}{2})^2 }{2\sigma^2 t}+\frac{\ln t}{2} +\ln \frac{(y-\mu t + \sigma^2 t/2)}{y} + \ln \left( \frac{S_0\sigma \sqrt{2\pi}}{2{\rho{c}}} \right)\to 0,
\end{equation}
which can then produce the exact same expression as (\ref{KLFSAp}) and we proceed as therein to get the second order approximation.
\end{proof}

\begin{proof}[\textbf{Proof of  Theorem \ref{sig0:delta_approx}}]
{In Lemmas~\ref{sigma0:uniqueness} and~\ref{sig0:locmin} below, we show that $y(\sigma)^*\in [0,-\mu t]$ and $\delta(\sigma):=y^{*}(\sigma)+\mu t\nearrow 0$, as $\sigma \to 0$.} Then, 
replacing  $y=y^*(\sigma)$ with $-\mu t+\delta(\sigma)$ in (\ref{gbm:1st:Dy:L1})-(\ref{gbm:1st:Dy:L3}) and using that $\delta(\sigma) \to 0$, we get
\begin{align*}
	  \phi \left(\frac{-\delta(\sigma)+ \frac{\sigma^2 t}{2} }{\sigma\sqrt{t}} \right)\frac{\rho(r+f)}{\sigma\sqrt{t}} \to S_0 e^{\mu t}.
\end{align*}
Taking logarithms to both sides, dividing by $\ln (1/\sigma)$, and denoting $a=\ln S_0 + \mu t + \frac{1}{2} \ln t - \ln (r+f)+\frac{1}{2} \ln 2\pi$, we get
\begin{equation}\label{sig0:conv}
	1 - \frac{\delta^2(\sigma)}{2\sigma^2 t \ln (1/\sigma)}+\frac{\sigma^2 t \delta(\sigma)}{2\sigma^2 t \ln (1/\sigma)}-\frac{\sigma^4 t^2 /4}{2\sigma^2 t \ln (1/\sigma)}
	 -\frac{a}{\ln (1/\sigma)} \to 0
\end{equation}
As $\sigma \to 0$ and $\delta \to 0$, the third term and the fourth term of this expression converges to 0, 
which implies $\delta(\sigma)\sim -\sqrt{2\sigma^2 t \ln (1/ \sigma)}$, as $\sigma\to{}0$. 
To prove the second approximation, rewrite (\ref{sig0:conv}) as
\[
	\frac{1+ \frac{\delta(\sigma)}{\sqrt{2\sigma^2 t \ln (1/\sigma)}}}{c/\ln (1/\sigma) } + \frac{\delta(\sigma)}{4c}  -\frac{\sigma^2 t}{26c}  - \frac{1}{\big( 1- \frac{\delta(\sigma)}{ \sqrt{2\sigma^2 t \ln (1/\sigma)}} \big)}\to 0.
\]
Since $\delta(\sigma)\to 0$ and $\sigma \to 0$, and from the first-order approximation for $\delta(\sigma)$, the second and third terms above converge to 0, while the fourth term converge to $-1/2$. Therefore, 
\begin{equation*}
		 \frac{\delta(\sigma)}{\sqrt{2 \sigma^2 t \ln (1/\sigma)}}+1  \sim \frac{a}{2 \ln (1/\sigma)},
\end{equation*}
and this implies the second-order approximation. 
\end{proof}

\section{Supporting Lemmas}\label{SuppLmsC}
{\begin{lemma}\label{DerivationOfC}
	Under the setting and assumptions of Section \ref{Sec2}, the cost function can be written as in Eq.~(\ref{eq:CostContinuousApprox}).
\end{lemma}}
\begin{proof}
{Using the notation introduced above Eq.~(\ref{eq:CostContinuousApprox}), we have:
\begin{align*}
\bar{C}_{\delta,\eps}(x,t)&= E\left[\text{Cost}\times\mathbf{1}\{\bar{Y}_t > S_0 - x + \epsilon\}\right] +  E\left[\text{Cost}\times\mathbf{1}\{\bar{Y}_t \leq S_0 - x + \epsilon\}\right] \\
& = E\left[\left.\bar{S}_t+f-\bar{S}_{0}\right|\bar{Y}_t >\bar{S}_0 -x+\eps\right]P\left(\bar{Y}_t>\bar{S}_0-x+\eps\right)\\
&\quad+  E\left[\text{Cost}\times\mathbf{1}\{\bar{Y}_t \leq S_0 - x + \epsilon,E_t\}\right]\\
&\quad + E\left[\text{Cost}\times\mathbf{1}\{\bar{Y}_t \leq S_0 - x + \epsilon,E_t^c, \tau_{j+1} > t\}\right] \\ 
&\quad + 
E\left[\text{Cost}\times\mathbf{1}\{\bar{Y}_t \leq S_0 - x + \epsilon,E_t^c, \tau_{j+1} \leq t\}\right]\\
& = E\left[\left.\bar{S}_t+f-\bar{S}_0\right|\bar{Y}_t >\bar{S}_0 -x+\eps\right]P\left(\bar{Y}_t>\bar{S}_0-x+\eps\right)\\
&\quad+ (-x-r)P(\bar{Y}_t \leq \bar{S}_0-x+\eps,{E_{t}}) \\
&\quad + (-x+f+\eps)P(\bar{Y}_t \leq S_0 - x + \epsilon,E_t^c, \tau_{j+1} > t) \\ 
&\quad + 
(-x+f+2\eps)P(\bar{Y}_t \leq S_0 - x + \epsilon,E_t^c, \tau_{j+1} \leq t) \\
& = E\left[\left.\bar{S}_t+f-\bar{S}_0\right|\bar{Y}_t >\bar{S}_0 -x+\eps\right]P\left(\bar{Y}_t>\bar{S}_0-x+\eps\right)\\
&\quad+ \left[(-x-r)\rho(x,t) + (-x + f + 2\eps)(1-\rho(x,t))\right]P(\bar{Y}_t \leq \bar{S}_0-x+\eps) \\
&\quad -\eps P(\bar{Y}_t \leq \bar{S}_0-x+\eps, \tau_{j+1} > t,E_t^c),
\end{align*}
where we recall that $\rho(x,t)=P(E_t|\bar{Y}_t \leq S_0 - x + \epsilon)$. After a rearrangement of the terms above, we get Eq.~(\ref{eq:CostContinuousApprox}).}
\end{proof}
\begin{lemma}\label{SmplIdntN}
Let $\alpha:=(\mu- \sigma^{2}/2)/\sigma$ and $\beta:=(\mu+\sigma^{2}/2)/\sigma$,
and suppose that $\mu<0$. Then, 
\begin{equation}\label{SmplIneqBVU}
	1 + \frac{2\alpha}{\sigma} N\left(\alpha\sqrt{t}\right) - \frac{2\beta}{\sigma}e^{\mu t}  N \left(\beta\sqrt{t}\right) >\phi(\alpha\sqrt{t})\frac{2a|\mu|}{\sigma}\sqrt{t},
\end{equation}
where $a=1$ if $\mu\leq{}-\sigma^{2}/2$ and $a=2$ if $\mu>-\sigma^{2}/2$.
\end{lemma}
\begin{proof}
Let $A$ denote the left-hand side of the inequality in (\ref{SmplIneqBVU}) and note that 
\begin{align*}
	A =
	N\left(-\alpha\sqrt{t}\right) - e^{\mu t}  N \left(\beta\sqrt{t}\right) +\frac{-2\mu}{\sigma^{2}}\left(e^{\mu t}  N \left(\beta\sqrt{t}\right) -N\left(\alpha\sqrt{t}\right)\right)>N\left(-\alpha\sqrt{t}\right) - e^{\mu t}  N \left(\beta\sqrt{t}\right),
\end{align*}
since
		\begin{align*}
		e^{\mu t} N(\beta\sqrt{t})-N(\alpha\sqrt{t})&=\int_{-\infty}^{\beta\sqrt{t}}e^{\mu t}\phi(z)dz-\int_{-\infty}^{\alpha\sqrt{t}}\phi(z)dz\\
		&=\int_{-\infty}^{0}\frac{1}{\sigma\sqrt{t}}\left(e^{\mu{}t}\phi\left(\frac{v+\beta t}{\sigma\sqrt{t}}\right)-\phi\left(\frac{v+\alpha t}{\sigma\sqrt{t}}\right)\right)dv\\
		&=  \int_{-\infty}^{0}\phi\left(w+\alpha \sqrt{t}\right)\left(e^{-\sigma\sqrt{t}w}-1\right)dw>0.
	\end{align*}
The inequality for the case $\mu\leq{}-\sigma^{2}/2$ then follows because
\begin{align*}
	N(-\alpha\sqrt{t})-e^{\mu t} N(\beta\sqrt{t})&>N(-\alpha\sqrt{t})- N(\beta\sqrt{t})\\
	&\geq{}\min\{\phi(\alpha\sqrt{t}),\phi(\beta\sqrt{t})\}(-\alpha-\beta)\sqrt{t}\\
	&=\phi(\alpha\sqrt{t})\frac{2|\mu|}{\sigma}\sqrt{t}.
\end{align*}
where we used that $-\alpha>\beta$ and $|\alpha|>|\beta|$. For the other case $\mu>{}-\sigma^{2}/2$, let 
\[
	g(x)=-1-\frac{2}{\sigma}\alpha N(\alpha x)+ \frac{2}{\sigma}\beta e^{\mu x^{2}}N(\beta x)
\]
and note that, since $\beta>0$ when  $\mu>{}-\sigma^{2}/2$, 
\[
	g'(x)=\frac{4}{\sigma}\left( \mu \phi (\alpha x) +\beta \mu  x e^{\mu x^2}N (\beta x) \right)\leq{}4\mu \phi (\alpha x).
\]
Thus, for some $\theta\in(0,\sqrt{t})$, $g(\sqrt{t})=g'(\theta)\sqrt{t}\leq{}\frac{4}{\sigma}\mu \phi (\alpha \theta)\sqrt{t}\leq{}\frac{4}{\sigma}\mu \phi (\alpha \sqrt{t})\sqrt{t}$,
which is the same inequality as in (\ref{SmplIneqBVU}) when $\mu>{}-\sigma^{2}/2$.
\end{proof}

\begin{lemma}
\label{gbm:ttoinfty:Lem2} 
 {Let $\mu < 0$ and suppose that $\rho=\rho(y,t)$ is constant in $y$ and $t$. Let $y^*(t)$ be as in Theorem \ref{gbmtknotexistence}.} Then, $y^*(t)/t \nrightarrow 0$, as $t \to \infty$.
\end{lemma}
\begin{proof}
 Let us recall $\mathcal{E}_1(z)$ from (\ref{ErrorFuncDef}), the notation in (\ref{TNFGBM}), and  the notation $d_{t},e_{t},f_{t}$ introduced in (\ref{TNUP}). First, let us consider the case when $\mu +\frac{ \sigma^2}{2} <0$.  Note that when $\mu +\frac{ \sigma^2}{2} <0$, $(y-\mu t \pm \sigma^{2}t/2)/\sqrt{t}\to \infty$, as $t\to\infty$, regardless of the value of $y\geq{}0$. Next, we rewrite (\ref{gbm:1st:Dy:L1})-(\ref{gbm:1st:Dy:L3}) using  the approximation $\phi(z)/z$ and error functions as follows:
\begin{align}\label{Eq2L1}
	 \frac{\frac{\partial C}{\partial y} }{ \phi({d_{t}})} &=  -S_0 e^{-y}N({d_{t}})/\phi\left({d_{t}} \right)\\
	\label{Eq2L2}
	&\quad + S_0 e^{-y} \bigg\{ -\frac{2|\mu|}{\sigma^2} \frac{1}{f_{t}}-\frac{2\alpha_+}{\sigma^{2}}\frac{1}{e_{t}}\bigg\} - \frac{2\rho c y}{\sigma^{2}t} \frac{1}{f_{t}}\\
	 \label{Eq2L4}
	 &\quad + S_0 e^{-y}\bigg\{ \frac{2|\mu|}{\sigma^2} \mathcal{E}_1 (-f_t)+ \frac{2\alpha_{+}}{\sigma^{2}} \mathcal{E}_1 (-e_t)\bigg\}+ \frac{2\rho c\alpha_{-}}{\sigma^{2}}\mathcal{E}_1(-f_t).
\end{align}
Note that the error terms in (\ref{Eq2L4}) converges to $0$, regardless of the value of $y>0$. 
 For easiness of notation, in what follows we simply write $y=y(t)=y^{*}(t)$. Let us assume that $y(t)/t\to 0$, as $t \to \infty$. Then, it is easy to see that the terms in (\ref{Eq2L2}) converge to 0, since 
\[
	\frac{\sigma\sqrt{t}}{y-\mu t \pm \frac{\sigma^2 t}{2}}=\frac{1}{\sqrt{t}}\frac{\sigma}{\frac{y}{t}-\mu  \pm \frac{\sigma^2 }{2}}\to{}0,\quad 
	\frac{1}{\sqrt{t}}\frac{y}{y-\mu t \pm \frac{\sigma^2 t}{2}}=
	\frac{1}{\sqrt{t}}\frac{\frac{y}{t}}{\frac{y}{t}-\mu  \pm \frac{\sigma^2 }{2}}\to{}0.
\]
For the RHS of (\ref{Eq2L1}), let us write it as
\begin{align*}
	-S_{0}\sqrt{2 \pi}\, \exp\left({t\left(-\frac{y}{t} + \frac{(-\frac{y}{t}-\mu  + \frac{\sigma^2}{2})^2}{2\sigma^2 }\right)}\right)N\bigg({t\left(-\frac{y}{t} + \frac{(-\frac{y}{t}-\mu  + \frac{\sigma^2}{2})^2}{2\sigma^2 }\right)}\bigg).
\end{align*}
It is then clear that, when $y/t\to{}0$,
the RHS of (\ref{Eq2L1}) would converge to $-\infty$. This is a contradiction since LHS of (\ref{Eq2L1}) is always 0. Therefore, $y(t)/t$ does not converge to $0$ as $t \to \infty$.
Now, let us consider the other case, when $\mu +\frac{ \sigma^2}{2} \geq 0$. In that case, we now rewrite (\ref{gbm:1st:Dy:L1})-(\ref{gbm:1st:Dy:L3}) as follows:
\begin{align}\label{Eq11L1}
	 \frac{\frac{\partial C}{\partial y} }{ \phi({d_{t}})} &=  -S_0 e^{-y} \frac{ N ({d_{t}})}{\phi ({d_{t}})}  
	 + S_0e^{-y-\frac{2y\mu}{\sigma^2}+\mu t}\bigg( \frac{2\mu}{\sigma^2}+1\bigg) \frac{ N(e_t)}{\phi ({d_{t}})} \\
	 \label{Eq11L2}
	 &\quad + S_0 e^{-y} \bigg\{ \bigg(\frac{-2\mu}{\sigma^2}\bigg) \frac{\sigma\sqrt{t}}{y-\mu t + \frac{\sigma^2 t}{2}} \bigg\}
	  +  \rho c \bigg\{ \frac{2}{\sigma\sqrt{t}} \frac{y}{y-\mu t + \frac{\sigma^2 t}{2}} \bigg\} \\
	 \label{Eq11L3}
	 &\quad + S_0 e^{-y}\bigg\{ \big(\frac{-2\mu}{\sigma^2}\big) \mathcal{E}_1 (-f_t)  \bigg\}
	  - \rho c \bigg(\frac{-2\mu}{\sigma^2}+1\bigg)\mathcal{E}_1(-f_t).
\end{align}
 Let us again assume that $\frac{y}{t}\to 0$, as $t \to \infty$. Then, (\ref{Eq11L2}) and (\ref{Eq11L3}) converge to 0, while the RHS of (\ref{Eq11L1})  
converges to $-\infty$ if $y/t\to{}0$ and we again have a contradiction.
\end{proof}

\begin{lemma}
\label{sigma0:uniqueness}
{Suppose that $\rho=\rho(y,t)$ is constant in $y$ and $t$ and let $\mu<0$. Then, there exists $\sigma_0$ such that, for all, $0 < \sigma <\sigma_0$, $ \frac{\partial C}{\partial y}>0$ for all $y > -\mu t$.}
\end{lemma}
\begin{proof}
Let us recall the first derivative in (\ref{gbm:1st:Dy:L1})-(\ref{gbm:1st:Dy:L3}) and the notation $d_{t},e_{t},f_{t}$ introduced in (\ref{TNUP}). Note that when $y> -\mu t$, we can use (\ref{imp.inequ}) to have the following inequality:\begin{align}
\label{sig0:lem1:eq1:L1}
\frac{\partial C}{\partial y} &> S_0 e^{-y}   \phi \left({d_{t}} \right) 
		 \left[\frac{1}{-e_t}-\frac{1}{-e_t^3}
+\left( \frac{-2\mu}{\sigma^2}\right) \left(\frac{1}{-f_t}-\frac{1}{-f_t^3}-\frac{1}{-e_t} \right)\right]\\
\label{sig0:lem1:eq1:L2}
		&+\rho (r+f) \phi \left({d_{t}} \right) \frac{2y}{\sigma\sqrt{t} (y - \mu t + \frac{\sigma^2}{2}t )}-S_0 e^{-y}  N \left({d_{t}} \right).
\end{align}
To prove this lemma, we will first show that when $y>-\mu t$, there exists $\sigma_1$ such that for $\sigma < \sigma_1$,  the expression in (\ref{sig0:lem1:eq1:L2}) is positive, and there exists $\sigma_2$ such that for $\sigma < \sigma_2$, the expression in (\ref{sig0:lem1:eq1:L1}) is also positive. Then, when $\sigma <\min \{\sigma_1, \sigma_2\}$, $\partial_y C>0 $ for all $y>-\mu t$.
Let's assume that $\sigma < \sqrt{-2\mu}$.  
After some simplification, the inequality (\ref{sig0:lem1:eq1:L2})$>0$ can be rewritten as the following:
\begin{equation*}
\label{sig0:lem1:eq2}
	\frac{1}{\sigma \sqrt{t}} > \frac{S_0 e^{-y}}{\rho(r+f)} \frac{N \left({d_{t}} \right)}{\phi \left({d_{t}}\right)}\frac{y-\mu t +\frac{\sigma^2 t}{2} }{y},
\end{equation*}
which is true when  $y > -\mu t$, $\frac{\sigma^2 t}{2} < -\mu t $, and
$\sigma < \rho(r+f) e^{-\mu t}\phi \left(\sqrt{-\mu t/2}\right)/3 S_0 N\left(\sqrt{-\mu t/2}\right)$.
Therefore, the inequality (\ref{sig0:lem1:eq1:L2})$>0$ is true for all $\sigma<\sigma_1$, where
\[
\sigma_1 = \min \left(\sqrt{-2\mu}, \frac{\rho(r+f) e^{-\mu t}}{3S_0} \frac{\phi \big(\sqrt{|\mu| t/2}\big)}{N\big(\sqrt{|\mu| t/2}\big)} \right).
\]
Now, there exists $\sigma_2$ such that for $\sigma < \sigma_2 \Rightarrow$ (\ref{sig0:lem1:eq1:L1})$>0$. Note that (\ref{sig0:lem1:eq1:L1})$>0$ can be rewritten as the following:
\begin{align*}
 &S_0 e^{-y} \left|   \left[\frac{1}{-e_t}-\frac{1}{-e_t^3}
+\left( \frac{-2\mu}{\sigma^2}\right) \left(\frac{1}{-f_t}-\frac{1}{-f_t^3}-\frac{1}{-e_t} \right)\right] \right| 
		   < (r+f)  \frac{y}{\sigma\sqrt{t} (y - \mu t + \frac{\sigma^2}{2}t )} ,
\end{align*}
which is true for every $\sigma<\sigma_2$ where 
\[
	\sigma_2 = \min \left( \sqrt{-2\mu }, \frac{\rho(r+f)e^{-\mu t}}{3 S_0}\frac{\mu^2 t^2}{(-\mu -\mu t + 9/4)\sqrt{-2\mu t} } \right).
\]
Let $\sigma_0 = \min (\sigma_1, \sigma_2)$. Then for $\sigma < \sigma_0$, $\frac{\partial C}{\partial y} > 0$ for all $y > -\mu t$.
\end{proof}

\begin{lemma}
\label{sig0:locmin}
Under the conditions of Theorem~\ref{sig0:delta_approx}, $y^*(\sigma)\nearrow -\mu t$ as $\sigma \to 0$.
\end{lemma}

\begin{proof}
Let us recall the expression of $\partial_y C$ given in (\ref{gbm:1st:Dy:L1})-(\ref{gbm:1st:Dy:L3}),  the notation in (\ref{TNFGBM}), and the notation $d_{t}, e_{t}, f_{t}$ introduced in (\ref{TNUP}), and $\mathcal{E}_1(z) = N(-z)/\phi(z) -1/z$. Note that the two terms in (\ref{gbm:1st:Dy:L1}), the first term in (\ref{gbm:1st:Dy:L2}), and the last term in (\ref{gbm:1st:Dy:L3}) can be approximated by $\phi(z)/z$ when $\sigma\to{}0$, regardless of the value of $y$. 
{The second term in (\ref{gbm:1st:Dy:L2}) can be approximated well by such an expression provided that $y > -\mu t$. However, from Lemma~\ref{sigma0:uniqueness}, we know that under small $\sigma$, $\partial_y C > 0$ when $y>-\mu t$, which implies that $y^*(\sigma) \leq -\mu t$.} So, for now we analyze only those terms that can be approximated well by $\phi(z)/z$. In that case, we can write:
\begin{align}
\nonumber
	\frac{\partial C}{\partial y} &=  -S_0 e^{-y}  N \left({d_{t}} \right)+\phi \left({d_{t}}\right) S_0 e^{-y} \bigg\{ \frac{2\mu}{\sigma^2} \frac{1}{f_t}-\frac{2\alpha_{+}}{\sigma^{2}}\frac{1}{e_t} \bigg\}+ \phi \left({d_{t}}\right) \frac{2\rho c}{\sigma\sqrt{t}} \bigg\{ 1-\frac{\alpha \sqrt{t}}{f_t} \bigg\} \\
	 \label{Eq1L4}
	 &\quad +\phi \left({d_{t}} \right) S_0 e^{-y}\bigg\{ \big(\frac{-2\mu}{\sigma^2}\big) \mathcal{E}_1 \left(-f_t\right) + \big( \frac{2\mu }{\sigma^2}+1 \big) \mathcal{E}_1 \left( -e_t\right) \bigg\}\\
	 \label{Eq1L5}
	&\quad -\phi \left({d_{t}} \right) \rho(r+f) \left(\frac{-2\mu}{\sigma^2}+1\right)\mathcal{E}_1\left(-f_t\right)
\end{align}
Now, we show that the terms involving $\mathcal{E}_{1}$ converge to $0$. To this end, recall that $\mathcal{E}_{1}(z)\sim -1/{z^{3}}.$
Also, $\sigma \to 0$, without loss of generality we can assume that $\frac{\sigma^2 t}{2} < {(-\mu t)}/{2}$. Since $y>0$, $y - \mu t -{\sigma^2 t}/{2} > {-\mu t}/{2} $. Then, all the terms in (\ref{Eq1L4})-(\ref{Eq1L5}) converge to $0$ at the order of $\mathcal{O}(\sigma)$.
Therefore, when $\sigma \to 0$, 
\begin{align}
\label{C:sig0:L1}
	\frac{\partial C}{\partial y }\to & -S_0 e^{-y}  N \left({d_{t}} \right)+\sigma\sqrt{t}\phi \left({d_{t}} \right)S_{0}e^{-y} \frac{y+\mu t + \frac{\sigma^2 t}{2}}{(y-\mu t - \frac{\sigma^2 t}{2})(y-\mu t + \frac{\sigma^2 t}{2})}  \\
\label{C:sig0:L3}
	 &\quad +\phi \left({d_{t}} \right) (r+f)\frac{2}{\sigma\sqrt{t}} \frac{y}{y-\mu t + \frac{\sigma^2 t}{2}}.
\end{align}
If we substitute $y$ with $-\mu t$ to the RHS of (\ref{C:sig0:L1})-(\ref{C:sig0:L3}), it converges to $\infty$ as  $\sigma\to{}0$,  because (\ref{C:sig0:L3})$\to{}\infty$.
This shows that there exists a local minimum in $(0,-\mu t)$, since $\lim_{\sigma\to{}0}\partial_{y}C(0,t)=-S_{0}$. Let us write this minimum as $y^{*}(\sigma)$. 

Suppose that $y^{*}(\sigma)+\mu t \to \delta\in (\mu t,0)$, as $\sigma\to{}0$. Then, for constant $\delta$, the first term of 
(\ref{C:sig0:L1}) converges to $ -S_0 e^{-\delta+\mu t}$, while the second term in (\ref{C:sig0:L1}) and (\ref{C:sig0:L3}) converges to 0 at the order of $\mathcal{O}(\exp(-1/\sigma^2)/\sigma)$.
Thus, if $\delta$ is fixed constant,
$\frac{\partial C}{\partial y } (y^*(\sigma) ) \to -S_0 e^{-\delta+\mu t}$,
which is a contradiction as it should converge to $0$. 
\end{proof}

\section{Details on the computation of $\rho(x,t)$ and Proofs of Section~\ref{NwSect} }\label{DtlsRho}
\subsection{Computation of $\rho(x,t)$ and Estimation of Parameters}
Let us recall that  
{$\rho(x,t)$} is the probability that {a bid limit}  order placed at level ${S}_{0}-x$ at time $0$ is executed {before time $t$} during the first time period at which the {best bid} price is at the level ${S}_{0}-x$, conditioning on the latter event to happen. As mentioned in Remark \ref{CmptRho} (find notation therein), a reasonable formula for $\rho(x,t)$ is given by
\begin{equation}\label{ModelRhoFnl}
		\rho(x,t):=\sum_{i=1}^{\infty}\sum_{j=0}^{Q^{b}_{x}(0)}{ f^{a}(i)}\int_{0}^{t}f_{\tau}(s|0<\tau<t)P({ {N^{b,x}_{s}}}=j)\alpha_{t-s}(i,Q^{b}_{x}(0)-j+1){ ds.}
\end{equation}
Our goal in this subsection is to specify all the elements for the computation of $\rho(x,t)$. Whenever possible, we also estimate the underlying parameters using real  data, which, in this work, consists of level I LOB data for MSFT from April 17th to April 28th 2015 (8 days).

For the density $f_{\tau}(s|0<\tau<t)$ in (\ref{ModelRhoFnl}), we consider the Bachelier and the Black-Scholes Model. To distinguished both models, we denote $\tau_{_{BM}}$ (respectively, $\tau_{_{GBM}}$) the hitting time of the level $S_{0}-x$ for a BM (respectively, GBM) starting at $S_{0}$. Then, from (\ref{JntCndSY}), we can deduce that
\begin{align}
	\label{ftaux}
	f_{\tau_{_{BM}}}(s|0<\tau_{_{BM}}<t) &= \frac{\frac{1}{\sigma\sqrt{s^3}}x \phi\left(\alpha_{t}(x)s\right)}{\left(N\left(-\alpha_t(x)\right)+e^{\frac{-2x \mu }{\sigma^2} }N\left(\beta_t(x)\right)\right)},\\
\label{ftaux2}
	f_{\tau_{_{GBM}}}(s|0<\tau_{_{GBM}}<t)&=\frac{\frac{y}{\sigma s\sqrt{s}}\phi \left(\frac{y+\alpha_{-}s }{\sigma\sqrt{s}} \right)}
	{ \left(N \left(\frac{-y-\alpha_{-}t }{\sigma\sqrt{t}} \right) + e^{-\frac{2y\mu}{\sigma^2}+y} N \left(\frac{-y+\alpha_{-}t }{\sigma\sqrt{t}} \right)\right) },
\end{align}
where we used the notations given in Eqs.~(\ref{TNFBM}) and (\ref{TNFGBM}).

Recall that  $f^{a}(i)$ is the distribution of the best ask queue after the best bid price drops and {a} new best ask queue {fills the gap, shrinking the spread to 1 tick.} In Figure~\ref{Fai_distibution}, we show the sample distribution of $f^{a}(i)$ from the MSFT data. 
\begin{figure}
\centering
	\includegraphics[width=0.5\textwidth]{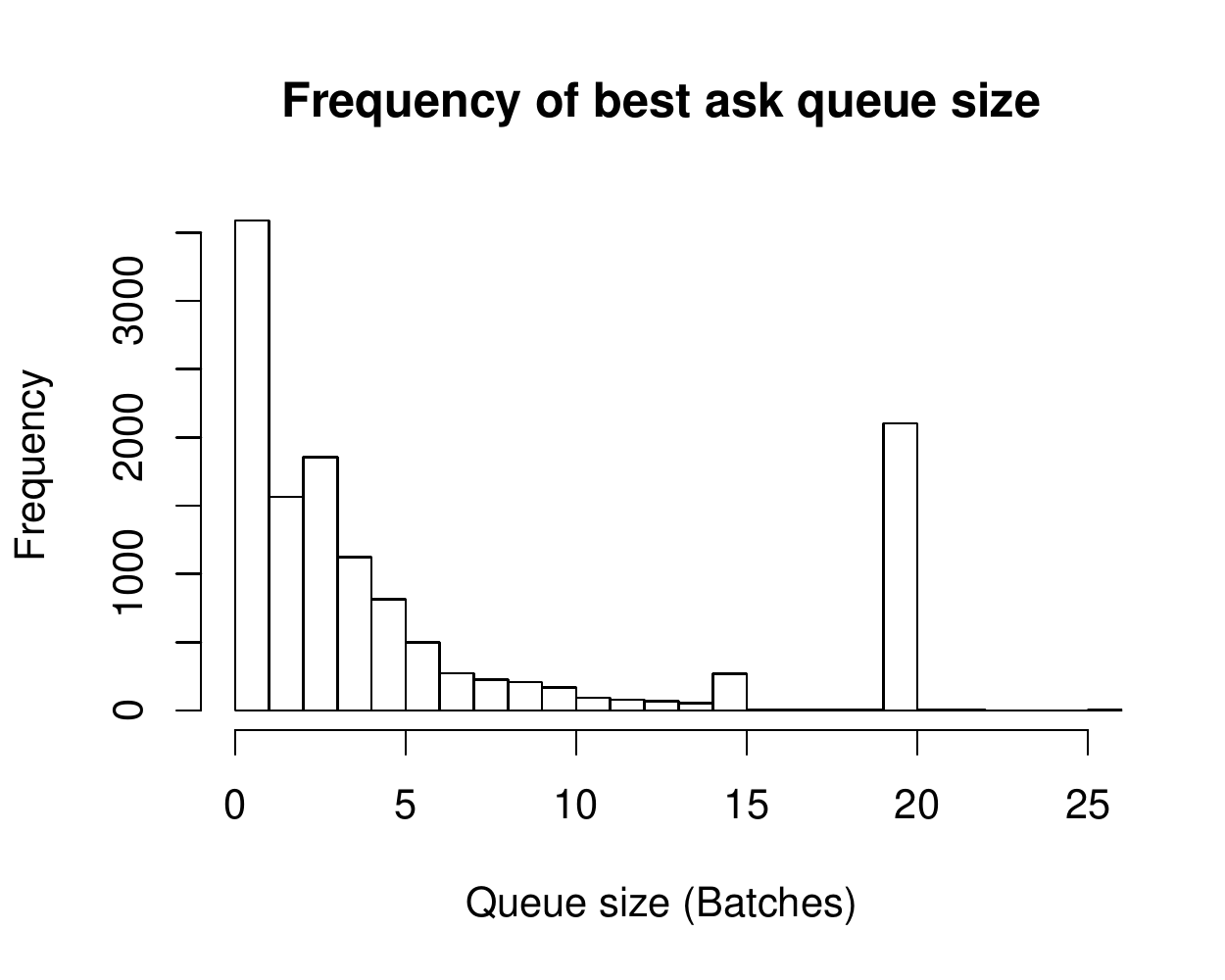}
	\caption{Distribution of best ask queue volume from MSFT data from April 17th to April 28th (8 days). The unit of queue size is a batch (100 stocks).}
\label{Fai_distibution}
\end{figure}

To determine $\alpha_{t}(i,j)$ and $P({ {N^{b,x}_{s}}}=j)$ in (\ref{ModelRhoFnl}), we need to specify the future flow of order, for which, we assume that the arrival, execution, and cancellation of limit orders at any level follow independent Poisson processes with respective intensity rates $\lambda_{\ell}$,  $\mu_{\ell}$, and $\theta_{\ell, {k}}$. Here, $\ell$ is either $a$ or $b$, depending on whether the order is in ask or bid side, and ${k}=1,2,\dots$ are the number of ticks away from the best bid or ask of the opposite side (so, $\theta_{a, 1}$ and $\theta_{b, 1}$ are the cancellation rates at the best ask and bid). For simplicity, we assume that $\theta_{a, {k}}=\theta_{b, {k}}=:\theta_{{k}}$, for ${k}=2,3,\dots$.  The arrival rates $\lambda_{a}$, $\lambda_{b}$, $\theta_{a,1}+\mu_{a}$, and $\theta_{b,1}+\mu_{b}$ are estimated using the MSFT data and are given in Table \ref{Table:Rates}. The chosen values for $\theta_{\ell,{k}}$ (${k}=2,3,\dots$) are borrowed from \cite{cont2010stochastic}.

From the assumptions in the previous paragraph, we have that 
\begin{align}\nonumber
	P({N^{b,k\eps}_{s}}=j) =\frac{e^{-\theta_k s} (\theta_k s)^j}{j!}, \text{ for } 0\leq j<Q^{b}_{k\eps}(0),\quad 
	P({N^{b,k\eps}_{s}}=Q^{b}_{k\eps}(0)) = \sum_{j=Q_{k\eps}(0)}^{\infty} \frac{e^{-\theta_k s} (\theta_{k} s)^j}{j!},
\end{align}
where $Q^{b}_{k\eps}(0)$ is the number of outstanding bid limit orders at the level $S_{0}-k\eps$ at time $0$. Finally, we turn to the computation of $\alpha_{u}(i,\ell)$, the probability that the best bid of a LOB gets depleted before the best ask and before time $u$ when there are $i$ and $\ell$ orders at the best ask and bid bid price, respectively. Note that if $\sigma_a^{i}$ and $\sigma_b^{\ell}$ are respectively the times until the best ask and bid queues get depleted when there are $i$ and $\ell$ shares in the queues at time $0$, then
 \begin{align}\label{alphat}
 	\alpha_{u}(i,\ell)=P(\sigma_b^{\ell}<\sigma_a^{i}, \sigma_b^{\ell} <u)
     	&=\int_0^{u} P(\sigma_b^{\ell} <s) P(\sigma_a^{i} \in ds)+P(\sigma_b^{\ell}<u)P(\sigma_a^{i}>u).
 \end{align}
Under our Poissonian setting, the distributions of $\sigma_a^{i}$ (already computed in \cite{cont2013price}) and $\sigma_b^{\ell}$ are as following:
 \begin{align}
\nonumber 
	g_a^{i}(s):=P(\sigma_a^{i} \in ds)&=\frac{i}{s}\left( \frac{\mu_a+\theta_a}{\lambda_a}\right)^{\frac{i}{2}} I_{i}\left( 2\sqrt{ \lambda_a (\theta_a+\mu_a)}s\right)e^{-s(\lambda_a+\theta_a+\mu_a)} ds,\\
 	\label{sigb}g_b^{\ell}(s):=P(\sigma_b^{\ell}\in ds)&=\frac{s^{\ell-1}e^{-s(\mu_b+\theta_b)}(\mu_b+\theta_b)^{\ell}}{\Gamma(\ell)} ds.
 \end{align}
 So, we can easily compute (\ref{alphat}) by numerical integration. 
 Note that $\alpha_{u}(i,\ell)$ is the same as $\rho(\eps,u)$, which is a proxy of $\rho(0^{+},t)$. Also, $\rho(0^+)$, which was defined as the limiting value of $\alpha_{u}(i,\ell)$ as $u\to\infty$, can be computed as
 \begin{equation}\label{LmtRho0}
{ \rho(0^+):= \lim\limits_{u\to\infty}\alpha_{u}(i,\ell) =\int_0^{\infty} P(\sigma_b^{\ell} <s) P(\sigma_a^{i} \in ds).} 
\end{equation}
From (\ref{sigb}), we also have that
 \[
 g_b^1(s)={e^{-s(\mu_b+\theta_b)}(\mu_b+\theta_b)}, \quad g_b^1(0)=\mu_b+\theta_b.
 \]

\begin{table}
	\begin{center}
	\begin{tabular}{|c| c |c|}
	\hline
	$\mu_a+\theta_{a,1}$ & Depletion rate of best ask queue & 19.32\\
	\hline
	$\lambda_a$ & Addition rate of best ask queue & 21.78\\
	\hline
	$\mu_b+\theta_{b,1}$ & Depletion rate of best bid queue & 18.68\\
	\hline
	$\lambda_b$ & Addition rate of best bid queue & 21.98\\
	\hline
\end{tabular}
\end{center}
\caption{Parameters for the addition and depletion rate of best ask and best bid queues from MSFT data. Unit of the rate is in batches, and the time unit is one second.}\label{Table:Rates}
\end{table}

\subsection{Proofs of Section~\ref{NwSect}}

\begin{proof}[\textbf{Proof of Proposition~\ref{Rhoxt:Diff:xINfty}}]
We show the proof of the result for the Bachelier model. The proof for the Black-Scholes framework can be done similarly.
Let us start by noting that, since we are assuming that $Q^{b}_{x}(0)=0$ for large enough $x$, $\rho(x,t)$ takes the form:
\begin{equation}\label{DfnAlphaQ0}
		\rho(x,t)=\sum_{i=1}^{\infty}f_{a}(i)\int_{0}^{t}f_{\tau}(s|0<\tau<t)\alpha_{t-s}(i,1){ ds}.
\end{equation}
By the continuity of $\partial_x f_{\tau}(s|0<\tau<t)$, we can use the Leibniz's rule for differentiation under the integral sign to get
\begin{align*}
\partial_x \rho(x,t) &=\sum_{i=1}^{\infty}f_{a}(i)\int_{0}^{t}\partial_x f_{\tau}(s|0<\tau<t)\alpha_{t-s}(i,1){ ds}\\
&= \sum_{i=1}^{\infty}f_{a}(i)\int_{0}^{t}
\frac{\frac{1}{\sigma\sqrt{s^3}}\phi\left(\frac{x+\mu s}{\sigma\sqrt{s}}\right)\phi\left(\frac{x+\mu t}{\sigma\sqrt{t}}\right)\frac{2x}{\sigma^2 \sqrt{t}}\left(1-\frac{t}{s}\right) }{P(\tau_x<t)^2}
\alpha_{t-s}(i,1){ ds},
 \end{align*}
 where we used the expression in (\ref{ftaux}) and the notation in (\ref{TNFBM}). Since the integrand above is negative, we conclude that $\partial_x \rho(x,t)\leq{}0$.
\end{proof}

\begin{proof}[\textbf{Proof of Proposition ~\ref{Rhoxtx2:xInfty:General}} ]
As in the proof of Proposition~\ref{Rhoxt:Diff:xINfty}, $\rho(x,t)$ takes the form (\ref{DfnAlphaQ0}) for large enough $x$.
Next, using the notation and formula in (\ref{alphat}), $\partial_u \alpha_{u}(i,\ell)=P(\sigma_b^{\ell} <u) g_a^{i}(u) +g_b^{\ell}(u)P(\sigma_a^{i}>u)- P(\sigma_b^{\ell}<u)g_a^i(u)$ and $\partial_u\alpha_{u}(i,1)|_{u=0}=g_b^1(0)$. Therfore, for any $\eps>0$,  there exists $\delta>0$  such that, for $s \in (t-\delta,t)$, $\left|\partial_u \alpha_{u}(i,1)|_{u=s}-g_b^1(0)\right|<\eps$. By the fact that $\alpha_{0}(i,1)=0$, there exists some $\theta\in[0,t]$ such that $\alpha_{t-s}(i,1)=\left(\partial_u \alpha_{u}(i,1)|_{u=\theta}\right)(t-s)$. Then, $I:=\int_{0}^{t}f_{\tau}(s|0<\tau<t)\alpha_{t-s}(i,1){ds}$ is such that
\begin{align}\nonumber
	I&
	=\int_{0}^{t}f_{\tau}(s|0<\tau<t)\left(\partial_u \alpha_{u}(i,1)|_{u=\theta}\right)(t-s){ ds},\\
	&\geq (g_b^1(0)-\eps) \int_{0}^{t}f_{\tau}(s|0<\tau<t)(t-s)ds\\
	&\quad 
+\int_{0}^{t-\delta}f_{\tau}(s|0<\tau<t)\left(\partial_u \alpha_{u}(i,1)|_{u=\theta}-g_b^1(0)\right) (t-s)ds.\label{SectionD:Ineq1}
\end{align}
Note that by (\ref{ftaux}), the final integral in (\ref{SectionD:Ineq1}) converges to 0 at the order of $\mathcal{O}(x^2 e^{-x^2 \delta/2\sigma^2t(t-\delta)})$. Therefore, it is sufficient to show the order of $\liminf_{x\to\infty}x^2\int_{0}^{t}f_{\tau}(s|0<\tau<t)(t-s)ds$ to complete the proof. By (\ref{imp.inequ}), $J(x,t)=\int_{0}^{t}f_{\tau}(s|0<\tau<t)(t-s)ds$ is such that
\begin{align}
\nonumber
	x^2J(x,t)&
=x^2t-\frac{x^3\left(  N\left(-\alpha_t\right) -e^{\frac{-2\mu x}{\sigma^2}}  N\left(\beta_t \right)\right) }{(-\mu)\left({N\left(-\alpha_t\right)+e^{\frac{-2x \mu }{\sigma^2} }N\left(\beta_t\right)}\right)}\\
&= \frac{x^2\phi\left(\alpha_t\right)\left(\frac{t}{\beta_t^3}-\frac{t}{\alpha_t^3}+\frac{x}{(-\mu)\beta_t^3} +\frac{x}{(-\mu)\alpha_t^3} +\mathcal{E}(x)\right)}{\left({N\left(-\alpha_t\right)+e^{\frac{-2x \mu }{\sigma^2} }N\left(\beta_t\right)}\right)},
\label{SectionD:Prop2:eq}
\end{align}
where $\mathcal{E}(x)$ converges to 0 at the order of $\mathcal{O}({1/x^4})$, and the denominator in (\ref{SectionD:Prop2:eq}) converges to $2\phi(\alpha_t)x/(x+\mu t)(x-\mu t)$. Therefore, 
\[
\lim_{x\to\infty}	x^2J(x,t)=
\lim_{x\to\infty}\frac{x^2\phi\left(\alpha_t\right)\sigma^3 t\sqrt{t}\left( \frac{4(-\mu t) x^3}{(x+\mu t)^3 (x-\mu t)^3}\right) }{2\phi(\alpha_t)x/(x+\mu t)(x-\mu t)} =2\sigma^2 t^2,
\]
and, finally get that $\liminf_{x\to\infty} x^2 \rho(x,t)	\geq g_b^1(0) 2\sigma^2 t^2$. In the geometric Brownian motion model, the only difference between $\rho(x,t)$ and $\tilde{\rho}(y,t)$ is that $y$ is used instead of $x$ and $\mu-\sigma^2/2$ is used instead of $\mu$. Following same steps, we will get the same lower bound for $\liminf_{y\to\infty} y^2 \tilde{\rho}(y,t)	$.	
\end{proof}


\begin{thebibliography}{14}
\providecommand{\natexlab}[1]{#1}
\providecommand{\url}[1]{\texttt{#1}}
\expandafter\ifx\csname urlstyle\endcsname\relax
  \providecommand{\doi}[1]{doi: #1}\else
  \providecommand{\doi}{doi: \begingroup \urlstyle{rm}\Url}\fi

\bibitem[Abergel and Jedidi(2013)]{abergel2013mathematical}
F.~Abergel and A.~Jedidi.
\newblock A mathematical approach to order book modeling.
\newblock \emph{International Journal of Theoretical and Applied Finance},
  16\penalty0 (05):\penalty0 1350025, 2013.

\bibitem[Alfonsi et~al.(2010)Alfonsi, Fruth, and Schied]{alfonsi2010}
A.~Alfonsi, A.~Fruth, and A.~Schied.
\newblock Optimal execution strategies in limit order books with general shape
  functions.
\newblock \emph{Quantitative Finance}, 10\penalty0 (2):\penalty0 143--157,
  2010.

\bibitem[Cartea and Jaimungal(2014)]{cartea2014modelling}
A.~Cartea and S.~Jaimungal.
\newblock Buy low, sell high: a high frequency trading perspective.
\newblock \emph{SIAM J. Financ. Math.}, 5\penalty0 (1):\penalty0 415--444,
  2014.

\bibitem[Ch{\'a}vez-Casillas and Figueroa-L{\'o}pez(2014)]{chavez2014one}
J.~A. Ch{\'a}vez-Casillas and J.~E. Figueroa-L{\'o}pez.
\newblock One-level limit order books with sparsity and memory.
\newblock \emph{arXiv preprint arXiv:1407.5684}, 2014.

\bibitem[Cont and De~Larrard(2013)]{cont2013price}
R.~Cont and A.~De~Larrard.
\newblock Price dynamics in a markovian limit order market.
\newblock \emph{SIAM Journal on Financial Mathematics}, 4\penalty0
  (1):\penalty0 1--25, 2013.

\bibitem[Cont and Kukanov(2013)]{ContKukanov2013}
R.~Cont and A.~Kukanov.
\newblock Optimal order placement in limit order markets.
\newblock \emph{Available at ssrn 2155218}, 2013.

\bibitem[Cont et~al.(2010)Cont, Stoikov, and Talreja]{cont2010stochastic}
R.~Cont, S.~Stoikov, and R.~Talreja.
\newblock A stochastic model for order book dynamics.
\newblock \emph{Operations research}, 58\penalty0 (3):\penalty0 549--563, 2010.

\bibitem[Gruber and Schweizer(2006)]{gruber2006diffusion}
U.~Gruber and M.~Schweizer.
\newblock A diffusion limit for generalized correlated random walks.
\newblock \emph{Journal of applied probability}, 43\penalty0 (1):\penalty0
  60--73, 2006.

\bibitem[Guilbaud and Pham(2013)]{GuilbaudPham2013}
F.~Guilbaud and H.~Pham.
\newblock Optimal high-frequency trading with limit and market orders.
\newblock \emph{Quantitative Finance}, 13\penalty0 (1):\penalty0 79--94, 2013.

\bibitem[Guo et~al.(2016)Guo, De~Larrard, and Ruan]{guo2013optimal}
X.~Guo, A.~De~Larrard, and Z.~Ruan.
\newblock Optimal placement in a limit order book.
\newblock \emph{Mathematics and Financial Economics}, 2016.

\bibitem[Jacquier and Liu(2017)]{jacquierLiu2017}
A.~Jacquier and H.~Liu.
\newblock Optimal liquidation in a level-i limit order book for large-tick
  stocks.
\newblock \emph{Preprint. Available at arXiv:1701.01327 [q-fin.TR]}, 2017.

\bibitem[Jeanblanc et~al.(2009)Jeanblanc, Yor, and
  Chesney]{jeanblanc2009mathematical}
M.~Jeanblanc, M.~Yor, and M.~Chesney.
\newblock \emph{Mathematical methods for financial markets}.
\newblock Springer, 2009.

\bibitem[Maglaras et~al.(2015)Maglaras, C., and H.]{Manglaras2015}
C.~Maglaras, M.~C., and Z.~H.
\newblock Optimal execution in a limit order book and an associated
  microstructure market impact model.
\newblock \emph{Preprint available at ssrn 2610808}, 2015.

\bibitem[Renshaw and Henderson(1981)]{renshaw1981correlated}
E.~Renshaw and R.~Henderson.
\newblock The correlated random walk.
\newblock \emph{Journal of Applied Probability}, pages 403--414, 1981.

\end{thebibliography}

\end{document}